\documentclass{amsart}
\usepackage{fullpage,graphicx,subfigure,mathpazo,color}
\usepackage{amsmath,amscd,tikz,mathrsfs}
\usepackage[normalem]{ulem}
\usepackage{setspace}
\usepackage{lipsum}
\usepackage{amsfonts}
\usepackage{graphicx}
\usepackage{epstopdf}
\usepackage{algorithmic}

\usepackage{cleveref}

\usepackage{amssymb}
\usepackage{amsfonts}
\usepackage{subfigure}

\usepackage{extpfeil,extarrows}
\usepackage{setspace}
\usepackage[titletoc]{appendix}
\usetikzlibrary{arrows,calc}
\usetikzlibrary{positioning}
\usetikzlibrary{decorations.markings}

\newcommand{\ii}{\mathrm{i}}
\newcommand{\ee}{\mathrm{e}}
\newcommand{\dd}{\mathrm{d}}

\newcommand{\sn}{\mathrm{sn}}
\newcommand{\dn}{\mathrm{dn}}
\newcommand{\cn}{\mathrm{cn}}

\newcommand{\tn}{\mathrm{tn}}

\newcommand{\cd}{\mathrm{cd}}
\newcommand{\nd}{\mathrm{nd}}

\newcommand{\scd}{\mathrm{scd}}

\newtheorem{theorem}{Theorem}
\newtheorem{lemma}{Lemma}
\newtheorem{prop}{Proposition}

\newtheorem{definition}{Definition}

\newtheorem{remark}{Remark}

\usepackage{graphicx}
\usepackage{titlesec}

\titleformat{\section}{\centering\LARGE\bfseries}{\thesection}{1em}{}
\titleformat{\subsection}{\Large\bfseries}{\thesubsection}{1em}{}
\begin{document}

\title{Stability of elliptic function solutions for the focusing modified KdV equation \thanks{This project is supported by the National Natural Science Foundation of China (Grant No. 11771151), the Guangzhou Science and Technology Program of China (Grant No. 201904010362)} }

\author{Liming Ling}

\address{School of Mathematics, South China University of Technology, Guangzhou, China, 510641}
\email{linglm@scut.edu.cn}
\author{Xuan Sun}
\address{School of Mathematics, South China University of Technology, Guangzhou, China, 510641}
\email{masunx@mail.scut.edu.cn}

\begin{abstract}
We study the spectral and orbital stability of elliptic function solutions for the focusing modified Korteweg-de Vries (mKdV) equation and construct the corresponding breather solutions to exhibit the stable or unstable dynamic behavior. The elliptic function solutions of the mKdV equation and related fundamental solutions of the Lax pair are exactly represented by theta functions. Based on the `modified squared wavefunction' (MSW) method, we construct all linear independent solutions of the linearized mKdV equation and then provide a necessary and sufficient condition of the spectral stability for elliptic function solutions with respect to subharmonic perturbations. In the case of spectrum stability, the orbital stability of elliptic function solutions is established in a suitable Hilbert space. Using Darboux-B\"acklund transformation, we construct breather solutions to exhibit unstable or stable dynamic behavior. Through analyzing the asymptotic behavior, we find that the breather solution under the $\cn$-type solution background is equivalent to the elliptic function solution adding a small perturbation as $t\to\pm\infty$.


\end{abstract}

\date{\today}

\maketitle

\section{Introduction}

In this work, we mainly study the stability of the elliptic function solutions of the focusing modified Korteweg-de Vries (mKdV) equation
\begin{equation}\tag{mKdV}\label{eq:mKdV}
	u_t+6u^2u_x+u_{xxx}=0,
\end{equation}
where $u=u(x,t)$ is a real-valued function with $(x,t)\in \mathbb{R}^2$. The mKdV equation has applications in diverse physical contexts, such as water waves and plasma physics \cite{AblowitzS-81,Cheemaa-2020-study,Schamel-1973-modified}. As we know that the \eqref{eq:mKdV} equation is related to the Korteweg-de Vries (KdV) equation by the Miura transform \cite{MiuraGK-68}, and it can be regarded as the generalization of the KdV equation. It is a well-known completely integrable model admitting the Lax-pair formulation \cite{Lax-68}, the bi-Hamiltonian structure \cite{MaF-96}, and infinite conserved quantities \cite{MiuraGK-68}. In finite-dimensional mechanics, if the system has sufficiently many (half the dimension of the phase space) Poisson commuting and functionally independent conserved quantities, then it is completely integrable. Actually, the \eqref{eq:mKdV} equation admits infinite many independent conserved quantities $\mathcal{H}_i$, $i=0,1,2,\cdots$ \cite{Dickey-2003-soliton}, in which the first three conservation laws are given in the main text (Eq. \eqref{eq:H0}). For the infinite-dimensional integrable system, the Lax representation is a crucial and useful feature. The Lax pair for the \eqref{eq:mKdV} equation admits the following linear system:
\begin{equation}\label{eq:Lax pair} 
		\Phi_x(x,t;\lambda)=\mathbf{U}(\lambda;u)\Phi(x,t;\lambda), \qquad \Phi_t(x,t;\lambda)=\mathbf{V}(\lambda;u)\Phi(x,t;\lambda),
	\end{equation}
	where the spectral parameter $\lambda\in\mathbb{C}\cup\{\infty\}$,
	\begin{equation}\label{eq:U,V}
		\mathbf{U}(\lambda; u)
		=-\ii \lambda \sigma_3+\mathbf{Q}, \quad
		\mathbf{V}(\lambda; u)=4\lambda^2\mathbf{U}(\lambda; u)
		+2\lambda\ii\sigma_3(\mathbf{Q}_x-\mathbf{Q}^2)-(\mathbf{Q}_{xx}-2\mathbf{Q}^3), \quad
		\mathbf{Q}
		=\begin{bmatrix}
			0 & u \\ -u & 0
		\end{bmatrix},
	\end{equation}
	and the matrix $\sigma_3:=\mathrm{diag}(1,-1)$ is the third Pauli matrix. The Lax pair can be derived from the $2\times 2$ AKNS system by the reductions \cite{AblowitzKNS-74}. The compatibility condition of the linear system \eqref{eq:Lax pair}: $\Phi_{tx}(x,t;\lambda)=\Phi_{xt}(x,t;\lambda)$ is equivalent to the zero-curvature equation $\mathbf{U}_t(\lambda;u)-\mathbf{V}_x(\lambda;u)+\left[\textbf{U}(\lambda;u),\textbf{V}(\lambda;u)\right]=0$ with the commutator defined by $\left[\textbf{A},\textbf{B}\right]=\textbf{A}\textbf{B}-\textbf{B}\textbf{A}$, which yields the \eqref{eq:mKdV} equation. Due to the Lax integrability, the \eqref{eq:mKdV} equation can be solved by the inverse scattering transform, which is widely used to solve a large number of equations \cite{AblowitzKNS-74,BilmanM-19,GardnerGKM-67,ZakharovS-72}. The infinite many conservation laws can also be derived by the Lax representation \cite{AblowitzKNS-74}. In addition, the well-posedness of the mKdV equation has been studied by many scholars \cite{CollianderKSTT-03,KenigPV-93}.

\subsection{Review on the stability analysis of the mKdV equation}\label{subsection:review-the-stability-analysis}

	The stability analysis for the solitary or periodic waves is a classic and crucial problem in the study of nonlinear partial differential equations. As early as the 20th century, many scholars were engaged in studying spectral stability \cite{CourantH-53,LevitanS-75,Rowlands-74}. This research has continued to the present. Deconinck and Kutz computed the spectrum of the maximal extension of linear operators using the Floquet-Fourier-Hill method (FFHM) \cite{DeconinckK-06}. The spectral stability analysis for the nonlinear wave equations was given by Yang in the monograph \cite{Yang-10}. The number of negative directions of the second variation of the energy is one of the methods to help us study the spectral stability of nonlinear waves, which had been proved by Kapitula, Kevrekidis, and Sandstede via the Krein signature \cite{KapitulaKS-04}. 
	Some propositions among the operator $\mathcal{L}$, $\mathcal{JL}$ and the eigenvalue $\Omega$ had been proposed by H\v{a}r\v{a}gus and Kapitula \cite{HaragusK-08}, using the Floquet-Bloch decomposition. The aforementioned spectrum analysis method had been utilized to study the nonlinear Schr\"{o}dinger (NLS) equation \cite{CuccagnaPV-05,DeconinckS-17,KapitulaKS-04}. Furthermore, there are also a large number of spectral stability studies on other equations, such as the coupled NLS equation \cite{Pelinovsky-05,PelinovskyY-05}, the KdV equation \cite{BottmanD-09,PSKP-21}, and so on. 

	An extensive development of the orbital stability theory for the solitary wave solutions has been obtained in the past years by Benjamin, Bona, Grillakis, Shatah, Strauss, and Weinstein \cite{Benjamin-72,Bona-75,BonaSS-87,GrillakisSS-87,GrillakisSS-90,Weinstein-85,Weinstein-86}. 
	Alejo and Mu\~{n}oz \cite{Alejo-2013-nonlinear} analyzed the stability of breather solutions by utilizing a new Lyapunov functional to describe the dynamics of small perturbations. 
	Semenov \cite{Semenov-2022-orbital} studied the orbital stability of the multi-soliton/breather solutions of the mKdV equation by modifying the Lyapunov functional. In the aforementioned literatures, the scholars mainly considered the nonlinear waves with the condition $u(x)\rightarrow 0$ as $x\rightarrow\pm \infty$. Recently, a successful application of this theory has been obtained on the periodic boundary condition in the KdV equation \cite{PavaBS-06}, the critical KdV equation \cite{PavaN-09}, the NLS equation \cite{Pava-07}, the Hirota-Satsuma system \cite{Pava-05}, and so on. Based on the integrable structures of equations, a great deal of work has been performed on the study of the spectral or orbital stability of periodic wave solutions for the NLS equation \cite{ChenPW-20,DeconinckS-17,KapitulaH-07-small,KapitulaH-07-per}, the KdV equation \cite{PavaBS-06,BronskiJK-11}, the mKdV equation \cite{DeconinckK-10,Semenov-2022-orbital}, and so on. 

Then we briefly review the stability analysis for periodic solutions of the mKdV equation, which are closely related to this work. The periodic traveling wave solutions of the defocusing mKdV equation are spectrally stable, which was studied by Deconinck and Nivala \cite{Deconinck-10}.
The NLS equation also has similar results that elliptic function solutions of the defocusing NLS equation are spectrally stable, which was studied by Bottman, Deconinck, and Nivala \cite{BottmanDN-11}. Moreover, we know that the $\cn$-type solutions of the KdV equation are spectrally stable in \cite{DeconinckK-10}. Using the Weierstrass $\wp$ function, the $\zeta$ function and the spectral parameter $\lambda$ of the Lax pair to obtain squared eigenfunctions, Deconinck and Segal \cite{DeconinckS-17} proved that $\dn$-type solutions of the focusing NLS equation are spectrally stable with respect to co-periodic perturbations. Furthermore, the spectral stability of $\cn$-type solutions has also been studied by dividing the modulus $k$ into two different conditions \cite{DeconinckS-17,DeconinckyU-20}. However, there is no systematic work on the spectral stability analysis of the focusing mKdV equation. Therefore, one of the aims of this work is to study the spectral stability of the focusing mKdV equation.

For the studies of the orbital stability, there are many relevant results about the NLS equation. In \cite{BottmanDN-11}, authors studied the orbital stability of elliptic function solutions of the defocusing NLS equation. Based on spectrally stable conditions, Deconincky and Upsal studied that the orbital stability of elliptic function solutions of the focusing NLS equation with respect to subharmonic perturbations obtained in \cite{DeconinckyU-20} by constructing a new Lyapunov function under higher-order conserved quantities. In \cite{Pava-07}, Pava obtained that $\dn$-type solutions were orbitally stable both for the focusing NLS equation and the focusing mKdV equation in the space $H^1([-T,T])$. For the mKdV equation, all periodic traveling wave solutions in the defocusing case were orbitally stable with respect to subharmonic perturbations in the space $H^2([-PT,PT])$, $P\in \mathbb{N}$, which was established by Deconinck and Nivala \cite{Deconinck-10}. Then, it is natural to consider whether there exists a suitable function space such that the elliptic function solutions of the focusing mKdV equation are orbitally stable.

\subsection{Main results} \label{sec:main} 
The \eqref{eq:mKdV} equation has the elliptic function solutions 
\begin{equation}\label{eq:elliptic-solution}
	u(x,t)=k\alpha \cn (\alpha(x-2s_2t),k)\qquad \text{and} \qquad 
	u(x,t)=\alpha \text{dn}(\alpha(x-2s_2t),k),
\end{equation}
where $\cn(\cdot,k)$ and $\dn(\cdot,k)$ denote the Jacobi elliptic functions with elliptic modulus $k=\sqrt{\frac{u_3-u_2}{u_3-u_1}}$; $-2s_2=-(u_1+u_2+u_3)$ is the velocity between time $t$ and space $x$; $\alpha=\sqrt{u_3-u_1}$; and $u_1,u_2,u_3$ are defined in \eqref{eq:u1,u2,u3}. The details of the above solutions can be found in Proposition \ref{prop:solution-u^2}. For convenience, we often omit the modulus $k$ in this work. To examine the traveling wave solutions, we introduce a moving coordinate form
\begin{equation}\label{eq:traveling-wave-transform}
\left(x,t\right)\xRightarrow[t=t]{\xi=x-2s_2 t}\left(\xi,t\right)
\end{equation}
to convert the non-zero velocity $-2s_2$ into a stationary one in \eqref{eq:Rutux}. Then, the \eqref{eq:mKdV} equation turns into
\begin{equation}\label{eq:mKdV1}
	u_{t}-2s_2u_{\xi}+u_{\xi\xi\xi}+6u^2u_{\xi}=0.
\end{equation}
Here, to avoid introducing too many notations, we still use $u(\xi,t)$ to denote the function $u(x,t)$ under the new coordinate $(\xi,t)$. For elliptic function solutions \eqref{eq:elliptic-solution}, we use the notation $u\equiv u(\xi)$ to denote them, i.e., $u(\xi)=k\alpha \cn (\alpha\xi)$ and $u(\xi)=\alpha \dn(\alpha\xi)$.

To study the stability of elliptic function solutions of the \eqref{eq:mKdV} equation, we need to solve the linearized mKdV equation. The squared eigenfunctions can be utilized to construct solutions of the linearized mKdV equation. Thus, combining the algebraic-geometry method with the effective integration method, we obtain elliptic function solutions \eqref{eq:elliptic-solution} of the \eqref{eq:mKdV} equation and the corresponding fundamental matrix solution of the Lax pair simultaneously. By Lax pair \eqref{eq:Lax-pair-1} and the eigenvalue $y$ of the matrix $\mathbf{L}(\xi,t;\lambda)$ in \eqref{eq:Lax pair-1}, the solution $\Phi(\xi,t;\lambda)$ of Lax pair \eqref{eq:Lax-pair-1} could be derived as \eqref{eq:Phi}. We introduce a uniform parameter $z$ in a rectangular region instead of spectral parameter $\lambda\in \mathbb{C}$, which was established in Appendix \ref{appendix:commonality map} regarding the conformal mapping between $\lambda$ and $z$. Therefore, we could avoid the multi-valued function $y$ (refer to equation \eqref{eq:y lambda }) so that the study of the $\dn$-type and $\cn$-type periodic problem becomes simultaneous. Then, we obtain the solution $\Phi(x,t;\lambda)$ in terms of theta functions with respect to the parameter $z$.

\begin{theorem}\label{theorem:solution-Theta form}
	The fundamental solution $\Phi(x,t;\lambda)$ of Lax pair \eqref{eq:Lax pair} can be represented as the theta functions form:
	\begin{equation}\label{eq:Phi-Theta}
		\Phi(x,t;\lambda)=\frac{\alpha\vartheta_2\vartheta_4}{\vartheta_3\vartheta_4(\frac{\alpha\xi}{2K})}
		\begin{bmatrix}
			\frac{\vartheta_1(\frac{\ii(z-l)-\alpha \xi  }{2K})} {\vartheta_4(\frac{\ii(z-l)}{2K})}E_1 
			& \frac{\vartheta_3(\frac{\ii(z+l)+\alpha \xi  }{2K})} {\vartheta_2(\frac{\ii(z+l)}{2K})}E_2 \\ -\frac{\vartheta_3(\frac{\ii(z+l)-\alpha \xi  }{2K})} {\vartheta_2(\frac{\ii(z+l)}{2K})}\ee^{\alpha\xi Z(2\ii l+K)}E_1
			& -\frac{\vartheta_1(\frac{\ii(z-l)+\alpha \xi  }{2K})} {\vartheta_4(\frac{\ii(z-l)}{2K})}\ee^{\alpha\xi Z(2\ii l+K)}E_2
		\end{bmatrix},
	\end{equation}
	where $l=0$ or $\frac{K'}{2}$, $\xi=x-2s_2 t$, theta functions $\vartheta_i(z)$s and functions $E_1=E_1(\xi,t;z)$, $E_2=E_2(\xi,t;z)$ are defined in \eqref{eq:theta1234} and \eqref{eq:E1-E2} respectively, and $K=K(k), K'=K(k')$ are the complete elliptic integrals in \eqref{eq:J-K-E-int}.
\end{theorem}

The eigenvalue of the linearized mKdV equation \cite{KapitulaD-15} shows that $\cn$-type solutions of the focusing mKdV equation are not spectrally stable with respect to any perturbations. In such a case, we want to consider whether suitable perturbations exist such that the $\cn$-type solutions under these perturbations are spectrally stable. To study the spectral stability of elliptic function solutions, we introduce perturbations of the stationary solution
	\begin{equation}\label{eq:perturbation}
		v(\xi,t)=u(\xi)+\epsilon w(\xi,t)+\mathcal{O}(\epsilon^2),
	\end{equation}
	where $\epsilon$ is a small parameter and $w(\xi,t)$ is a real-valued function of $(\xi,t)\in \mathbb{R}^2$. Plugging \eqref{eq:perturbation} into \eqref{eq:mKdV1}
	and considering the first-order term of $\epsilon$, we obtain the linearized equation
	\begin{equation}\label{eq:linearized-mkdv}
		\partial_t w =-\partial_{\xi}^3w +2s_2\partial_{\xi}w -6u^2\partial_{\xi}w -12w u\partial_{\xi}u,
	\end{equation}
	where $u\equiv u(\xi)$ denotes the elliptic function solution \eqref{eq:elliptic-solution} and $w\equiv w(\xi,t)$.
	Since equation \eqref{eq:linearized-mkdv} is autonomous in time, we can decompose $w(\xi,t)$ into the following form
	\begin{equation}\label{eq:w}
		w(\xi,t)=W(\xi)\exp(\Omega t)+W^*(\xi)\exp(\Omega^* t),
	\end{equation}
	by separating variables.
	Then, we obtain the linearized spectral problem of equation \eqref{eq:linearized-mkdv}:
	\begin{equation}\label{eq:spectral}
		\partial_\xi  (-\partial_\xi  ^2+2s_2-6u^2)W=\mathcal{JL}W=\Omega W, \qquad W(\xi) \in C_b^0(\mathbb{R}),
	\end{equation} 
	where $\mathcal{J}=\partial_\xi,\mathcal{L}=-\partial_\xi  ^2+2s_2-6u^2$, $\Omega\in \mathbb{C}$, and $C_b^0(\mathbb{R})$ denotes the space of bounded continuous functions on the real line. The spectrum is defined as 
		\begin{equation}\label{eq:spectrum}
			\sigma(\mathcal{JL}):=\{\Omega\in \mathbb{C}| W(\xi)\in C^0_b(\mathbb{R})\}.
		\end{equation} 
	Due to the Hamiltonian structure of the spectrum \cite{HaragusK-08}, an elliptic function solution $u$ is spectrally stable with respect to perturbations in $C_b^0(\mathbb{R})$ if $\sigma(\mathcal{JL})\subset \ii \mathbb{R}$. Then, the definition of spectral stability is given as follows: 
\begin{definition}\label{define:spect-stable}
		An elliptic function solution $u(\xi)$ is spectrally stable to perturbations $w(\xi,t)$ in $C_b^0(\mathbb{R})$, where $w(\xi,t)=W(\xi)\exp(\Omega t)+W^*(\xi)\exp(\Omega^* t)$, if $\Omega\in \ii \mathbb{R}$. In brief, the stability spectrum is defined as $\sigma(\mathcal{JL})\subset \ii \mathbb{R}$, where $\sigma(\mathcal{JL})$ is defined in \eqref{eq:spectrum}.
	\end{definition}

Based on the MSW method, we get the squared eigenfunction $W(\xi)$, which could be used to gain all solutions of equation \eqref{eq:spectral} in Lemma \ref{lemma:W}. As Deconinck and Kapitula pointed out in \cite{DeconinckK-10}, the spectrum of the focusing mKdV equation is no longer confined to the real axis, which makes the detailed analysis of the bounded eigenfunctions more difficult. To overcome this difficulty, we use theta functions to express the squared eigenfunction $W(\xi)$, which converts the problem of analyzing bounded functions into studying the Zeta function. For the stability analysis, we just consider the bounded function $W(\xi)$ which implies that the real part of the exponent of the function $W(\xi)$ is zero, i.e., $z$ must satisfy \eqref{eq:Re}. Combining \eqref{eq:I-1-2} with \eqref{eq:Phi-Theta}, this relationship on $z$ is equivalent to
\begin{equation}\label{eq:set Q}
	Q:=\left\{  z\in \mathbb{C}\left| {\Re}\left(I (z)\right)=0 , z\in S \right. \right\},
\end{equation}
where $S$ is defined by \eqref{eq:set S} and $I(z)$ is given by \eqref{eq:I}. Then, we get the consequence for the spectral stability.

\begin{theorem}\label{theorem:spec-dn}
	The $\dn$-type solutions of the mKdV equation \eqref{eq:mKdV1} are spectrally stable.
\end{theorem}

Since the relationship between spectral parameter $\lambda$ of the Lax pair and eigenvalue $\Omega$ in the linearized spectral problem \eqref{eq:spectral} is different from the one in the focusing NLS equation, we get the following distinct stability criterion. For the $\dn$-type solutions of the mKdV equation, the square of the eigenvalue could be represented by a cubic polynomial of the variable $\lambda^2$:
\begin{equation}\Omega^2=-64\lambda^2(\lambda^2-\lambda_1^2)(\lambda^2-\lambda_2^2). \end{equation}
Thus, when $0<{\Im}(\lambda_1)<{\Im}(\lambda)<{\Im}(\lambda_2)$ with $\lambda\in \ii \mathbb{R}$, it follows $\Omega \in \ii \mathbb{R}$. For the NLS equation, in view of \cite{DeconinckyU-20}, the relationship between $\Omega$ and $\lambda$ is 
\begin{equation}\Omega^2=-(\lambda^2-\lambda_1^2)(\lambda^2-\lambda_2^2),
\end{equation}
with $\lambda_1,\lambda_2\in \ii \mathbb{R}$ and $0<{\Im}(\lambda_1)<{\Im}(\lambda_2)$, which implies $\Omega \in \mathbb{R}$, ${\Im}(\lambda_1)<{\Im}(\lambda)<{\Im}(\lambda_2),\lambda\in \ii \mathbb{R}$. Therefore, we can conclude that the $\dn$-type solutions of the mKdV equation are spectrally stable but unstable for the NLS equation.

For the $\cn$-type solutions, we mainly consider the spectral stability with respect to the subharmonic perturbations. The value of modulus $k$ divides the spectral problem into two different types, proved in Proposition \ref{prop: zr0 zi0}. One is that the spectral curve of $z$ intersects with the real axis, and the other is that the spectral curve of $z$ intersects with the imaginary axis. Especially, we use Figure \ref{fig: sub cn} and Figure \ref{fig:cn sub point} to illustrate the above two conditions. Based on the different requirements of the spectral curve, we get the following theorem of the spectral stability.
\begin{theorem}\label{theorem:subharmonic-1-}
	The spectral stability of the cnoidal wave solutions for the \eqref{eq:mKdV} equation could be divided into the following two categories:
	\begin{itemize}
		\item If $\frac{2E(k)}{K(k)}\geq 1$, i.e. $ k\le \hat{k} \approx 0.9089 $, the $\cn$-type solutions are spectrally stable with respect to perturbations of period $2PT$, where $P\le \frac{\pi}{\pi+ M(z_c)}$ and $z_c$ satisfies the condition in Proposition \ref{prop: zr0 zi0}. 
		\item If $\frac{2E(k)}{K(k)}<1$, the $\cn$-type solutions are co-periodic subharmonic stable and have no other subharmonic perturbations.
	\end{itemize}
\end{theorem}

Based on the results of spectral stability, we further study the orbital stability of the above elliptic function solutions in a suitable function space.

\begin{definition}\label{define:orbital stable}
	The elliptic function solution $u(\xi)$ of the mKdV equation is orbitally stable with respect to perturbations in a Hilbert space $X$ if for any solution $v(\xi,t)$ of the mKdV equation and any given $\epsilon>0$, there exists $\delta>0$ satisfying
	\begin{equation} \| v(\xi,0)- T(\gamma(0))u(\xi) \|_{X}\le \delta,\end{equation}
	such that
	\begin{equation} \max_{t\in \mathbb{R}} \inf_{\gamma\in \mathbb{R}} \| v(\xi,t)-T(\gamma(t))u(\xi,t ) \|_{X} \le \epsilon, 
	\end{equation}
	where $\|\cdot \|$ denotes the norm obtained through $\left\langle\cdot, \cdot\right\rangle $ in the space $X$ and the operator $T(\gamma(t))$ is defined here as
	\begin{equation}\label{eq:defin-operator}
		T(\gamma(t))u(\xi)\equiv u(\xi+\gamma(t)).
	\end{equation}
\end{definition}
In this paper, we mainly consider two Hilbert spaces $H^1([-PT,PT])$ and $H^2([-PT,PT])$. For any conserved quantities $\mathcal{H}_i$ in the mKdV hierarchy, the corresponding operator $\mathcal{L}_i$ and Krein signature $\mathcal{K}_i(z)$ are defined in Definition \ref{defin:Krein}. Then, we obtain Lemma \ref{lemma:alpha-0} and Lemma \ref{lemma:H<u}, which will help us to establish the proof of the orbital stability. With the aid of methods in \cite{GrillakisSS-87,GrillakisSS-90,KapitulaP-13}, we provide an orbital stability analysis and come to the following theorems.

\begin{theorem}\label{theorem:orbital-cn-}
	If the $\cn$-type solutions $u(\xi)$ are spectrally stable with respect to perturbations of period $2PT,P\in \mathbb{Z}_{+}$ and $P< \frac{\pi}{\pi+ M(z_c)}$, then they are orbitally stable in the space $H^2_{per}([-PT,PT])$.
\end{theorem}

\begin{theorem}\label{theorem:orbital-dn-}
	The $\dn$-type solutions $u(\xi)$ are  orbitally stable in the space $H^2_{per}([-PT,PT])$, $P\in \mathbb{Z}_{+}$.
\end{theorem}

For the integrable equations, a particular feature is that there exist abundant exact solutions with diverse dynamics. We will provide some exact solutions to describe the stable and unstable dynamics. Based on the Darboux-B\"{a}cklund transformation, we construct breather solutions $u^{[1]}(\xi,t)$ and $u^{[2]}(\xi,t)$ corresponding to the parameter $\lambda$. For the $\dn$-type solutions, we construct new solutions $u^{[1]}(\xi,t)$ \eqref{eq:breather-dn-function} of equation \eqref{eq:mKdV1}. By choosing a special parameter $z$, we use the solution $u^{[1]}(\xi,t)$ to describe stable dynamics (See Figure \ref{fig:breaher-dn}).

For the $\cn$-type solutions, we construct a new solution $u^{[2]}(\xi,t)$ \eqref{eq:breather-cn-function} of equation \eqref{eq:mKdV1}, which could describe unstable dynamics (See Figure \ref{fig:breaher-cn}). As $t \rightarrow\pm\infty$, the function $u^{[2]}(\xi,t)$ could be regarded as a translation of $u(\xi)$ in \eqref{eq:elliptic-solution},
\begin{equation}\label{eq:t-infty}
	\begin{split}
		&u_{\pm \infty}^{[2]}(\xi)=\lim_{t\rightarrow \pm\infty}u^{[2]}(\xi,t)=\alpha k\cn(\alpha\xi\pm 2\ii (z-z^*)). 
	\end{split}
\end{equation}
More precisely, the asymptotic behavior of function $u^{[2]}(\xi,t)$ is given by
\begin{equation}\label{eq:u2-cn}
	\begin{split}
		u^{[2]}(\xi,t)=&u_{\pm \infty}^{[2]}(\xi)+w_{\pm}(\xi,t)+\mathcal{O}\left(\ee^{\mp4E_{R}t} \right),\qquad t\rightarrow \pm \infty,
	\end{split}
\end{equation}
where $w_{\pm}(\xi,t)$ is defined in equation \eqref{eq:defin-w}. Based on equation \eqref{eq:u2-cn}, the linearly unstable dynamics for the $\cn$-type solutions will be shown by the breather $u^{[2]}(\xi,t)$ in Subsection \ref{sub:cn-unstable-explicit}.	

The main contributions of this work are the following:
\begin{itemize}
	\item We study the linearized spectral problem of the focusing mKdV equation on the elliptic function background. For the unstable case, we consider subharmonic perturbations with the integer multiples of the period and then give the necessary and sufficient conditions for spectral stability with respect to the subharmonic perturbations. Furthermore, based on the above stable conditions of spectral stability, we study the orbital stability problem.
	
	\item Compared to previous studies on the stability problem by the MSW method, we use the theta function theory to develop this method. There are some advantages of utilizing theta functions. On the one hand, for the calculations of Jacobi elliptic functions, we can analyze the poles or zeroes by using the Liouville theorem to avoid complicated computations, as in \cite{FengLT-19}. On the other hand, since the spectrum can be represented by the Zeta function for the stability analysis, we can use the conformal transformation between $\lambda$ and $z$ to establish the spectral and orbital stability.
	
	\item With the aid of the Darboux-B\"{a}cklund transformation, we construct the breather solutions represented by theta functions to exhibit the stable or unstable dynamics. Through the representation of theta functions to breather solutions, their asymptotic analysis can be performed, which is consistent with the linear stability analysis for elliptic function solutions as $t\to\pm \infty$.
\end{itemize}

\subsection{Outline for this work}
The organization of this work is as follows. In Section \ref{sec:solutions}, using the effective integration method \cite{Kamchatnov-97,Kamchatnov-00,Shin-12}, we obtain the elliptic function solutions of the mKdV equation and the fundamental solutions for the corresponding Lax pair. With the aid of the theory of theta functions, the Jacobi elliptic solutions can be rewritten by theta functions. In Section \ref{section:spectral analysis}, we study the linearized spectral problem of the focusing mKdV equation by using the squared eigenfunctions and analyze the spectral stability of the periodic waves with respect to subharmonic perturbations. In Section \ref{section:orbial stability}, based on the spectrally stable condition, we further prove the orbital stability of periodic waves in a proper functional space. In Section \ref{section:breather solutions}, based on the Darboux-B\"acklund transformation, we construct breather solutions to exhibit the stable or unstable dynamics of the mKdV equation.

\section{Elliptic function solutions of the mKdV equation and its Lax pair}\label{sec:solutions}

In this section, we aim to get the elliptic function solutions of the mKdV equation and the fundamental solutions of the corresponding Lax pair by using the algebraic geometry method \cite{BelokolosBEI-94} and effective integration method \cite{Kamchatnov-97,Kamchatnov-00,Shin-12}. More basic theories and methods are mentioned in the references \cite{FengLT-19,TracyCL-84}. Under the condition of the genus-$1$ case, we obtain the elliptic function solutions of the focusing mKdV equation by the effective integration technique. And then, the solutions of the Lax pair are represented by theta functions for the uniform parameter $z$.

Matrices $\mathbf{U}(\lambda;u)$ and $\mathbf{V}(\lambda;u)$ defined in Lax pair \eqref{eq:Lax pair} satisfy the following symmetric properties:
\begin{equation}\label{eq:sym-1}
	\mathbf{U}^{\dagger}(\lambda^*;u)=-\mathbf{U}(\lambda;u),
	\,\,\,
	\mathbf{U}^{\top}(-\lambda;u)=-\mathbf{U}(\lambda;u);\quad \mathbf{V}^{\dagger}(\lambda^*;u)=-\mathbf{V}(\lambda;u),
	\,\,\,
	 \mathbf{V}^{\top}(-\lambda;u)=-\mathbf{V}(\lambda;u),
\end{equation}
by which we deduce that if $\Phi(x,t;\lambda)$ is a solution of \eqref{eq:Lax pair}, matrices $\Phi^{\dagger}(x,t;\lambda^*)$ and $\sigma_2 \Phi(x,t;\lambda)^{\top}\sigma_2^{\top}$ are both the solutions of the adjoint Lax pair:
\begin{equation}\label{eq:Phi^dagger}
		\Psi_x(x,t;\lambda)=-\Psi(x,t;\lambda) \mathbf{U}(\lambda;u),\qquad
		\Psi_t(x,t;\lambda)=-\Psi(x,t;\lambda)\mathbf{V}(\lambda;u).
\end{equation}
Combining Lax pair \eqref{eq:Lax pair} with its adjoint form \eqref{eq:Phi^dagger}, we can verify that the matrix function 
\begin{equation}\label{eq:L-Phi-matrix}
	\mathbf{L}(x,t;\lambda):=\frac{1}{2}\Phi(x,t;\lambda) \sigma_3 \sigma_2 \Phi(x,t;\lambda)^{\top}\sigma_2^{\top }, \qquad \sigma_2=
	\begin{bmatrix}
		0 & -\ii \\ 
		\ii & 0
	\end{bmatrix},
\end{equation}
satisfies the stationary zero curvature equations
\begin{equation}\label{eq:LL1}
	\mathbf{L}_x(x,t;\lambda)=[\mathbf{U}(\lambda;u),\mathbf{L}(x,t;\lambda)], \qquad \mathbf{L}_t(x,t;\lambda)=[\mathbf{V}(\lambda;u),\mathbf{L}(x,t;\lambda)].
\end{equation}
The compatibility condition of the above equations \eqref{eq:LL1}, $\mathbf{L}_{xt}(x,t;\lambda)=\mathbf{L}_{tx}(x,t;\lambda)$, also yields the \eqref{eq:mKdV} equation.

Suppose the function matrices 
\begin{equation}\nonumber
	\Phi(x,t;\lambda):=
	\begin{bmatrix}
		\phi_1(x,t;\lambda) & \phi_2 (x,t;\lambda)\\
		\psi_1(x,t;\lambda) & \psi_2(x,t;\lambda)
	\end{bmatrix}, \qquad \text{and} \qquad 
	\mathbf{L}(x,t;\lambda)=
	\begin{bmatrix}
		-f(x,t;\lambda) &g(x,t;\lambda)\\  
		h(x,t;\lambda) & f(x,t;\lambda)
	\end{bmatrix},
\end{equation}
satisfy the Lax pair \eqref{eq:Lax pair} and the stationary zero curvature equation \eqref{eq:LL1}, respectively. We aim to calculate the exact expression of matrix function $\mathbf{L}(x,t;\lambda)$. For the genes-1 case, we assume that $\mathbf{L}(x,t;\lambda)$ is a quadratic polynomial of $\lambda$: $\mathbf{L}(x,t;\lambda)= \mathbf{L}_0(x,t)\lambda^2+\mathbf{L}_1(x,t)\lambda+\mathbf{L}_2(x,t)$. Inserting this ansatz into equation \eqref{eq:LL1} and comparing the coefficients of $\lambda$, we obtain
\begin{equation}\label{eq:soL1}
	\mathbf{L}(x,t;\lambda)=-\ii(\alpha_0\lambda+\alpha_1)\mathbf{U}(\lambda;u) -\frac{\alpha_0}{2}\sigma_3\left( \mathbf{Q}^2- \mathbf{Q}_x\right)-\alpha_2\sigma_3,
\end{equation} 
where $\alpha_i \in \mathbb{R},i=0,1,2$, and matrices $\mathbf{Q}$ and $\mathbf{U}(\lambda;u)$ are defined in equation \eqref{eq:U,V}. Furthermore, we get $\alpha_1=0$ and
\begin{equation}\label{eq:Qtx1}
	\alpha_0u_t+4\alpha _2u_x=0.
\end{equation}
Without loss of generality, we can set $\alpha_0=1$. The determinant of $\mathbf{L}(x,t;\lambda)$ is given by
\begin{equation}\label{eq:det L}
		\det(\mathbf{L}(x,t;\lambda))=-f(x,t;\lambda)^2-g(x,t;\lambda)h(x,t;\lambda)
		=-\lambda^4-s_1\lambda^3-s_2\lambda^2-s_3\lambda-s_4
		\equiv P(\lambda).
\end{equation}
Comparing the coefficients of the $\det(\mathbf{L}(x,t;\lambda))$, we obtain
\begin{equation}\label{eq:s1s21}
	s_1=0,\qquad \alpha_2=\frac{1}{2}s_2,\qquad
	\frac{s_3}{u^2}=\mu-\mu=0, 
	\qquad \mu^2 =\frac{R(u)}{4u^2},
\end{equation}
where $\mu:=-\frac{\ii }{2}(\ln u)_x$ and $R(u)=(u^2-s_2)^2-4s_4$. From equation \eqref{eq:s1s21} and the definition of $\mu$, it follows
\begin{equation}\label{eq:Rutux}
	u_t=-4\alpha_2u_x =-2s_2u_x=-2s_2\sqrt{-R(u)}.
\end{equation}
Under the transformation \eqref{eq:traveling-wave-transform}, equation \eqref{eq:Rutux} can be reduced to
\begin{equation}\label{eq:u xi}
u_{\xi}=\sqrt{-R(u)}.
\end{equation}
Then we have the following proposition:
\begin{prop}\label{prop:solution-u^2}
	The modulus square of elliptic function solutions of equation \eqref{eq:mKdV1} could be represented as
	\begin{equation}\label{eq:solution u^2}
		u^2(\xi)=k^2\alpha^2(\sn^2(K+2\ii l)-\sn ^2(\alpha \xi)),\qquad l=0 \,\, \text{  or  }\,\, \frac{K'}{2},
	\end{equation}
	where the modulus $k=\sqrt{\frac{u_3-u_2}{u_3-u_1}},$ $ \alpha^2=u_3-u_1$, and $u_{1,2,3}$ can be parameterized by
	\begin{equation}\label{eq:u1,u2,u3}
		u_1=-\alpha^2\dn^2(K+2\ii l),\qquad u_2=-k^2\alpha^2\cn^2(K+2\ii l), \qquad
		u_3=k^2\alpha^2\sn^2(K+2\ii l).
	\end{equation}
\end{prop}
\begin{proof}
	Squaring equation \eqref{eq:u xi} and multiplying both sides by $u^2$, we obtain 
	\begin{equation}\label{eq:ux u}
			\left((u^2)_{\xi}\right)^2
			=-4(u^2)^3+8s_2(u^2)^2+4(4s_4-s_2^2)u^2
			=-4(u^2-u_1)(u^2-u_2)(u^2-u_3), 
	\end{equation}
	where $u_1,u_2$, and $u_3$ are given by \eqref{eq:u1,u2,u3}. When $\alpha>0,l\in [0,\frac{K'}{2}],k\in(0,1)$, the range of parameters \eqref{eq:u1,u2,u3} is $u_1\le 0 \le u_2<u_3$. Furthermore, we show the equivalence between the triple tuples $(u_1,u_2,u_3)$ and the one $(\alpha,k,l)$ in Remark \ref{remark:implicit-u123}. Comparing the coefficients of equation \eqref{eq:ux u} with respect to $u^2$, we get $s_2=\frac{1}{2}(u_1+u_2+u_3)$, $u_1=0$ or $u_2=0$, i.e., $l=\frac{K'}{2}$ or $l=0$. Thus, by the Jacobi elliptic function theory, the solution for equation \eqref{eq:ux u} is given by function $k^2\alpha^2(\sn^2(K+2\ii l)-\sn ^2(\alpha \xi))$. Thus, by the elliptic function theory, the solution for equation \eqref{eq:ux u} is given by \eqref{eq:solution u^2}.
\end{proof}

\begin{remark}\label{remark:implicit-u123}
	There is a one-to-one correspondence between the triple tuples $(u_1,u_2,u_3)$ and $(\alpha,k,l)$, where $u_1\le 0 \le u_2<u_3$ and $(\alpha,k,l)$ is in the region $\{(\alpha,k,l)|\alpha>0,0<k<1,0\le l\le \frac{K'}{2}\}$. Based on the inverse function theorem, we only need to verify the non-degenerate for the Jacobian matrix of $\alpha,k,l$. Actually, by derivative formulas of the Jacobi elliptic functions with respect to variable $z$ \cite[p.25]{ByrdF-54}, the Jacobian matrix between the triple tuples $(u_1,u_2,u_3)$ and $(\alpha,k,l)$ is
	\begin{equation}\label{eq:define-scd}
		\frac{\partial\left( u_1,u_2,u_3 \right)}{\partial\left( \alpha,k,l \right)}=16 \ii \alpha^5 k^3 \scd(K+2\ii l), \qquad  \scd(\cdot):=\sn(\cdot)\cn(\cdot)\dn(\cdot),	
	\end{equation}
 which is non-degenerate for any $\alpha>0$, $k\in(0,1)$ and $l\in (0,\frac{K'}{2})$.
\end{remark}

Under the coordinate transformation \eqref{eq:traveling-wave-transform}, the solution matrix $\Phi(x,t; \lambda)$ and the matrix $\mathbf{L}(x,t;\lambda)$ turn to:
	\begin{equation}\label{eq:Lax pair-1}
	\Phi(\xi,t; \lambda)
	=\begin{bmatrix}
	\phi_1(\xi,t; \lambda) & \phi_2(\xi,t; \lambda)
	\\ 
	\psi_1(\xi,t; \lambda) & \psi_2(\xi,t; \lambda)
	\end{bmatrix}, \quad 
	\mathbf{L}(\xi,t ;\lambda)=
	\begin{bmatrix}
	-f(\xi,t ;\lambda)  &  g(\xi,t ;\lambda)\\
	h(\xi,t ;\lambda)  &  f(\xi,t ;\lambda)
	\end{bmatrix},
	\end{equation}
where
\begin{equation}\label{eq:f-g-h-xi}
	f(\xi,t ;\lambda)=\lambda^{2}+\frac{s_2}{2}-\frac{1}{2} u^2, \quad
	g(\xi,t ;\lambda)=-\ii  u(\lambda  - \mu), \quad h(\xi,t ;\lambda)=\ii u(\lambda+\mu), \quad \mu=-\frac{\ii }{2}(\ln u)_{\xi}.
\end{equation}
The Lax pair with respect to parameters $\xi$ and $t$ is
	\begin{equation}\label{eq:Lax-pair-1}
		\Phi_{\xi}(\xi,t ;\lambda)=\mathbf{U}(\lambda;u)\Phi(\xi,t ;\lambda), \qquad
		\Phi_{t}(\xi,t ;\lambda)=\hat{\mathbf{V}}(\lambda;u)\Phi(\xi,t ;\lambda),\,\,\, \hat{\mathbf{V}}(\lambda;u):=\mathbf{V}(\lambda;u)+2s_2\mathbf{U}(\lambda;u).
	\end{equation} 
	Now, we proceed to obtain the solutions of the Lax pair \eqref{eq:Lax-pair-1}. Firstly, we consider the eigenvalue of $\mathbf{L}(\xi,t ;\lambda)$. Set the determinant of $\mathbf{L}(\xi,t ;\lambda)$ to be $-y^2$, i.e., $y^2=f^2(\xi,t ;\lambda)+g(\xi,t;\lambda)h(\xi,t;\lambda)$, and then $\pm y$ are the eigenvalues of the matrix function $\mathbf{L}(\xi,t;\lambda)$. Considering the eigenvector of $\mathbf{L}(\xi,t;\lambda)$, we get the following lemma:
\begin{lemma} \label{lemma:span}
	The linear spaces $\mathrm{span}\left\{ \left( 1,r_1(\xi,t;\lambda) \right)^{\top} \right\} $ and $\mathrm{span}\left\{ \left( 1,r_2(\xi,t;\lambda) \right)^{\top} \right\} $ are the kernels of matrices $\mathbf{L}(\xi,t;\lambda)-y\mathbb{I}$ and $\mathbf{L}(\xi,t;\lambda)+y\mathbb{I}$, respectively,
	where
	\begin{equation}\label{eq:r1,r2}
		r_1(\xi ,t  ;\lambda)=\frac{f(\xi ,t  ;\lambda)+y}{g(\xi ,t  ;\lambda)}=\frac{h(\xi ,t  ;\lambda)}{y-f(\xi ,t  ;\lambda)}, \qquad r_2(\xi ,t  ;\lambda)=\frac{f(\xi ,t  ;\lambda)-y}{g(\xi ,t  ;\lambda)}=\frac{-h(\xi ,t  ;\lambda)}{y+f(\xi ,t  ;\lambda)}.
	\end{equation}
	In addition, the linear spaces $\mathrm{span}\left\{ \hat{\Phi }_1(\xi,t;\lambda) \right\} $ and $\mathrm{span}\left\{ \hat{\Phi }_2(\xi,t;\lambda) \right\} $ are also the kernels of matrices $\mathbf{L}(\xi,t;\lambda)-y\mathbb{I}$ and $\mathbf{L}(\xi,t;\lambda)+y\mathbb{I}$ respectively, where 
	\begin{equation}\label{eq:Phi-hat}
		\begin{bmatrix}
			\hat{\Phi }_1(\xi,t;\lambda) & \hat{\Phi }_2(\xi,t;\lambda)
		\end{bmatrix}:=\Phi(\xi,t;\lambda)\begin{bmatrix}
			1 & 1 \\r_1(0,0;\lambda) & r_2(0,0;\lambda) 
		\end{bmatrix},
	\end{equation}
	and the matrix function $\Phi(\xi,t;\lambda)$ is the fundamental solution of Lax pair \eqref{eq:Lax-pair-1} with $\Phi(0,0;\lambda)=\mathbb{I}$.
\end{lemma}

\begin{proof}
	By the definition of function $r_1(\xi,t;\lambda)$ in \eqref{eq:r1,r2}, it is easy to verify that $\mathrm{span}\left\{ \left( 1,r_1(\xi,t;\lambda) \right)^{\top} \right\} $ is the kernel of matrix $\mathbf{L}(\xi,t;\lambda)-y\mathbb{I}$ and $\mathrm{span}\left\{ \left( 1,r_2(\xi,t;\lambda) \right)^{\top} \right\} $ is the kernel of matrix $\mathbf{L}(\xi,t;\lambda)+y\mathbb{I}$.
	If $\Phi(\xi,t;\lambda)$ is a solution of Lax pair \eqref{eq:Lax-pair-1}, then the matrix $\mathbf{L}(\xi,t;\lambda)\Phi(\xi,t;\lambda)$ is a solution of Lax pair \eqref{eq:Lax-pair-1}. On the other hand, the matrix function $\Phi(\xi,t;\lambda)\mathbf{L}(0,0;\lambda)$ is also the solution of Lax pair \eqref{eq:Lax-pair-1}. The solutions $\Phi(\xi,t;\lambda)\mathbf{L}(0,0;\lambda)$ and $\mathbf{L}(\xi,t;\lambda)\Phi(\xi,t;\lambda)$ share the same initial condition at $(\xi,t)=(0,0)$, since  $\Phi(0,0;\lambda)=\mathbb{I}$. By the uniqueness and existence theorem of the ordinary differential equation, we get $\Phi(\xi,t;\lambda)\mathbf{L}(0,0;\lambda)=\mathbf{L}(\xi,t;\lambda)\Phi(\xi,t;\lambda)$.
	Then, the vector $\Phi(\xi,t;\lambda)
	\begin{bmatrix}
		1 \\r_1(0,0;\lambda) 
	\end{bmatrix}$ is the kernel of $\mathbf{L}(\xi,t;\lambda)-y\mathbb{I}$, and the vector $\Phi(\xi,t;\lambda)
	\begin{bmatrix}
		1 \\r_2(0,0;\lambda) 
	\end{bmatrix}$ is the kernel of $\mathbf{L}(\xi,t;\lambda)+y\mathbb{I}$. We could refer to \cite{FengLT-19,LingS-2021} for the detailed calculation.
\end{proof}

Since both vectors $\left\{ \left( 1,r_1(\xi,t;\lambda) \right)^{\top} \right\} $ and $\left\{\hat{\Phi}_1(\xi,t;\lambda)\right\}$ are the kernels of matrices $\mathbf{L}(\xi,t;\lambda)-y\mathbb{I}$ in Lemma \ref{lemma:span}, we obtain
\begin{equation}\label{eq:r_1-phi-psi}
	r_i(\xi,t;\lambda)=\frac{\psi_i(\xi,t;\lambda)}{\phi_i(\xi,t;\lambda)}, \qquad i=1,2.
\end{equation} 
So $\phi_i=\phi_i(\xi,t;\lambda), i=1,2$, can be derived from the equations:
\begin{equation}\label{eq:phi_x,phi}
	\phi_{i,\xi }=-\ii \lambda\phi_i+u r_i\phi_i, \qquad
		\phi_{i,t}=(2\ii \lambda u^2-4\ii \lambda^3-2\ii \lambda s_2)\phi_i+(2s_2u-2u^3-u_{\xi\xi}+4\lambda^2u+2\ii \lambda u_{\xi })r_i\phi_i,
\end{equation}
Combining the first equation of functions $r_{1,2}(\xi,t;\lambda)$ in \eqref{eq:r1,r2} with the equation \eqref{eq:phi_x,phi}, we obtain functions $\phi_{1,2}$:
\begin{equation}\label{eq:phi1,2}
	\phi_1(\xi,t)=\sqrt{\frac{u^2(\xi )-\beta_1}{u^2(0)-\beta_1}}\exp\left(\theta_1\right), \qquad
	\phi_2(\xi,t)=\sqrt{\frac{u^2(\xi )-\beta_2}{u^2(0)-\beta_2}}\exp\left(\theta_2\right),
\end{equation}
where $\phi_{1,2}(0,0)=1$ and
\begin{equation}\label{eq:beta1-beta2-theta-1-theta-2}
	\beta_{1,2}=2\lambda^2+s_2\mp 2y, \qquad 
	\theta_{1,2}=\int_{0}^{\xi}\frac{2\ii \lambda\beta_i}{u^2(s)-\beta_i}\mathrm{d} s +\ii \lambda \xi \pm 4\ii \lambda y t.
\end{equation}
Considering the second equation of function $r_{1,2}(\xi,t;\lambda)$ in \eqref{eq:r1,r2} and equation \eqref{eq:phi_x,phi}, we obtain 
\begin{equation}
	\psi_1(\xi,t)=\psi_1(0,0)\sqrt{\frac{u^2(\xi )-\beta_2}{u^2(0)-\beta_2}}\exp\left(-\theta_2\right), \qquad
	\psi_2(\xi,t)=\psi_2(0,0)\sqrt{\frac{u^2(\xi)-\beta_1}{u^2(0)-\beta_1}}\exp\left(-\theta_1\right),
\end{equation}
where $\psi_{1}(0,0)=-\sqrt{\frac{u^2(0)-\beta_{2}}{u^2(0)-\beta_{1}}}$ and $\psi_{2}(0,0)=-\sqrt{\frac{u^2(0)-\beta_{1}}{u^2(0)-\beta_{2}}}$. We get the following theorem by ignoring the constant factors of vector solutions.
\begin{theorem}\label{theorem:solution}
	A fundamental solution of Lax pair \eqref{eq:Lax pair} is given by 
	\begin{equation}\label{eq:Phi}
		\Phi(x,t  ;\lambda)=
		\begin{bmatrix}
			\sqrt{u^2(\xi )-\beta_1}\exp(\theta_1) & 
			\sqrt{u^2(\xi )-\beta_2}\exp(\theta_2) \\
			-\sqrt{u^2(\xi )-\beta_2}\exp(-\theta_2) &
			-\sqrt{u^2(\xi )-\beta_1}\exp(-\theta_1)
		\end{bmatrix},
	\end{equation}
	where $\xi=x-2s_2t$, $\beta_{1,2}$ and $\theta_{1,2}$ are defined in equation \eqref{eq:beta1-beta2-theta-1-theta-2}.
\end{theorem}

In what follows, we aim to use theta functions to represent solutions $u,\phi_i,\psi_i, \ i=1,2$. Taking the derivative of the second equation of \eqref{eq:f-g-h-xi} with respect to $\xi$ yields $g_{\xi}=-\ii u_{\xi}(\lambda-\mu)+\ii u\mu_{\xi}$. Setting $\lambda=\mu$, we obtain $g_{\xi}=\ii u \mu_{\xi}$ and $g_{\xi}=-2\ii \lambda g+2uf$ by utilizing the stationary zero curvature equation of the matrix $\mathbf{L}(\xi,t;\lambda)$, which implies $\mu_{\xi}=-2\ii f$. It follows from \eqref{eq:det L} that
\begin{equation}\label{eq:y lambda }
	(\mu_{\xi})^2=-4 f^2 =4P(\mu), \qquad \text{and} \qquad -y^2=P(\lambda)\equiv-\prod_{i=1}^{4}(\lambda-\lambda_i),
\end{equation}
which means that the algebraic curve with genus one can be parameterized by the uniformization variable $z$: 
\begin{equation}\label{eq:y_mu_lambda}
	y(z)=\frac{\alpha}{2}\frac{\dd }{\dd z}\mu\left(\frac{\ii (z-l)}{\alpha}\right), \qquad \text{and} \qquad \lambda(z)=\mu\left(\frac{\ii (z-l)}{\alpha}\right),
\end{equation}
where $l =0$ or $l=\frac{K'}{2}$. Then, we establish the conformal map (in Appendix \ref{appendix:commonality map}) between $\lambda$-plane and $z$-plane by the following proposition.
\begin{prop} \label{prop:lambda-comforming}
	The function $\lambda(z)$: 
	\begin{subequations}\label{eq:lambda}
		\begin{align}
			\lambda(z)
			=&\frac{\ii \alpha }{2}\dn(\ii (z-l))\tn(\ii (z-l)),\quad l=0, \label{eq:lambda-a} \\
			\lambda(z)
			=&\frac{\ii \alpha k^2}{2}\sn(\ii(z-l))\cd(\ii(z-l)), \quad l=\frac{K'}{2}, \label{eq:lambda-b}
		\end{align}
	\end{subequations}
	constructs the conformal mapping, which maps the rectangle in $z$-plane onto a whole $\lambda$-plane with two cuts.
\end{prop}
\begin{proof}
	By the definition of $\mu:=-\frac{\ii}{2}(\ln u)_{\xi}$ and solution \eqref{eq:elliptic-solution}, we get 
	\begin{equation}\label{eq:mu}
		\mu(\xi  )
		=\frac{\ii \alpha\dn(\alpha \xi )\tn(\alpha \xi)}{2} , \qquad \text{and}
		\qquad \mu(\xi  )
		=\frac{\ii\alpha k^2\sn(\alpha \xi )\cd(\alpha \xi )}{2},
	\end{equation}
	with solutions $u=k\alpha\cn(\alpha\xi)$ and $u=\alpha\dn(\alpha\xi)$, respectively.
	Moreover, by \eqref{eq:y lambda } and \eqref{eq:y_mu_lambda}, we get \eqref{eq:lambda} and the elliptic integral 
	\begin{equation}\label{eq:z-l-lambda}
		\begin{split}
			z(\lambda)-l=\frac{\alpha}{2}\int_{0}^{\lambda}\frac{\dd s}{\sqrt{(s^2-\lambda_1^2)(s^2-\lambda_1^{*2})}},\,\, l=0, \qquad
			z(\lambda)-l=\frac{\alpha}{2}\int_{0}^{\lambda}\frac{\dd s}{\sqrt{(s^2-\lambda_2^2)(s^2-\lambda_3^2)}},\,\, l=\frac{K'}{2},
		\end{split} 
	\end{equation}
	where $\lambda_1=\frac{\ii \alpha}{2}(k-\ii k'),\lambda_2=\frac{\ii \alpha k^2}{2(1+k')},\lambda_3=\frac{\ii \alpha(1+k')}{2}$.
	Lemma \ref{lemma:map} shows that $z(\lambda)$ is a conformal mapping that maps the upper half plane onto the rectangle $[- \frac{\ii K}{2},\frac{\ii K}{2}]\times [0, K']$ (refer to Figure \ref{fig:map1} in Appendix \ref{appendix:commonality map}). Since $\lambda(z)$ is the inverse function of $z(\lambda)$, we can prove that $\lambda(z)$ is a conformal mapping that maps the rectangle $[-\frac{\ii K}{2},\frac{\ii K}{2}]\times [- K', K']$ onto a whole plane with cuts connecting the points $\lambda_i$s.
\end{proof}

Therefore, by the above analysis, we need to consider the $z$-region:
\begin{equation}\label{eq:set S}
	S:=\left\{z\in \mathbb{C} \left|-K'+l\le {\Re}  (z)\le K'+l, -\frac{K}{2}\le {\Im}(z)\le \frac{K}{2} \right.\right\}.
\end{equation}

\begin{remark}
	For $l=0$ or $\frac{K'}{2}$, the function $\lambda^2(z)$ could be written in a uniform form
	\begin{equation}\label{eq:lambda^2}
		\lambda^2(z)=\frac{\alpha^2}{4}\left( \dn^2(K+2\ii l)+k^2-2+\dn^2(\ii(z-l))+\dn^2(\ii(z+l)+K+\ii K')\right).
	\end{equation}
\end{remark}
\begin{lemma}\label{lemma:lambda-2 y beta1-2 }
For $l=0$ or $\frac{K'}{2}$, we can get the following representation:
	\begin{equation}\nonumber
		u^2-\beta_1= \frac{\alpha^2\vartheta_2^2\vartheta_4^2\vartheta_1\left(\frac{\ii(z-l)-\alpha\xi}{2K} \right)\vartheta_1\left(\frac{\ii(z-l)+\alpha\xi}{2K} \right)}{\vartheta_3^2\vartheta_4^2(\frac{\alpha\xi}{2K})\vartheta_4^2(\frac{\ii (z-l)}{2K})},\quad
		u^2-\beta_2= \frac{\alpha^2\vartheta_2^2\vartheta_4^2\vartheta_3\left(\frac{\ii(z+l)-\alpha\xi}{2K} \right)\vartheta_3\left(\frac{\ii(z+l)+\alpha\xi}{2K} \right)}{\vartheta_3^2\vartheta_4^2(\frac{\alpha\xi}{2K})\vartheta_2^2(\frac{\ii (z+l)}{2K})}.
	\end{equation}
\end{lemma}

\begin{proof} 
	It is easy to verify $(4\ii \lambda\beta_1)^2=-4(\beta_1-u_1)(\beta_1-u_2)(\beta_1-u_3)$ by \eqref{eq:det L}, \eqref{eq:ux u}, and \eqref{eq:y lambda }. Combining \eqref{eq:det L}, \eqref{eq:s1s21}, \eqref{eq:y lambda } together with \eqref{eq:y_mu_lambda}, we obtain $2\lambda\beta_1=2\lambda(2\lambda^2+s_2-2y)=2y\left(\frac{\dd  y}{\dd  \lambda}-2\lambda\right)=\alpha \frac{\dd  }{\dd  z}\left(y(z)-\lambda^2(z)\right)$. Furthermore, we have $-4\lambda\beta_1=\alpha\beta_{1,z}$ and then $(-\ii \alpha \beta_{1,z})^2= -4(\beta_1-u_1)(\beta_1-u_2)(\beta_1-u_3)$
	holds.
	By \eqref{eq:ux u}, we know $\beta_1=\alpha^2 k^2 \left(\sn^2(K+2\ii l)- \sn^2(\ii (z-c)) \right)$, where $ c\in \mathbb{C}$ is an undetermined constant. Combining \eqref{eq:y_mu_lambda} with \eqref{eq:lambda^2}, we substitute $z=0$ into them and obtain $\beta_1=2\lambda^2+s_2-2y
	=\alpha^2k^2\left(\sn^2(K+2\ii l)-\sn^2(\ii l) \right)$, which implies $c=l$.	 
	From the existence and uniqueness theorem of the ordinary differential equation, we get $\beta_1=\alpha^2 k^2 \left(\sn^2(K+2\ii l)-\sn^2(\ii (z-l)) \right)$. Based on $-4\lambda\beta_1=\alpha\beta_{1,z}$, we have $2\lambda\beta_1=\frac{\alpha}{2}\beta_{1,z}=\ii \alpha^3k^2\scd(\ii(z-l))$, where $\scd(z)$ is defined in \eqref{eq:define-scd}. By solution \eqref{eq:solution u^2}, the following equation holds:
	\begin{equation}\label{eq:u^2-beta1}
			u^2-\beta_1
			=\alpha^2 k^2 \left(\sn^2(\ii (z-l))-\sn ^2(\alpha \xi ) \right).
	\end{equation}
	Similarly, we obtain $\beta_2=\alpha^2 k^2 \left(\sn^2(K+2\ii l)-\sn^2(\ii (z+l)+K+\ii K') \right)$, $2\lambda\beta_2
	=-\ii \alpha^3k^2\scd(\ii(z+l)+K+\ii K')$ and
	\begin{equation}\label{eq:u^2-beta2}
			u^2-\beta_2=\alpha^2 k^2\left(\sn^2(\ii (z+l)+K+\ii K')-\sn ^2(\alpha \xi )\right).
	\end{equation}

	Then, we use theta functions to represent functions \eqref{eq:u^2-beta1} and \eqref{eq:u^2-beta2}, which are double periodic meromorphic functions with respect to variable $\xi$ having the period $\frac{2K}{\alpha},\frac{\ii K'}{\alpha}$. So we merely analyze functions in the periodic area $\xi \in [-\frac{K}{\alpha},\frac{K}{\alpha}]\times[0,\frac{\ii K'}{\alpha}]$.  We first consider function \eqref{eq:u^2-beta1}. Rewriting it as $(\sn(\ii (z-l))-\sn (\alpha \xi ))(\sn(\ii (z-l))+\sn(\alpha \xi ))$, we get that $\xi=\frac{\pm\ii(z-l)}{\alpha}$ are the simple zeros. And the point  $\xi=-\frac{\ii K'}{\alpha}$ is the double pole. Then we have $\sn^2(\ii (z-l))-\sn ^2(\alpha \xi )=
	C_1\frac{\vartheta_1\left(\frac{\ii(z-l)-\alpha\xi}{2K} \right)\vartheta_1\left(\frac{\ii(z-l)+\alpha\xi}{2K} \right)}{\vartheta_4^2(\frac{\alpha \xi}{2K})}$, where $C_1$ is an undetermined constant. Plugging $\xi=0$ in the above equation, we get $C_1=\frac{\vartheta_3^2\vartheta_4^2}{\vartheta_2^2\vartheta_4^2(\frac{\ii (z-l)}{2K})}$. Similarly, we could express the function $u^2-\beta_2$ in terms of theta functions.
	Then, Lemma \ref{lemma:lambda-2 y beta1-2 } holds.
\end{proof}

By \eqref{eq:u^2-beta1} and \eqref{eq:u^2-beta2}, the shift formula of Jacobi theta functions \cite[p.86]{ArmitageE-06}, the translation formulas between Jacobi elliptic and theta functions \cite[p.83]{ArmitageE-06}, and Lemma \ref{lemma:int}, we obtain the exact expressions of solutions $\phi_{1,2},\psi_{1,2}$ in \eqref{eq:Phi}
\begin{equation}\label{eq:phi1,2-dn}
	\phi_1=\alpha\frac{\vartheta_2\vartheta_4\vartheta_1(\frac{\ii (z-l)-\alpha\xi  }{2K})}{\vartheta_3\vartheta_4(\frac{\ii (z-l)}{2K})\vartheta_4(\frac{\alpha\xi  }{2K})} E_1(\xi,t;z), \qquad \phi_2
	=\alpha\frac{\vartheta_2\vartheta_4\vartheta_3(\frac{\ii (z+l)+\alpha\xi  }{2K})}{\vartheta_3\vartheta_2(\frac{\ii (z+l)}{2K})\vartheta_4(\frac{\alpha \xi  }{2K})} E_2(\xi,t;z),
\end{equation}
where 	 
\begin{equation}\label{eq:E1-E2}
		E_1(\xi,t;z)=\ee^{(\alpha Z(\ii (z-l))+\ii \lambda)\xi   +4\ii  y \lambda t },\qquad
		E_2(\xi,t;z)=\ee^{(-\ii\frac{\alpha\pi}{2K}-\alpha Z(\ii (z+l)+K+\ii K')+\ii \lambda)\xi   -4\ii  y \lambda t }.
\end{equation}
Based on the definition of functions $\lambda(z)$ in \eqref{eq:lambda}, $\mu(\xi)=-\frac{\ii}{2}(\ln u)_{\xi}$ and equation \eqref{eq:y_mu_lambda}, considering the variable $\alpha\xi$ as a whole, we know that the poles of the function $\lambda(z)-\mu(\xi)$ are $\alpha\xi=-2\ii l+(2m+1)K$ and $\alpha\xi=(2m+1)K+\ii(2n+1)K'$, $m,n\in \mathbb{Z}$. The zeros of the function $\lambda(z)-\mu(\xi)$ are $\alpha\xi=\ii (z-l)+2mK+2\ii nK'$, $\alpha\xi=-\ii (z+l)+(2m+1)K+\ii (2n+1)K'$, $n,m\in \mathbb{Z}$. Based on Liouville Theorem, we get 
	\begin{equation}\label{eq:lambda-mu}
		\lambda(z)-\mu(\xi)=\frac{\ii\alpha\vartheta_2\vartheta_4\vartheta_3(\frac{2\ii l}{2K})\vartheta_1(\frac{\ii(z-l)-\alpha\xi}{2K})\vartheta_3(\frac{\ii(z+l)+\alpha\xi}{2K})}{2\vartheta_3\vartheta_2(\frac{\ii(z+l)}{2K})\vartheta_4(\frac{\ii(z-l)}{2K})\vartheta_2(\frac{\alpha \xi +2\ii l}{2K})\vartheta_4(\frac{\alpha \xi}{2K})},
	\end{equation}
	since when $\xi=0$, $\lambda(z)-\mu(0)=\lambda(z)$.
	Then, we consider the function $r_1(\xi,t;\lambda)$ defined in \eqref{eq:r1,r2}. By \cite{FengLT-19}, the solution $u(\xi)$ could be expressed in terms of theta functions as $	u(\xi)=\alpha\frac{\vartheta_2\vartheta_4\vartheta_2(\frac{\alpha\xi+2\ii l}{2K})}{\vartheta_3\vartheta_3(\frac{2\ii l}{2K})\vartheta_4(\frac{\alpha\xi}{2K})}\ee^{-\alpha Z(K+2\ii l)\xi}$. 
Combining Lemma \ref{lemma:lambda-2 y beta1-2 } with \eqref{eq:r1,r2} and \eqref{eq:lambda-mu}, we get
\begin{equation}
	r_1(\xi,t)=-\frac{\vartheta_4(\frac{\ii (z-l)}{2K})\vartheta_3(\frac{-\ii (z+l)+\alpha\xi}{2K})}{\vartheta_1(\frac{\ii (z-l)-\alpha\xi}{2K})\vartheta_2(\frac{-\ii (z+l)}{2K})}\ee^{\alpha Z(2\ii l+K)\xi},\qquad r_2(\xi,t)=-\frac{\vartheta_2(\frac{\ii (z+l)}{2K})\vartheta_1(\frac{\ii (z-l)+\alpha\xi}{2K})}{\vartheta_3(\frac{\ii (z+l)+\alpha\xi}{2K})\vartheta_4(\frac{\ii (z-l)}{2K})}\ee^{\alpha Z(2\ii l+K)\xi},
\end{equation}
which implies that the function $\psi_{1,2}=r_{1,2}(\xi,t;\lambda)\phi_{1,2}$ could be rewritten as 
\begin{equation}
	\psi_1=-\alpha\frac{\vartheta_2\vartheta_4\vartheta_3(\frac{\ii (z+l)-\alpha\xi  }{2K})}{\vartheta_3\vartheta_2(\frac{\ii (z+l)}{2K})\vartheta_4(\frac{\alpha \xi  }{2K})} E_1(\xi,t;z)\ee^{\alpha Z(2\ii l+K)\xi}, \quad 
	\psi_2=-\alpha\frac{\vartheta_2\vartheta_4\vartheta_1(\frac{\ii (z-l)+\alpha\xi  }{2K})}{\vartheta_3\vartheta_4(\frac{\ii (z-l)}{2K})\vartheta_4(\frac{\alpha\xi  }{2K})} E_2(\xi,t;z)\ee^{\alpha Z(2\ii l+K)\xi}.
\end{equation}
Therefore, we establish Theorem \ref{theorem:solution-Theta form}. The solution $\Phi(x,t;\lambda)$ in Theorem \ref{theorem:solution} can be expressed in terms of theta functions in Theorem \ref{theorem:solution-Theta form}. Compared with the method in \cite{FengLT-19}, we tremendously simplify the tedious calculations by replacing the conversion formulas and the additional formulas of theta functions with the Jacobi elliptic function theory by analyzing poles and zeros.

\section{Spectral stability analysis}\label{section:spectral analysis}
In this section, we mainly focus on the linear stability analysis of the mKdV equation. We rewrite Lax pair \eqref{eq:Lax-pair-1} as the spectral problem
	\begin{equation}
		\begin{bmatrix}
			\ii \partial_{\xi} & -\ii u \\
			-\ii u  &  -\ii \partial_{\xi}
		\end{bmatrix}\Phi =\lambda \Phi.
	\end{equation}
We define the set of all $\lambda$ such that Lax pair \eqref{eq:Lax-pair-1} has bounded solutions as the Lax spectrum $\sigma(L)$ in \cite{DeconinckS-17,DeconinckyU-20}. Since the problem is not self-adjoint, therefore $\lambda$ is not confined to the real axis (i.e., $\sigma(L)\nsubseteq \mathbb{R} $), which is a main stumbling block to examining the stability of the focusing mKdV equation \cite{Deconinck-10}. To overcome it, we turn to study the uniform variable $z$ such that the perturbation $w(\xi;t)$ is a bounded function. We define the set satisfying the above conditions as $Q$ in \eqref{eq:set Q}. Based on the conformal mapping between the spectral parameter $\lambda$ and the uniform variable $z$, we obtain the region of the spectral parameter $\lambda$.

By the infinite-dimensional Hamiltonian structure of the spectrum \cite{HaragusK-08}, we know that an elliptic function solution is spectrally stable with respect to perturbations $W(\xi)\in C_b^0(\mathbb{R})$ if $\Omega\subset \ii \mathbb{R}$, which is provided in Definition \ref{define:spect-stable}. We gain the exact expression of the function $W(\xi)$ based on the squared-eigenfunction method. After studying the properties of the function $W(\xi)$, we get some fundamental lemmas, which are helpful in studying the spectral stability.

\begin{lemma}\label{lemma:W}
	All solutions of equation \eqref{eq:spectral} could be constructed from the function
	\begin{equation}\label{eq:w-solution}
		 W(\xi)\equiv W(\xi;\Omega)=\left(\phi_1^2(\xi,t)-\psi_1^2(\xi,t)\right)\exp(-\Omega t),\qquad \Omega=8\ii \lambda y,
	\end{equation} where $\phi_1$ and $\psi_1$ are given in \eqref{eq:Lax pair-1} and \eqref{eq:Phi}.
\end{lemma} 

\begin{proof}
	 By the stationary zero curvature equation of the matrix $\mathbf{L}(\xi,t;\lambda)$, we get
	\begin{equation}\label{eq:g,h}
		\begin{split}
			 g_t(\xi,t)=&\left(-\partial_{\xi\xi\xi} +(2s_2-6u^2)\partial_{\xi}-6uu_{\xi} \right)g(\xi,t)-6uu_{\xi}h(\xi,t),  \\
			 h_t(\xi,t)=&\left(-\partial_{\xi\xi\xi} +(2s_2-6u^2)\partial_{\xi}-6uu_{\xi}\right)h(\xi,t)-6uu_{\xi}g(\xi,t),  \\
		\end{split}
	\end{equation} 
	where $g(\xi,t)=\phi_1^2(\xi,t)$ and $h(\xi,t)=-\psi_1^2(\xi,t)$, which implies that $g(\xi ,t)+h(\xi,t)$ is a solution of linearized mKdV equation \eqref{eq:linearized-mkdv}. Combining $g(\xi,t)$, $h(\xi,t)$ with $\phi_1(\xi,t)$, $\psi_1(\xi,t)$, we could get that the function $g(\xi,t)+h(\xi,t)$ can be decomposed by separation of variables, which implies \eqref{eq:w-solution}. Corresponding to the dn-type and cn-type solutions, expressions $y$ are given by
	\begin{equation}
		\sqrt{(\lambda^2-\lambda_2^2)(\lambda^2-\lambda_3^2)}  \qquad \text{and} \qquad \sqrt{(\lambda^2-\lambda_1^2)(\lambda^2-\lambda_1^{*2})},
	\end{equation}
respectively, where $\lambda_i$,$i=1,2,3$ are given in \eqref{eq:z-l-lambda}. 
	
	We first consider the case of $\cn$-type solutions. The square of $\Omega$ could be written as
	\begin{equation}\label{eq:Omega-lambda}
		\Omega^2=(8\ii \lambda y)^2=-64\lambda^2(\lambda^2-\lambda_1^2)(\lambda^2-\lambda_1^{*2}).
	\end{equation}
	Let $F(\lambda)=-64\lambda^2(\lambda^2-\lambda_1^2)(\lambda^2-\lambda_1^{*2})-\Omega^2$. By the resultant of the function $F(\lambda)$ and its derivative $F'(\lambda)$, $\mathcal{R}(F(\lambda),F'(\lambda))=0$, we obtain five different zeros of $\Omega$: $0,\pm\Omega_0,\pm\Omega_0^*$.  And then, we prove the claims by the following three cases.
	\begin{enumerate}
		\item When $\Omega$ satisfies $\mathcal{R}(F(\lambda),F'(\lambda))\ne 0$, we get six different solutions $\pm \hat{\lambda}_i,i=1,2,3$ of \eqref{eq:Omega-lambda}. Since $\Omega(-\hat{\lambda}_i) =-8\ii \hat{\lambda}_i y(\hat{\lambda}_i)=-\Omega(\hat{\lambda}_i),i=1,2,3$, without loss of generality, we assume $\Omega(\hat{\lambda}_1)=\Omega(\hat{\lambda}_2)=\Omega(\hat{\lambda}_3)$ with $\hat{\lambda}_i\neq\hat{\lambda}_j$, $i\neq j, i,j=1,2,3$. Combining the equation $\psi_1=r_1\phi_1$ in \eqref{eq:r_1-phi-psi} together with the function $\phi_1$ in \eqref{eq:Phi}, we get that the function $W(\xi; \Omega)$ could be written as 
		\begin{equation}
			W(\xi;\Omega)=\phi_{1}^2(\xi,t;\lambda)(1-r_1^2(\xi,t;\lambda))\ee^{-\Omega t}
			=-\frac{\mu(\xi)(2u^2(\xi)-\beta_1-\beta_2)+\lambda(\beta_1-\beta_2)}{\lambda-\mu(\xi)}\ee^{2\theta_1-\Omega t}.
		\end{equation}		
		Thus, for different values of $\lambda$, the function $\lambda-\mu(\xi)$ has different zero points in the complex $\xi$ plane, which implies that functions $W_i(\xi;\Omega(\hat{\lambda}_i))$, $(i=1,2,3)$ have different singularity points in the complex $\xi$ plane. So, we obtain that functions $W_i(\xi;\Omega(\hat{\lambda}_i))$, $(i=1,2,3)$ are linearly independent.
		\item When $\Omega=0$, the solutions of \eqref{eq:Omega-lambda} are $0,\pm \lambda_1,\pm \lambda_1^*$. Plugging them into $W(\xi;\Omega)$ in \eqref{eq:w-solution}, we get five solutions with different values of $\lambda$. If $\lambda=0(z=0,z\in S)$, we get
		\begin{equation}W_1(\xi;0)
			=\frac{\alpha^2}{\vartheta_3^2\vartheta_4^2\left(\frac{\alpha\xi}{2K}\right)}\left(\vartheta_2^2\vartheta_1^2\left(\frac{\alpha\xi}{2K}\right)-\vartheta_4^2\vartheta_3^2\left(\frac{\alpha\xi}{2K}\right)\right).\end{equation}
		When $\lambda=\pm \lambda_1,\pm \lambda_1^*$, we set the corresponding function as $W_i(\xi;0), i=2,3,4,5$ respectively. And we can prove that $W_2(\xi;0)$ and $W_4(\xi;0)$ are linearly independent, in Lemma \ref{lemma:W2=W3}. Moreover, we know that $\frac{2K}{\alpha}$ is not the period of functions $W_2(\xi;0)$ and $W_4(\xi;0)$, but it is the period of function $W_1(\xi;0)$, which infers that functions $W_1(\xi;0)$, $W_2(\xi;0)$, and $W_4(\xi;0)$ are linearly independent.
		\item When $\Omega=\pm\Omega_0,\pm \Omega_0^*$, we only consider $\Omega=\Omega_0$, since the other cases can be analyzed similarly. We could set the roots of $\Omega_0=8\ii \lambda y$ are $\hat{\lambda}_1=\hat{\lambda}_2\ne\hat{\lambda}_3$. The spectral problem \eqref{eq:spectral} could be rewritten as
		\begin{equation}\label{eq:spectral-linear-ODE}
			\begin{bmatrix}
				W \\W_{\xi} \\W_{\xi \xi} 
			\end{bmatrix}_{\xi} =
			\begin{bmatrix}
				0 & 1 & 0 \\
				0 & 0 & 1 \\
				-12uu_{\xi}-\Omega & 2s_2-6u^2 & 0\\
			\end{bmatrix} 
			\begin{bmatrix}
				W \\W_{\xi} \\W_{\xi \xi} 
			\end{bmatrix},\end{equation}
		where $W=W(\xi;\Omega)$. Then, the fundamental solution matrix of the above differential equation is
		\begin{equation}
				\mathbf{W}^{[n]}(\xi;\Omega)=\mathbf{W}(\xi;\Omega) \mathbf{W}^{-1}(0;\Omega), \qquad
				\mathbf{W}(\xi;\Omega)=\begin{bmatrix}
					W_1(\xi;\Omega) & W_2(\xi;\Omega) & W_3(\xi;\Omega) \\
					W_{1,\xi}(\xi;\Omega) & W_{2,\xi}(\xi;\Omega) & W_{3,\xi}(\xi;\Omega) \\
					W_{1,\xi \xi} (\xi;\Omega) & W_{2,\xi \xi} (\xi;\Omega) & W_{3,\xi \xi} (\xi;\Omega)
				\end{bmatrix} 
		\end{equation} 
		with $\mathbf{W}^{[n]}(0;\Omega)=\mathbb{I}$. By the Abel theorem, we get that $\det(\mathbf{W}^{[n]}(\xi;\Omega))=1$. Thus, three linearly independent solutions of \eqref{eq:spectral} can be obtained by taking the limits $\Omega\to\Omega_0$: ${W}^{[n]}_{11}(\xi;\Omega_0)$, ${W}^{[n]}_{12}(\xi;\Omega_0)$, ${W}^{[n]}_{13}(\xi;\Omega_0)$, where ${W}^{[n]}_{ij}(\xi;\Omega_0)$ denotes the $(i,j)$ element of $\mathbf{W}^{[n]}(\xi;\Omega_0)$. 	\end{enumerate}
	
	Similar to the $\cn$-type solutions, we could also find three linearly independent solutions for the $\dn$-type solutions. Based on the above analysis, all linearly independent solutions to \eqref{eq:spectral} could be constructed by $W(\xi;\Omega)$ in \eqref{eq:w-solution}, so all eigenfunctions of spectral problem \eqref{eq:spectral} could be obtained.
\end{proof}

In \eqref{eq:y_mu_lambda}, \eqref{eq:lambda^2}, and \eqref{eq:Omega-lambda}, the function $\Omega(z)$ could be rewritten as the Jacobi elliptic function form:
\begin{equation}\label{eq:Omega(z)=}
	\Omega(z)=2\ii \alpha\frac{\dd\lambda^2(z)}{\dd z}
	=\alpha^3\left( k^2\scd(\ii (z-l))+\frac{k'^2\scd(\ii(z+l))}{\cn^4(\ii(z+l))}\right),
\end{equation}
where $\scd(z)$ is defined in \eqref{eq:define-scd}.

In the reference \cite{ChenP-18-mKdV,Deconinck-10}, the $\cn$-type solutions of the focusing mKdV equation are not spectrally stable with respect to arbitrary perturbations (amplitude). In the following, we aim to analyze the stability property of subharmonic perturbations, which is a particular perturbation with integer times of the period for solutions.

\subsection{Subharmonic stability analysis of the mKdV equation}\label{subsec:subharmonic stability analysis}

The goal of this subsection is to discuss the subharmonic stability analysis of the function $W(\xi;\Omega)$ associated with the values of functions $\Omega(z)$, $I(z)$ and $M(z)$ defined in \eqref{eq:Omega(z)=}, \eqref{eq:I}, and \eqref{eq:M}, respectively. In particular, we should pay attention to the boundedness of $W(\xi;\Omega)$.

\begin{definition}\label{defin:P-sub} 
		For the elliptic function solutions $u(\xi)$ with period $2T$, if the perturbation of this solution is $2PT$ periodic function, it is called a P-subharmonic perturbation of solution $u(\xi)$. If the period of perturbations is the same as the solution $u(\xi)$, we call it co-periodic perturbation.
	\end{definition}
	
	Combining Definition \ref{define:spect-stable} with Definition \ref{defin:P-sub}, we obtain the definition of subharmonic perturbations.
	
	\begin{definition}\label{define:spect-P}
		 If the perturbation $W(\xi;\Omega)$ is $2PT$ periodic function and $\Omega\in \ii \mathbb{R}$, i.e., the spectrum $\sigma(\mathcal{JL})$ satisfies 
		\begin{equation}
			\sigma_P(\mathcal{JL}):=\{\Omega\in \mathbb{C}| W(\xi;\Omega)\in C^0_b(\mathbb{R})\cap L^2([-PT,PT]) \} \subset \ii \mathbb{R},
		\end{equation} 
	then the solution $u(\xi)$ is P-subharmonic perturbation spectrally stable.
\end{definition}

Based on the Floquet theorem (Theorem in \cite{DeconinckK-06,Floquet-83}), we know that the solution $W(\xi;\Omega)$ in the linear homogeneous differential equation \eqref{eq:spectral-linear-ODE} are of the form
$W(\xi;\Omega)=\ee^{\ii \hat{\eta} \xi} \hat{W}(\xi;\Omega), \hat{W}(\xi+2T;\Omega)=\hat{W}(\xi;\Omega),\hat{\eta}\in \mathbb{C}$, where $2T$ is the period of the function $\hat{W}(\xi ;\Omega )$. Since the spectral problem \eqref{eq:spectral} is equivalent to \eqref{eq:spectral-linear-ODE}, every bounded solution of spectral problem \eqref{eq:spectral} is of the form
\begin{equation}\label{eq:W-W-eta}
	W(\xi ;\Omega )=\ee^{\ii \eta \xi  } \hat{W}(\xi ;\Omega ),\qquad \hat{W}(\xi  +2T;\Omega)=\hat{W}(\xi ;\Omega ),  \qquad \eta\in  \left[-\frac{\pi }{2T},\frac{\pi }{2T}\right). 
\end{equation}
Based on Definition \ref{defin:P-sub}, for the $2PT$-subharmonic perturbation problems,
	$\eta$ can be defined in any interval of length $\frac{2\pi}{2T}$, i.e.,
\begin{equation}\label{eq:eta}
	\eta=\frac{m\pi}{2PT}+\frac{(2n+1)}{2T}\pi ,\qquad m=-P,-P+1,\cdots,P-1, \quad \text{and} \quad n\in \mathbb{Z}.
\end{equation} 
By \eqref{eq:Phi-Theta}, \eqref{eq:w-solution}, and \eqref{eq:W-W-eta}, we get
\begin{equation} 
	\begin{split}
		\exp{(2\ii \eta T)}=&\frac{W(\xi+2T;\Omega)}{W(\xi;\Omega)}
		=\exp{\left(4\alpha Z(\ii(z-l))T+4\ii \lambda T \right)}.
	\end{split}
\end{equation}
And then, we define the function $M(z)$ as
\begin{equation}\label{eq:M}
	M(z):=2\eta T=-4\ii \alpha Z(\ii  (z-l))T+4\lambda(z) T.
\end{equation}
Together with \eqref{eq:eta}, the $2PT$-subharmonic perturbation problems must satisfy $M(z)=\frac{n\pi }{P},n\in \mathbb{Z}$. 

From Lemma \ref{lemma:W} and the spectral problem \eqref{eq:spectral}, we know that only when the real part of the exponent is zero, i.e.,
\begin{equation}\label{eq:Re}
	{\Re}(\alpha Z(\ii (z-l))+\ii \lambda)=0, \quad \text{and} \quad {\Re} (-\alpha Z(\ii  (z+l)+K+\ii K')+\ii \lambda )=0, 
\end{equation}
the solution $W(\xi;\Omega)$ is bounded. We find the relationship between eigenfunctions of the spectral problem and solutions of the Lax pair in Lemma \ref{lemma:W}. 
The linear combinations of equations in \eqref{eq:Re}:
\begin{equation}\label{eq:I-1-2}
	{\Re}(\alpha Z(\ii (z-l))-\alpha Z(\ii  (z+l)+K+\ii K')+2\ii \lambda)=0 \text{  and  }
	{\Re}\left(\alpha Z(\ii (z-l))+\alpha Z(\ii  (z+l)+K+\ii K')\right)=0,
\end{equation} 
are equivalent to \eqref{eq:Re}. By \eqref{eq:Phi}, we get that the determinant of matrix $\Phi(\xi,t;\lambda)$ is a constant. Together with $\eqref{eq:Phi-Theta}$, the first one of equation \eqref{eq:I-1-2} holds. Therefore, for $\alpha \in \mathbb{R}$, we just need to analyze $ {\Re}(I(z))=0$, where
\begin{equation} \label{eq:I}
	I(z):= Z(\ii (z-l))+ Z(\ii  (z+l)+K+\ii K').
\end{equation} 

Similar to the literature \cite{DeconinckS-17}, differentiating with respect to $z_{R},z_{I}$ on the curve ${\Re}\left(I(z) \right)=C$, we could get the tangent vector 
\begin{equation}\label{eq:tangent-vector}
	\left(-\frac{\dd {\Re} (I)}{\dd z_{I}}, \frac{\dd {\Re} (I)}{\dd z_{R}} \right)= \left({\Im}(I'(z)), {\Re}(I'(z)) \right), \qquad I'(z):=\frac{\dd I(z)}{\dd z}
\end{equation}   
where $C$ is a constant and $z_{R},z_{I}$ denote the real and imaginary part of $z$ respectively. Once we find a point $z$ satisfying ${\Re}(I(z))=0$, we could get a curve, in which all the points $z$ satisfy ${\Re}(I(z))=0$ by the tangent vector \eqref{eq:tangent-vector}. The derivative of $I(z)$ is
\begin{equation}\label{eq:I'}
	I'(z)
	=\ii \left( \dn^2(\ii  (z-l))+ \dn^2(\ii  (z+l)+K+\ii K')-\frac{2E(k)}{K(k)} \right).
\end{equation}
By \eqref{eq:lambda} and the definition of $M(z)$ in \eqref{eq:M}, we obtain $M'(z):=\frac{\dd M(z)}{\dd z}
	=-2\ii\alpha T I'(z)$,
which implies $M(z)=-2\ii \alpha T I(z)+C, C\in \mathbb{R}$. Substituting $z=0$ into the above equations, we get
\begin{equation}\label{eq:M(z)-value}
	M(z)=-2\ii \alpha T I(z)+\pi, \quad l=0,\quad \text{and} \quad
	M(z)=-2\ii \alpha T I(z)+2\pi, \quad l=\frac{K'}{2}.
\end{equation}
And we consider the value $z$ in the rectangular area $S$, where the set $S$ is defined in \eqref{eq:set S}. Using the formulas of the Zeta function \cite[p.33]{ByrdF-54}, we obtain that when $l=0$, the periods of function ${\Re}(I(z))$ are $2K'$ and $K'+\ii K$; when $l=\frac{K'}{2}$, the periods of function ${\Re}(I(z))$ are $2K'$ and $\ii K$. Thus, for any $\hat{z}\in \mathbb{C}$, we can find a point $z\in S$, such that ${\Re}(I(\hat{z}))={\Re}(I(z))$.
For the boundedness of the function $W(\xi;\Omega)$, we merely need to consider the set $Q$ defined in \eqref{eq:set Q}. By the expression of $I(z)$ in \eqref{eq:I}, we get the feature about it:

\begin{prop}\label{prop:Q_r-z=l}
	For the set $Q$, we get the following propositions:
	\begin{itemize}
		\item[{\rm (1)}] If $Q_r=\left\{z\left|z\in \mathbb{R}, z \in S\right. \right\}$, we get $Q_r\subseteq Q$ and $\Omega(z)\in \ii \mathbb{R},z\in Q_r$.
		\item[{\rm (2)}] The set $Q$ is symmetric about the line $z=l$ and the line ${\Im}(z)=0$.
	\end{itemize}
\end{prop}
\begin{proof}
	(1): By the function $I(z)$ in \eqref{eq:I} and formulas of the Zeta function \cite[p.34]{ByrdF-54}, for any $z\in Q_r$, we get
	\begin{equation}
		I^*(z)
		=Z(-\ii(z-l))+Z(-\ii(z+l)+K-\ii K')
		=-Z(\ii(z-l))-Z(\ii(z+l)+K+\ii K')
		=-I(z),
	\end{equation}
	which implies $I(z) \in \ii \mathbb{R}$, so we get $Q_r \subseteq Q$. By \eqref{eq:Omega(z)=}, we obtain $\Omega^*(z)=-\Omega(z)$, i.e., $\Omega(z)\in \ii \mathbb{R}$.
	
	(2): We set two points $\tilde{z}_{1,2}=\pm \tilde{z}+l$ that are symmetric about the line $z=l$. The values of $I(\tilde{z}_{1,2})$ are
	\begin{equation}\begin{split}
			I(\tilde{z}_1)&=I(\tilde{z}+l)
			=Z(\ii \tilde{z})+Z(\ii \tilde{z}+K+\ii K'+2\ii l)
			=Z(\ii \tilde{z})+Z(\ii \tilde{z}+K-\ii K'+2\ii l)+ \frac{\ii \pi}{K},\\
			I(\tilde{z}_2)&=I(-\tilde{z}+l)
			=Z(-\ii \tilde{z})+Z(-\ii \tilde{z}+K+\ii K'+2\ii l)
			=-Z(\ii \tilde{z})-Z(\ii \tilde{z}+K-\ii K'-2\ii l).
	\end{split}\end{equation}
	Letting $I(\tilde{z}_1)\in \ii\mathbb{R}$, we know $I(\tilde{z}_2)=-I(\tilde{z}_1)+\frac{\ii \pi}{K}\in \ii\mathbb{R},l=0$ and  $I(\tilde{z}_2)=-I(\tilde{z}_1)\in \ii\mathbb{R},l=\frac{K'}{2}$. So, we get that $Q$ is symmetric about the line $z=l$.
	By the equation
	\begin{equation}\begin{split}
			I(z^*)&=Z(\ii(z^*-l))+Z(\ii(z^*+l)+K+\ii K')=-Z^*(\ii(z-l))-Z^*(\ii(z+l)-K+\ii K')=-I^*(z),
	\end{split}\end{equation}
	we obtain that the set $Q$ is symmetric about the line ${\Im}(z)=0$.
\end{proof}

\begin{lemma}\label{lemma:M increase}
	Along the curve ${\Re}(I(z))=0$, the value of $M(z)$  increases (decreases)	in the upper half-plane, and it decreases (increases) in the lower half-plane.
\end{lemma}

\begin{proof}
	By \eqref{eq:M(z)-value}, the directional derivative of $M(z)$ along the curve ${\Re}(I(z))=0$ is given by:
	\begin{equation}\label{eq:muz}
		\begin{split}
			\left( \frac{\dd  M(z)}{\dd  z_R},\frac{\dd  M(z)}{\dd  z_I} \right) \cdot \left( {\Im} (I'(z)),{\Re}(I'(z)) \right) 
			=&2\alpha T \left( \frac{\dd  {\Im}(I)}{\dd  z_R},\frac{\dd  {\Im}(I)}{\dd  z_I} \right) \cdot \left( {\Im} (I'(z)),{\Re}(I'(z)) \right)
			\\ =&2\alpha T  \left( \left({\Im}(I'(z))\right)^2+\left({\Re}(I'(z)) \right)^2 \right),
		\end{split}
	\end{equation}
	where $z=z_R+\ii z_I, z_I,z_R\in \mathbb{R}$ and $z\in Q$ in \eqref{eq:set Q}. Since the directional derivative of $M(z)$ with respect to $z$ is nonzero along the curve ${\Re}(I(z))=0$, the value of $M(z)$ is increasing or decreasing. By the symmetry of the curve ${\Re}(I(z))=0$ in Proposition \ref{prop:Q_r-z=l}, we get that if the value of $M(z)$ increases (decreases) in the upper half-plane, it decreases (increases) in the lower half-plane.
\end{proof}

For the different solutions of the mKdV equation, we divide their spectral stability analysis into two subsections.

\subsection{Spectral stability of $\dn$-type solutions}\label{subsec:dn}

In this subsection, we analyze the condition for the spectral stability of the $\dn$-type solutions, i.e., $l=\frac{K'}{2}$.
\begin{lemma}\label{lemma:dn-I(z)}
	If $z=(2m-1)\frac{K'}{2}+\ii z_{I}\in S,z_{I}\in \mathbb{R}, m=0,1,2$, with $z_{I}\neq\frac{nK}{2},n=0,\pm 1$, then $I(z)\notin \ii \mathbb{R}$.
\end{lemma}
\begin{proof}	
	Plugging $z=\frac{(2m-1)K'}{2}+\ii z_{I}$ into \eqref{eq:I} and utilizing formulas of the Zeta function \cite[p.33]{ByrdF-54}, we get 
		\begin{equation}
			\begin{split}
			I(z)+I^*(z)
				=&Z(-z_{I}+\ii(m-1) K')+Z(-z_{I}+\ii(m+1) K'+K)\\
				&+Z(-z_{I}-\ii (m-1)K')+Z(-z_{I}-\ii(m+1) K'+K)\\
				=&Z(-z_{I}-\ii (m+1)K')-\frac{\ii \pi m}{K} +Z(-z_{I}+\ii(m+1)K'+K)\\
				&+Z(-z_{I}-\ii(m+1) K')-\frac{\ii \pi}{K}+Z(-z_{I}+\ii(m+1) K'+K)+\frac{\ii\pi(m+1) }{K}\\
				=&2Z(-z_{I}-\ii(m+1) K')+2Z(-z_{I}+\ii(m+1) K'+K).
			\end{split}
		\end{equation}
		Since $Z(u)$ is an odd function, we get that if $-z_{I}-\ii(m+1) K'=-(-z_{I}+\ii(m+1) K'+K)+2nK,n\in \mathbb{Z}$, the equation $I^*(z)+I(z)=0$ holds. By $z=\frac{(2m-1)K'}{2}+\ii z_I\in S$, we get $z_{I}=0,\pm \frac{K}{2}$.
\end{proof}

\begin{lemma}\label{lemma: dn z}
	If $l=K'/2$, the set $Q$ \eqref{eq:set Q} could be rewritten as \begin{equation}
		Q=Q_0:=\left\{z\left|z=z_{R}+\frac{\ii }{2}nK\in S,\ n\in \mathbb{Z},\,\, z_{R}\in\mathbb{R} \right. \right\}.
	\end{equation}
	Moreover, for any $z\in Q$, $\Omega(z)\in \ii \mathbb{R}$.
\end{lemma}
\begin{proof}
	The condition of $z\in \mathbb{R}$ has been proved in Proposition \ref{prop:Q_r-z=l}. We consider $z=z_{R}\pm\frac{\ii K}{2}\in Q$. Plugging $z=z_{R}\pm \frac{\ii K}{2}$ into \eqref{eq:I} and utilizing formulas of the Zeta function \cite[p.33]{ByrdF-54}, we get
	\begin{equation}\label{eq:Q-I}
		\begin{split}
			I(z)+I^*(z)
			=&Z\!\left(\ii\left(z_{R}+\frac{3K'}{2}\right)\mp\frac{K}{2}\right)+Z\!\left(\ii\left(z_{R}-\frac{K'}{2}\right)\mp\frac{K}{2}+K\right)\\
			&-Z\!\left(\ii\left(z_{R}-\frac{K'}{2}\right)\pm \frac{K}{2}\right)-Z\!\left(\ii\left(z_{R}+\frac{3K'}{2}\right)\pm\frac{K}{2}-K\right)
			=0.
	\end{split}\end{equation}
	Therefore, we obtain the set $Q_0\subseteq Q$.
	
	Assuming $z_0 \in Q$ but $z_0 \notin Q_0$, we get a curve $l_1$ which goes through $z_0$ and satisfies ${\Re}(I(z))=0$ by the tangent vector. From the definition of the Zeta function and $I(z)$, we know that $I(z)$ only has first-order poles, which implies that only one curve satisfying ${\Re}(I(z))=0$ goes through the poles. Because the pole point is in the set $Q_0$ and for any $z\in Q_0$ the inequality $I'(z)\neq 0$ holds, and the curve $l_1$ does not intersect with the set $Q_0$. By Lemma \ref{lemma:dn-I(z)}, we find that on the boundary of the set $S$, if $\Re(I(z))=0$, the point $z$ must satisfy $z\in Q_0$. So, the curve does not intersect with the boundary.
	Thus the curve $l_1$ is a closed one. In the interior of a closed curve, by the maximum principle of harmonic function, we know that all the points $z$ satisfy ${\rm Re}(I(z))=0$, so $I'(z)=0$ in this closed region. However, there are only two points such that $I'(z)=0, z\in Q$, so we get the contradiction. Therefore, $Q \subseteq Q_0$.
	
	Finally, plugging $z=z_{R}\pm \frac{\ii K}{2}$ into \eqref{eq:Omega(z)=} and using the shift formulas of the Jacobi elliptic functions \cite[p.20]{ByrdF-54}, we get
	\begin{equation}\begin{split}
			\Omega^*(z)
			=&\alpha^3\left( k^2\scd\left(\ii z_{R}- \frac{\ii K'}{2}\pm\frac{K}{2} \right)+\frac{k'^2\scd\left(\ii z_{R}+ \frac{\ii K'}{2}\pm\frac{K}{2} \right)}{\cn^4\left(\ii z_{R}+ \frac{\ii K'}{2}\pm\frac{K}{2} \right)} \right)\\
			=&\alpha^3\left(k^2\scd\left(\ii z_{R}- \frac{\ii K'}{2}\mp\frac{K}{2} \right) +\frac{k'^2\scd\left(\ii z_{R}+ \frac{\ii K'}{2}\mp\frac{K}{2} \right)}{\cn^4\left(\ii z_{R}+ \frac{\ii K'}{2}\mp\frac{K}{2} \right)}\right)\\
			=&-\Omega(z),
	\end{split}\end{equation}
where the function $\scd(z)$ is defined in \eqref{eq:define-scd}. Together with Proposition \ref{prop:Q_r-z=l}, we verify $\Omega(z)\in \ii \mathbb{R}$, $z\in Q$. 
\end{proof}

\newenvironment{proof-spec-dn}{\emph{Proof of Theorem \ref{theorem:spec-dn}.}}{\hfill$\Box$\medskip}
\begin{proof-spec-dn}
	Lemma \ref{lemma: dn z} claims that the set corresponding to all bounded spectral functions of the mKdV equation with the $\dn$-type solutions is $Q_0$, and all elements of $Q_0$ satisfy $\Omega(z) \in \ii \mathbb{R}$. 
	By Definition \ref{define:spect-stable}, the $\dn$-type solutions of the mKdV equation \eqref{eq:mKdV1} are spectrally stable. 
\end{proof-spec-dn}

By choosing parameters $k=0.9975$, $\alpha=\frac{1}{16}$, we exhibit the set $Q$, functions $\lambda(z)$ and $\Omega(z),z\in Q$, in Figure \ref{Fig dn}.

\begin{figure}[ht]
	\centering
	\includegraphics[width=1\linewidth]{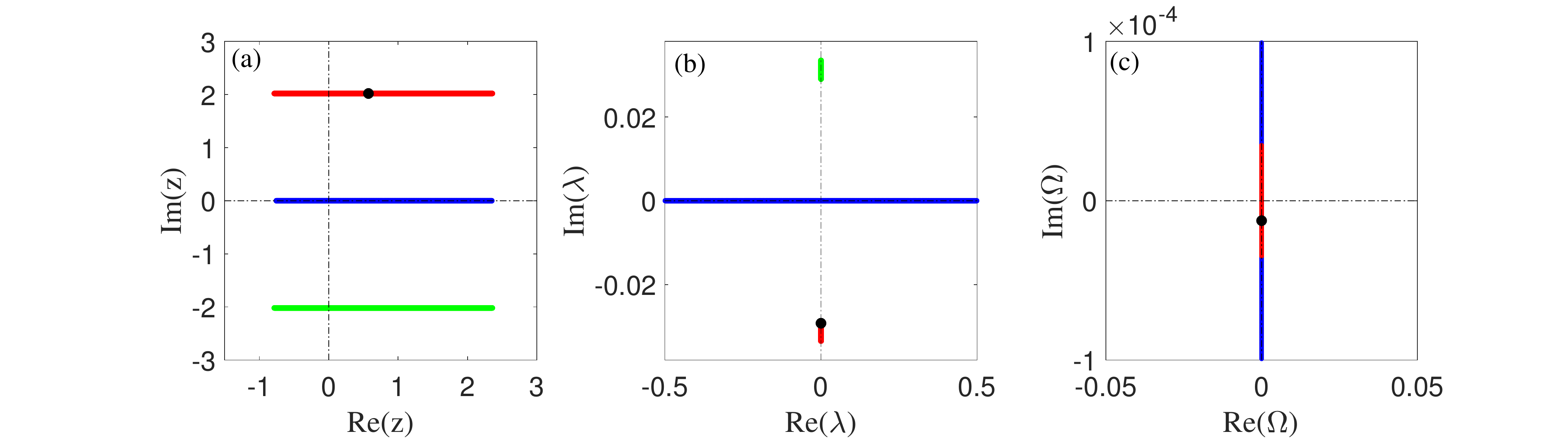}
	\caption{(a) \(\{z| z\in Q \}\), (b)  \(\{\lambda(z) | z\in Q \}\), (c)  \(\{\Omega(z) | z\in Q \}\). In subfigure (c), the green line coincides with the red line.  The black points represent the spectral points of the corresponding breather solutions by Darboux-B\"acklund transformation.}
	\label{Fig dn}
\end{figure}

\subsection{Spectral stability of $\cn$-type solutions} \label{subsec:cn}
In this subsection, we mainly study the spectral stability of the $\cn$-type solutions, i.e., $l=0$, with respect to the subharmonic perturbations.

\begin{prop}\label{prop:Q-M-cn}
	Setting $z_{1,2}=\frac{K'}{2}\pm \ii \frac{K}{2}, z_{3,4}=-\frac{K'}{2}\pm \ii \frac{K}{2}$, we obtain $\rm (a) :$ $z_i\in Q$; $\rm (b):$ $\Omega(z_i)=0$; and $\rm (c):$ $M(z_i)=\pi $ mod $2\pi$, $i=1,2,3,4$.
\end{prop}
\begin{proof}
	(a): From the definition of $I(z)$ in \eqref{eq:I}, we get
	\begin{equation}\label{eq:Z-value 1}
			I(z_3)
			=Z\!\left(-\frac{K}{2}-\frac{\ii K'}{2}\right)+Z\!\left(\frac{K}{2}+ \frac{\ii K'}{2}\right)=0,\qquad
			I(z_4)
			=Z\!\left(\frac{K}{2}-\frac{\ii K'}{2}\right)+Z\!\left(\frac{3K}{2}+ \frac{\ii K'}{2}\right)
			=0.
	\end{equation}
	By \eqref{eq:I} and formulas of the Zeta function \cite[p.34]{ByrdF-54}, we obtain 
	\begin{equation}\label{eq:I(-z),I(z)}
			I(-z)
			=-Z(\ii z)- Z(\ii z+K+\ii K'-2K-2\ii K')
			=-Z(\ii z)- Z(\ii z+K+\ii K')-\frac{\ii \pi}{K}
			=-I(z)-\frac{\ii \pi}{K}.
	\end{equation}
	Combining \eqref{eq:Z-value 1} with \eqref{eq:I(-z),I(z)}, we get 
	\begin{equation}\label{eq:I-z_34}
			I(z_1)=I(z_4)-\frac{\ii \pi}{K}
		=-\frac{\ii \pi}{K}, \qquad I(z_2)=-I(z_3)-\frac{\ii \pi}{K}
		=-\frac{\ii \pi}{K}.
	\end{equation}
	Thus, $I(z_i)\in \ii \mathbb{R},i=1,2,3,4$, which implies  $z_i\in Q,i=1,2,3,4$.

	(b): Since $\cn\left(z_{1,2}\right)=\cn\left(z_{3,4}\right)=(1\mp \ii )\sqrt{\frac{k'}{2k}}$ in \cite[p.21]{ByrdF-54}, we get $k^2\cn^4(\ii z_i)+k'^2=0,i=1,2,3,4$. By \eqref{eq:Omega(z)=}, we obtain	
	\begin{equation}\label{eq:Omega(z)}
		\begin{split}
			-\Omega(z_i)
			=&\alpha^3\left( k^2\scd(\ii z_i)+\frac{k'^2\scd(\ii z_i)}{\cn^4(\ii z_i)} \right)
			=\alpha^3(k^2\cn^4(\ii z_i)+k'^2)\frac{\scd(\ii z_i)}{\cn^4(\ii z_i)}=0.
		\end{split}		
	\end{equation}
	
	(c): By \eqref{eq:M(z)-value}, \eqref{eq:Z-value 1}, and \eqref{eq:I-z_34}, it is easy to obtain $M\left(z_1\right)=M\left(z_2\right)=-\pi$ and $M\left(z_4\right)=M\left(z_3\right)=\pi$.
\end{proof}

\begin{remark}\label{remark:Omega-sym}
	The curve ${\Re}(\Omega(z))=0$ is also symmetric about the origin point, lines ${\Im}(z)=0$ and ${\Re}(z)=0$. Since $\sn(z)$ is an odd function and $\cn(z)$ and $\dn(z)$ are even functions, together with \eqref{eq:Omega(z)=}, we obtain 
	\begin{equation}\nonumber
			\begin{split}
				\Omega(-z)
				=-\alpha^3\left( k^2\scd(\ii z)+\frac{k'^2\scd(\ii z)}{\cn^4(\ii z)} \right)=\Omega(z),\quad
				\Omega(z^*)
				=\alpha^3\left( -k^2\scd(\ii z)-\frac{k'^2\scd(\ii z)}{\cn^4(\ii z)} \right)^*=\Omega^*(z),
			\end{split}
	\end{equation}
where $\scd(z)$ is defined in \eqref{eq:define-scd}. Thus, if $z_0$ satisfies ${\Re}(\Omega(z_0))=0$, points $z_0^*,-z_0,-z_0^*$ also satisfy ${\Re}(\Omega(-z_0))={\Re}(\Omega(-z_0^*))={\Re}(\Omega(z_0^*))=0$. 
\end{remark}

\begin{lemma} \label{lemma:I-z}
On the boundary of the set $S$, the values of the function $I(z)$ have the following properties:
	\begin{itemize}
		\item[{\rm (a)}] On the lines $z=z_R\pm \ii \frac{K}{2},z\in S$, only four points $z_1,z_2,z_3,z_4$ satisfy ${\Re}(I(z))=0$, i.e, $\{z|z=z_R \pm \ii \frac{K}{2}\}\cap Q=\{z_1,z_2,z_3,z_4\}$.
		\item[{\rm (b)}] If $z=\pm K'\pm \ii z_I, z_I\neq 0, z\in S$, we have ${\Re}(I(z))\neq 0$.
	\end{itemize}
\end{lemma}	
\begin{proof}	
	We first consider the first quadrant of the set $S$, called $S_1$. By the symmetry of the set $Q$ proved in Proposition \ref{prop:Q-M-cn}, the computations of values $z\in S$ in the second, third, and fourth quadrants are the same as $z\in S_1$. 
	
	(a): Utilizing the derivative formulas \cite[p.25]{ByrdF-54} and the half arguments formulas \cite[p.24]{ByrdF-54} of Jacobi elliptic functions in turn, we could rewrite the function \eqref{eq:I'} as 
	\begin{equation}\label{eq:I'-simplify}
		I'(z)=\ii \left( \dn^2(\ii z)+\dn^2(\ii z+K+\ii K')-\frac{2E}{K} \right)
		=\ii \left( \frac{2\cn(2\ii z)}{1+\cn(2\ii z)}-\frac{2E}{K} \right).
	\end{equation}
	Plugging $z=z_R+\ii \frac{K}{2}$ into \eqref{eq:I'-simplify}, using shift formulas \cite[p.20]{ByrdF-54} and imaginary arguments formulas \cite[p.24]{ByrdF-54}, in turn, we get
	\begin{equation}
			I'\left( z_R+\ii \frac{K}{2} \right)
			=\frac{2\ii k'\sn(2\ii z_R)}{\dn(2\ii z_R)+k'\sn(2\ii z_R)}-\frac{2\ii E}{K}
			=\ii \left( k'^2\sn^2(2 z_R,k')-\frac{2E}{K} \right)-2k'\sn(2z_R,k')\dn(2  z_R,k').
	\end{equation}
	Thus, for all $z_R\in [0,K']$, ${\Re}\left(I'\left( z_R+\ii \frac{K}{2} \right)\right)<0$, which implies that on the line $z=z_R+\ii \frac{K}{2}$, the value of ${\Re}(I(z))$ is decreasing. Since ${\Re}\left(I\left(\frac{K'}{2}+ \frac{\ii K}{2}\right)\right)=0$, we get ${\Re}(I(z))\neq 0$, when $z=z_R+\ii \frac{K}{2}\in S_1, z_R\neq \frac{K'}{2}$.
	
	(b): Substituting $z=K'+\ii z_I,z_I\neq 0$ into \eqref{eq:I'-simplify}, the derivative of $I(z)$ with respect to $z$ is 
	\begin{equation}
		\begin{split}
			I'\left( K'+\ii z_I \right)
			=&\ii \left( \frac{2\cn(-2 z_I+2\ii K')}{1+\cn(-2 z_I+2\ii K')}-\frac{2E}{K} \right)
			=\ii \left( \frac{-2\cn(2 z_I)}{1-\cn(2 z_I)}-\frac{2E}{K} \right).				
		\end{split}
	\end{equation}
	Since $\cn(2z_I)\in[0,1), z\in(0,\frac{K}{2}]$, we get that for all $z_I\in (0,\frac{K}{2}]$, ${\Im}\left(I'\left(\frac{K'}{2}+\ii z_I\right)\right)<0$, which implies that on the line $z_R=K',z_I\neq 0$ the value of ${\Re}(I(z))$ is monotonous. By \eqref{eq:I}, we get $I(0)=Z(0)+Z(K+\ii K')=-\frac{\ii \pi}{2K}\in \ii \mathbb{R}$.
	Therefore, on the line $z_R=K', z_I>0,z=z_R+\ii z_I\in S$, we have ${\Re}(I(z))\neq 0$. 
\end{proof}

\begin{prop}\label{prop: zr0 zi0} 
	By \eqref{eq:I} and \eqref{eq:I'}, the following properties hold:
	\begin{itemize}
		\item[{\rm (a)}]  $I'(z)|_{z=z_i}\not\in \ii \mathbb{R},i=1,2,3,4$, where $z_{1,2}=\frac{K'}{2}\pm \ii \frac{K}{2}, z_{3,4}=-\frac{K'}{2}\pm \ii \frac{K}{2}$.
		\item[{\rm (b)}] If $\frac{2E(k)}{K(k)}>1$, then the set $Q$ intersects with the real axis at point $z_c \in \mathbb{R}$; if $\frac{2E(k)}{K(k)}\leq 1$, then the set $Q$ intersects with the imaginary axis at $z_{c} \in \ii \mathbb{R}$.  The set $Q$ consists of the real line and two curves.
	\end{itemize} 
\end{prop}
\begin{proof}	
	(a): Plugging $z=z_i$ $(i=1,2,3,4)$ into \eqref{eq:I'}, we obtain 
	\begin{equation}\label{eq:I(z)d}
		I'(z)|_{z=z_{1,3}} =\ii \left( 2k'^2-2\ii kk'+1-\frac{2E(k)}{K(k)}\right)\not\in \ii \mathbb{R},
		\quad I'(z)|_{z=z_{2,4}}=\ii \left( 2k'^2+2\ii kk'+1-\frac{2E(k)}{K(k)}\right)\not\in \ii \mathbb{R}.
	\end{equation} 
	
	(b): By imaginary argument formulas of the Jacobi elliptic functions \cite[p.24]{ByrdF-54}, we rewrite \eqref{eq:I'-simplify} as 
		\begin{equation}\label{eq:I'zl}
			\begin{split}
				I'(z)
				&=\ii \left( \frac{\cn(2\ii z,k)-1}{\cn(2\ii z,k)+1}+1-\frac{2E}{K} \right)
				=\ii \left( \frac{1-\cn(2z,k')}{1+\cn(2 z,k')}+1-\frac{2E}{K} \right).
			\end{split}
	\end{equation} 
	On the real axis, the function $1-\cn(2z,k')\in [0,2]$ is monotonically increasing for $z\in[0,K']$ and the function $1+\cn(2z,k')\in [0,2]$ is monotonically decreasing for $z\in[0,K']$. Thus, the function $\Im\left(I'(z)\right)$ is monotonically increasing for $z\in[0,K']$. By \eqref{eq:I'zl} and the second equation of \eqref{eq:value-E-K-lin}, we know $\Im\left(I'(0)\right)=1-\frac{2E(k)}{K(k)}$ and $\Im\left(I'\left(\frac{K'}{2}\right)\right)=2-\frac{2E(k)}{K(k)}>0$. Combined with the monotonicity of function $\Im\left(I'(z)\right)$, if $1-\frac{2E(k)}{K(k)}\leq 0$, there exists a unique point $0\leq \Re(z_c)<\frac{K'}{2}$, $z_c\in \mathbb{R}$ such that $I'(z_c)=0$ by the zero point theorem.
	If $1-\frac{2E(k)}{K(k)}>0$, the function $I'(z)$ has no zero in the real axis. By $\frac{\dd (\Re(I(z)))}{\dd z_I}|_{z=0}=-\Im (I'(0))=\frac{2E}{K}-1<0$ and  $\frac{\dd\Re(I(z))}{\dd z_I}|_{z=\frac{K}{2}}=\frac{2E}{K}>0$, there exists a unique point $z=z_{0}$ such that $\frac{\dd\Re(I(z))}{\dd z_I}|_{z=z_{0}}=0$ due to the monotonicity of $\Im(I'(z))$ with respect to the imaginary axis. Thus we know that for $\Im(z)\in(0,z_{0})$, the function $\Re(I(z))$ is a decreasing function with respect to $\Im(z)$ along the imaginary axis which implies that $\Re(I(z_{0}))<0$ since $\Re(I(0))=0$. Since $\Re \left(I\left(\frac{\ii K}{2}\right)\right)=k'>0$, we get that there exists a unique point $z_c\in \ii \mathbb{R}, z_{0}<\Im(z_c)<\frac{K}{2}$ such that $\Re(I(z_c))=0$ by the zero point theorem and the monotonicity of the function $\Re(I(z))$ with respect to the imaginary axis.
	
	We proceed to examine all possibilities for the components of the set $Q$. The curve $l_1\in Q$ ends at $z$ satisfying $I'(z)=\infty$ or the boundary of the set $S$
	and crosses to another component at $z$ with $I'(z)=0$. If the spectrum contains a closed curve, the cross point satisfies ${\Re}(I(z))=0$. In the interior of a closed curve, by the maximum value principle of the harmonic function, we have ${\Re}(I(z))=0$. Then $I(z)$ is a constant in this closed region.  However, this is impossible. Thus there is no closed curve with ${\Re}(I(z))=0$. Furthermore, by \eqref{eq:Z-value 1} and \eqref{eq:I-z_34}, we know ${\Re}(I(z_i))=0, i=1,2,3,4$. By the implicit function theorem, we know that there exist four curves with ${\Re}(I(z))=0$ to the harmonic function ${\Re}(I(z))$ departing from the points $z_i, i=1,2,3,4$ due to \eqref{eq:I(z)d}.
	
	We consider the case: $z_c\in\mathbb{R}$. Since $z\in\mathbb{R}$, we have ${\Re}(I(z))=0$. Especially, we have ${\Re}(I(\pm z_c))=0$. Furthermore, by $I'(z)|_{z=\pm z_c}=0$ and $I''(z)|_{z=\pm z_c}\neq 0$, then in the neighborhood of $z=\pm z_c$, we have Taylor expansions $I(z)=I(\pm z_c)+I''(\pm z_c)(z\pm z_c)^2+\mathcal{O}((z-z_c)^3)$. By the localized analysis and implicit function theorem, we find two curves ${\Re}(I(z))=0$ departing from the point $z=\pm z_c$. In the boundary of $S$, we get six points $z=\pm K'$ and $z=z_i$, $i=1,2,3,4$, which can emit the curves with ${\Re}(I(z))=0$. Similarly, by the localized analysis, on the point $z=\pm K'$, we find that only one curve emitting from it exists. And we know that the real axis goes through them. Thus the curve is the real axis.  Therefore, we conclude that the curve departing from the point $z=z_1$ goes across $z=z_c$ and ends with $z_2$, and another curve departing from the point $z=z_3$ goes across $z=-z_c$ and ends with $z_4$.
	
	Then, we consider the case: $z_c\in \ii\mathbb{R}$. Similar to the above analysis, we conclude that there are two curves emitting from $z=z_{1,3}$ that go across the imaginary axis at $\pm z_c\in \ii \mathbb{R}$ and end with $z=z_{2,4}$, respectively. Together with the first property of Proposition \ref{prop:Q_r-z=l}, we obtain that the set $Q$ consists of a real line and two curves.
\end{proof}

\begin{remark}\label{remark:z_c}
	When $1-\frac{2E(k)}{K(k)}\le0$, by $I'(z_{c})=0$, we get
	\begin{equation}
		0=\dn^2(\ii z_{c})+\dn^2(\ii z_{c}+K+\ii K')-\frac{2E(k)}{K(k)}
		=\frac{2\cn(2\ii z_{c})}{1+\cn(2\ii z_{c})}-\frac{2E(k)}{K(k)},
	\end{equation}
	which implies $\cn(2\ii z_{c})=\frac{E}{K-E}$. Then $z_{c}=-\frac{\ii}{2}F\left(\sin^{-1}\left( \sqrt{\frac{K(K-2E)}{(K-E)^2}}\right),k \right)$. When $\frac{2E(k)}{K(k)}\ge 1$, we get $\sqrt{\frac{K(K-2E)}{(K-E)^2}} \in \ii \mathbb{R}$, which means that $F\left(\sin^{-1}\left( \sqrt{\frac{K(K-2E)}{(K-E)^2}}\right),k \right)\in \ii \mathbb{R}$, so we have $z_c \in \mathbb{R}$. 	
\end{remark}

Define the function $\Lambda(z)$ as
\begin{equation}\label{eq:Lambda}
	\Lambda(z):=\frac{4}{\alpha^2}\lambda^2(z)=\dn^2(\ii z)+ \dn^2(\ii z+K+\ii K')-1,
\end{equation}
where $\lambda^2(z)$ is defined in \eqref{eq:lambda^2}.
By Proposition \ref{prop:lambda-comforming}, Lemma \ref{lemma:map}, functions \eqref{eq:lambda} and \eqref{eq:lambda^2}, the function $\Lambda(z)$ \eqref{eq:Lambda} maps the rectangular region $S_1$ onto a whole upper half plane, i.e., $\Lambda=\Lambda_R+\ii \Lambda_I,\Lambda_I>0$, where $S_1$ is defined as the first quadrant of $S$. Plugging $z=z_1:=\frac{K'}{2}+\frac{\ii K}{2}$ and $z=z_c$ into \eqref{eq:Lambda}, we get $\Lambda(z_1)=a+\ii b$, $a=1-2k^2$, $b=2k\sqrt{1-k^2}$, and $\Lambda(z_c)=\frac{2E}{K}-1$.
\begin{figure}[h]
	\centering
	\includegraphics[width=0.9\linewidth]{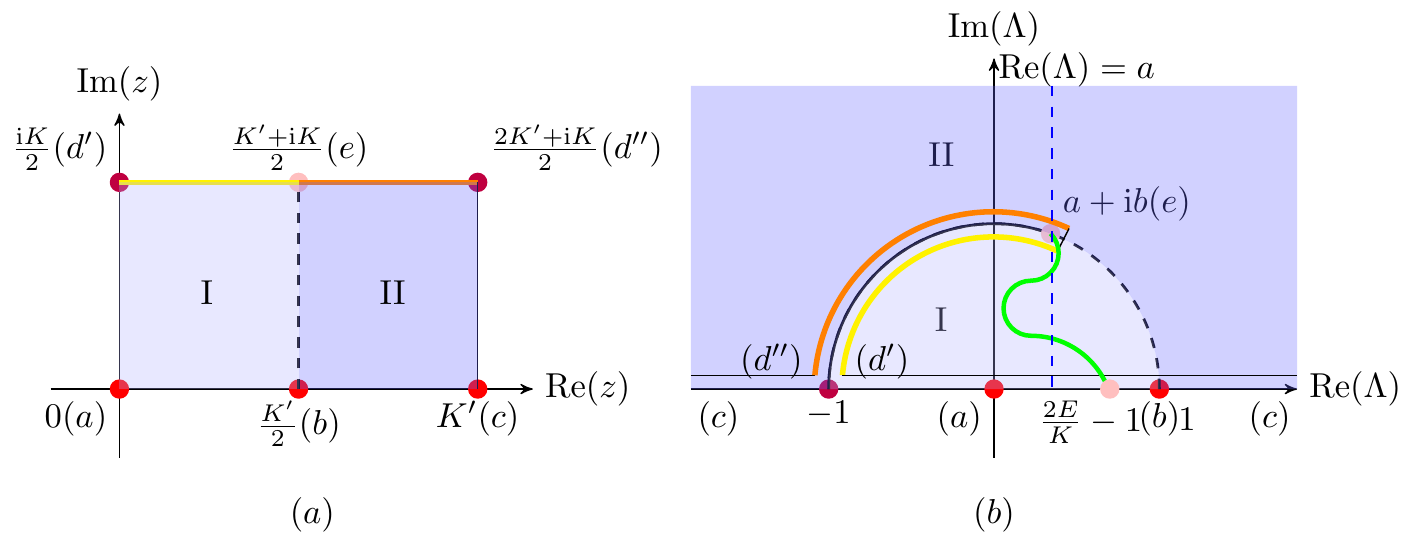}
	\caption{(a):$\{z|z\in S_1\}$; (b):$\{\Lambda(z)|z\in S_1\}$. $S_1$ is a subset of set $S$ in the first quadrant.}
	\label{fig:fig-lambda-z}
\end{figure}

\begin{remark}\label{remark:I-Lambda}
	By Lemma \ref{lemma:I-z}, we know that if $z=z_R+\ii \frac{K}{2}, z_R\neq \frac{K'}{2}$ or $z=K'+\ii z_I, z_I\neq 0,\frac{K}{2}$, ${\Re}(I(z))\neq 0$. Based on the function $\Lambda(z)$ and the inverse function $z(\Lambda)$, we know that the inequality ${\Re}(I(z(\Lambda)))\neq 0$ holds for all $\Lambda_R, \Lambda_I$ satisfying $\Lambda_R^2+\Lambda_I^2=1\,\, (\Lambda_R\neq a,\Lambda_I\neq b)$. 
\end{remark}

\begin{lemma}\label{lemma:z R O}
	For any $z \in Q \backslash (\mathbb{R}\cup \{ z_i| \ i=1,2,3,4\})$, we have $\Omega(z) \notin \ii  \mathbb{R}$.
\end{lemma}

\begin{proof}
	Without loss of generality, we consider $z\in S_1$ (the first quadrant of $S$). By the symmetry of the curve ${\Re}(\Omega(z))=0$ and the set $Q$, shown in Remark \ref{remark:Omega-sym} and Proposition \ref{prop:Q-M-cn}, respectively, the computations of values $z\in S$ in the second, third and fourth quadrants are the same as $z\in S_1$.
	
	We mainly prove that curves ${\Re}(I(z(\Lambda)))=0$ and ${\Re}(\Omega(z(\Lambda)))=0$ in the $\Lambda$-plane do not intersect. By \eqref{eq:Omega-lambda} and \eqref{eq:Lambda}, the real and imaginary parts of the function $\Omega^2$ are
	\begin{equation}\label{eq:Omega-Im_Re}
		\left\{\begin{aligned}
			{\Im}(\Omega^2)&=-\alpha^6\Lambda_I\left(3\Lambda_R^2-4a\Lambda_R-\Lambda_I^2+1\right),\\
			{\Re}(\Omega^2)&=-\alpha^6\left(\Lambda_R(\Lambda_R^2-2a\Lambda_R-3\Lambda_I^2+1)+2a\Lambda_I^2\right).
		\end{aligned} \right.
	\end{equation}
	The necessary and sufficient conditions of ${\Re}(\Omega)=0$ on $\Lambda$-plane are ${\Im}(\Omega^2)=0$ and ${\Re}(\Omega^2)\le 0$. Combined with \eqref{eq:Omega-Im_Re}, in the $\Lambda$-plane, the curve ${\Re}(\Omega)=0$ satisfying $\lambda(z)\not\in \mathbb{R}$ is equivalent to $\Lambda_I^2=3\Lambda_R^2-4a\Lambda_R+1, \Lambda_R\le a$. By \eqref{eq:Omega(z)=} and \eqref{eq:Lambda}, the function $\Omega(z)$ could be written by the derivative of $\Lambda(z)$ as $\Omega(z)= \frac{\ii\alpha^3}{2}\Lambda'(z)=-\frac{\alpha^3}{2}{\Im}\left(\Lambda'(z)\right)+\frac{\ii\alpha^3}{2}{\Re}\left(\Lambda'(z)\right)$, $\Lambda'(z):=\frac{\dd \Lambda(z)}{\dd z}$,
	which leads to
	\begin{equation}\label{eq:Omega^2-Lambda-z}
		\Omega^2(z)=\frac{\alpha^6}{4}\left[{\Im}^2\left(\Lambda'(z)\right)-{\Re}^2\left(\Lambda'(z)\right)- 2\ii \, {\Re}\left(\Lambda'(z)\right){\Im}\left(\Lambda'(z)\right)\right].
	\end{equation}
	
	Considering the curve ${\Re}(I(z(\Lambda)))=0$ in the $\Lambda$-plane, we aim to prove that the curve ${\Re}(I(z(\Lambda)))=0$ is in the region $\Lambda_R\ge a$, i.e., the curve ${\Re}(I(z(\Lambda)))=0$ is on the right side of the blue dashed line in Figure \ref{fig:fig-lambda-z}. By Proposition \ref{prop: zr0 zi0}, we know that the curve ${\Re}(I(z))=0$ in $z$-plane has a continuous curve on the region $z\in S_1$ with two end points $z=z_c$ and $z=\frac{K'}{2}+\ii \frac{K}{2}$. By the conformal mapping between $\Lambda$ and $z$, there is a curve in the $\Lambda$-plane with two end points $\Lambda(z_c)=\frac{2E}{K}-1$ and $\Lambda(z_1)=a+\ii b$. Furthermore, by \eqref{eq:value-E-K-lin-a}, we know that the point $\left(\frac{2E}{K}-1,0\right)$ is on the right side of the line $\Lambda_R=a$.

	In other words, we aim to prove that for any point $(\Lambda_R,\Lambda_I), \Lambda_R<a, \Lambda_I>0$, the inequality ${\Re}(I(z(\Lambda)))\neq0$ holds. Firstly, we introduce some formulas that are useful in the following analysis. Secondly, we study the derivative of the point $(\Lambda_R,\Lambda_I)=(a,b)$ to obtain the variation of the curve ${\Re}(I(z(\Lambda)))=0$. At last, we prove the statement by contradiction.
	By \eqref{eq:I'} and \eqref{eq:Lambda}, along the curve ${\Re}(I(z(\Lambda)))=0$, the tangent vector could be written as
	\begin{equation}\label{eq:tan-vec-Lambda-I}
		\left(-\frac{\dd {\Re}(I)}{\dd \Lambda_I},\frac{\dd {\Re}(I)}{\dd \Lambda_R} \right)=\left({\Im}\left(\frac{\dd I}{\dd \Lambda}\right),{\Re}\left(\frac{\dd I}{\dd \Lambda} \right)\right),
	\end{equation}
	where
	\begin{equation}\label{eq:d-I-Lambda}
		\begin{split}
			\frac{\dd I}{\dd \Lambda}=\frac{\dd I}{\dd z}\cdot \frac{\dd z}{\dd \Lambda}
			=\ii\left(\Lambda(z)+1-\frac{2E}{K}\right)\frac{\dd z}{\dd \Lambda}.
		\end{split}
	\end{equation}
	Since $z(\Lambda)$ is the inverse function of $\Lambda(z)$, the derivative of $z(\Lambda)$ could be obtained by function $\Lambda(z)$ as
	\begin{equation}\label{eq:inverse-z-Lambda}
		\frac{\dd z}{\dd \Lambda}=\frac{1}{\Lambda'(z)}=\frac{{\Re}\left(\Lambda'(z)\right)-\ii {\Im}\left(\Lambda'(z)\right)}{|\Lambda'(z)|^2}.
	\end{equation}

	Then, we study the derivative of $I(z(\Lambda))$ with respect to $\Lambda$ on the line $\Lambda_R=a$. Plugging $\Lambda_R=a$ into \eqref{eq:Omega-Im_Re} and \eqref{eq:Omega^2-Lambda-z}, we can get
	\begin{equation}\label{eq:Re-Im-Omega}
		-\frac{\left[{\Im}\left(\Lambda'(z)\right)\right]^2-\left[{\Re}\left(\Lambda'(z)\right)\right]^2}{2  {\Re}\left(\Lambda'(z)\right){\Im}\left(\Lambda'(z)\right)}
		=\frac{{\Re}(\Omega^2)}{{\Im}(\Omega^2)}
		=\frac{-\alpha^6a(b^2-\Lambda_I^2)}{-\alpha^6\Lambda_I(b^2-\Lambda_I^2)}=\frac{a}{\Lambda_I}.
	\end{equation}
	When $\Lambda_I\in(0,b)$, ${\Im}(\Omega^2)=-\alpha^6\Lambda_I(b^2-\Lambda_I^2)<0$. By \eqref{eq:Omega^2-Lambda-z}, we can get ${\Im}\left(\Lambda'(z)\right){\Re}\left(\Lambda'(z)\right)>0$.
	Furthermore, solving the quadratic equation formulated by the first and the last equality in \eqref{eq:Re-Im-Omega} with respect to $\frac{{\Im}\left(\Lambda'(z)\right)}{{\Re}\left(\Lambda'(z)\right)}$ and combining with \eqref{eq:inverse-z-Lambda}, we get 
	\begin{equation}\label{eq:Im-Re-z-Lambda}
		-\frac{{\Im}\left(\frac{\dd z}{\dd \Lambda}\right)}{{\Re}\left(\frac{\dd z}{\dd \Lambda}\right)}
		\xlongequal{\eqref{eq:inverse-z-Lambda}}
		\frac{{\Im}\left(\Lambda'(z)\right)}{{\Re}\left(\Lambda'(z)\right)}
		\xlongequal[\eqref{eq:Omega^2-Lambda-z}]{\eqref{eq:Re-Im-Omega}}
		-\frac{a}{\Lambda_I}+\sqrt{\left(\frac{a}{\Lambda_I}\right)^2+1}.
	\end{equation}
	Plugging \eqref{eq:Im-Re-z-Lambda} into \eqref{eq:d-I-Lambda}, elements of tangent vector \eqref{eq:tan-vec-Lambda-I} are 
	\begin{equation}
		{\Re}\left(\frac{\dd I}{\dd \Lambda}\right)
		=-\frac{1}{\Lambda_I}{\Re}\left(\frac{\dd z}{\dd \Lambda}\right)\cdot\left(\Lambda_I^2+\left(\frac{2E}{K}-1-a\right)\cdot\left( -a+\sqrt{a^2+\Lambda_I^2}\right)\right),
	\end{equation}
	and 
	\begin{equation}\label{eq:d-I-d-Lambda}
		{\Im}\left(\frac{\dd I}{\dd \Lambda}\right)
		={\Re}\left(\frac{\dd z}{\dd \Lambda}\right)\cdot\left(\sqrt{a^2+\Lambda_I^2}+1-\frac{2E}{K}\right).
	\end{equation}
	By \eqref{eq:d-I-Lambda}, \eqref{eq:d-I-d-Lambda} and \eqref{eq:value-E-K-lin}, we get 
	\begin{equation}
		\lim_{\Lambda_I\rightarrow b}\frac{{\Re}\left(\frac{\dd I}{\dd \Lambda}\right)}{	{\Im}\left(\frac{\dd I}{\dd \Lambda}\right)}=\frac{-\frac{2E}{K}\cdot\left( 1-a\right)}{2b\left(1-\frac{E}{K}\right)}<0.
	\end{equation}
	Combined with the variation of the curve ${\Re}(I(z))=0$ in the $z$-plane in Proposition \ref{prop: zr0 zi0}, the variation of curve ${\Re}(I(z(\Lambda)))=0$ at the point $(\Lambda_R,\Lambda_I)=(a,b)$ is that $\Lambda_R$ increases and $\Lambda_I$ decreases, which satisfies $\Lambda_R^2+\Lambda_I^2<1$. 
	
	By Remark \ref{remark:I-Lambda}, we know that the curve ${\Re}(I(z(\Lambda)))=0$ does not cross the circle $\Lambda_I^2+\Lambda_R^2=1$, excepting point $(\Lambda_R,\Lambda_I)=(a,b)$. Thus, if there exists a point $\Lambda$ satisfying $\Lambda_R<a$ on the curve ${\Re}(I(z(\Lambda_R,\Lambda_I)))=0$, as the green curve is shown in Figure \ref{fig:fig-lambda-z}, there are at least three points in line $\Lambda_R=a, \Lambda_I\in(0,b]$ such that ${\Re}(I(z(\Lambda)))=0$, i.e., the equation ${\Re}(I(z(\Lambda(a,\Lambda_I))))=0, \Lambda_I\in (0,b]$ has at least three different solutions. Thus, by Lagrange's mean value theorem, the function ${\Re}(I(z(\Lambda(a,\Lambda_I))))$ has at least two extreme points on the line $\Lambda_R=a$ and $\Lambda_I\in (0,b)$, i.e., $\frac{\dd  {\Re}(I(z(\Lambda(a,\Lambda_I))))}{\dd \Lambda_I}=-{\Im}\left(\frac{\dd I}{\dd \Lambda}\right)|_{\Lambda=a+\ii \Lambda_I}$ has at least two zeros. However, by \eqref{eq:Omega-Im_Re} and $a^2+b^2=1$, we get ${\Im}(\Omega^2)=-2  {\Re}\left(\Lambda'(z)\right){\Im}\left(\Lambda'(z)\right)=-\alpha^6\Lambda_I(b^2-\Lambda_I^2)\neq 0 $ for $\Lambda_R=a$ and $\Lambda_I\in (0,b)$. So we know ${\Re}\left(\Lambda'(z)\right)\neq 0$ for $\Lambda_R=a$ and $\Lambda_I\in(0,b)$, which further implies  ${\Re}\left(\frac{\dd z}{\dd \Lambda}\right)\neq 0$ by \eqref{eq:inverse-z-Lambda}. By \eqref{eq:d-I-d-Lambda}, we find that the function $\sqrt{a^2+\Lambda_I^2}+1-\frac{2E}{K}$ at most has one zero as $\Lambda_I\in(0,b)$. Thus, $\frac{\dd  {\Re}(I(z(\Lambda(a,\Lambda_I))))}{\dd \Lambda_I}=-{\Im}\left(\frac{\dd I}{\dd \Lambda}\right)|_{\Lambda=a+\ii \Lambda_I}$ has at most one zero for $\Lambda_R=a, \Lambda_I\in (0,b)$. So, we get the contradiction. Therefore, we prove that the curve ${\Re}(I(z(\Lambda)))=0$ on the $\Lambda$-plane satisfies the condition $\Lambda_R\ge a$. 
	
	Since the value of $\Lambda$ must satisfy $\Lambda_R\le a$ for ${\Re}(\Omega(z(\Lambda)))=0$,  and on the line $\Lambda_R=a$, we can verify that there only exists one point $(a,b)$ such that ${\Re}(\Omega(z(\Lambda)))=0$. Thus two curves ${\Re}(\Omega(z(\Lambda)))=0$ and ${\Re}(I(z(\Lambda)))=0$ only have one intersecting point $(a,b)$ on the $\Lambda$-plane with $\Lambda_I>0$. Therefore, in the $z$-plane, excepting $z=\frac{K'}{2}+\ii \frac{K}{2}\in S_1$, there does not exist any other intersecting points satisfies ${\Re}(\Omega(z(\Lambda)))=0$ and ${\Re}(I(z(\Lambda)))=0$ by the conformal transformation. 
	
	Similar conclusions can be obtained in the second, third and fourth quadrants. Thus, for any $z \in Q \backslash (\mathbb{R}\cup \{ z_i| \ i=1,2,3,4\})$, we have $\Omega(z) \notin \ii  \mathbb{R}$.
\end{proof}

The spectral stability with respect to the subharmonic perturbations of period $2PT$ is that all eigenvalues $\Omega$ of $2PT$ periodic function $W(\xi;\Omega)$ satisfying \eqref{eq:spectral} are imaginary, i.e., $\Omega(z)\in \ii \mathbb{R}$. Combining \eqref{eq:M(z)-value} with \eqref{eq:eta}, we set 
\begin{equation}\label{eq:Q_P}
	Q_P:=\left\{ z\in Q|M(z)=\frac{\pi}{P}m+(2n+1)\pi, \quad m=-P,\cdots,P-1, \quad n \in \mathbb{Z}\right\},
\end{equation} 
which contains the conditions of $z$ deriving all $2PT$ periodic functions. When for any $z\in Q_P$, the value $\Omega(z)\in \ii \mathbb{R}$, the corresponding solution is spectrally stable with respect to perturbations of period $2PT$. The set $Q_P$ could also be divided into two subsets $Q_P=Q_{P,R}\cup Q_{P,C}$, where 
\begin{equation}\label{eq:Q_Psub}
	Q_{P,R}:=\{z|z\in \mathbb{R},z\in Q_P\}, \qquad Q_{P,C}=\{z|z\notin \mathbb{R},z\in Q_P\}.
\end{equation}

\newenvironment{proof-spec-cn}{\emph{Proof of Theorem \ref{theorem:subharmonic-1-}.}}{\hfill$\Box$\medskip}
\begin{proof-spec-cn}
	By Definition \ref{define:spect-P}, to prove the spectral stability of the $\cn$-type solutions with the $P$-subharmonic perturbation, we should get the value of $P$ for all $z\in Q_P$, $\Omega(z)\in \ii \mathbb{R}$.
	By Proposition \ref{prop:Q_r-z=l}, we get $\Omega(z)\in \ii \mathbb{R}$ for any $z\in Q_{P,R}$. From Lemma \ref{lemma:z R O}, we know that for $z\in Q \backslash \mathbb{R}$,  $\Omega(z)\in \ii \mathbb{R}$ only if $z=z_i,$ $i=1,2,3,4$.  Thus, the spectral stability is converted into prove $Q_{P,C}=\{z_1,z_2,z_3,z_4\}$. We divided the proof into the following two categories for different conditions of the set $Q$ in Proposition \ref{prop: zr0 zi0}.
	
	When $\frac{2E}{K}\ge 1$ (denotes this case as type-I), by the symmetry of the set $Q$ and the function $\Omega(z)$, we need to study the case of $z\in S_1$. Since along the curve ${\Re}(I(z))=0$ from $z=z_c$ to $z=z_1$, the value of $M(z)$ is decreasing by Lemma \ref{lemma:M increase}. From Proposition \ref{prop:Q-M-cn}, we get that $M(z_1)=-\pi$. We must ensure that no other point in $Q_P$ intersects with the curve $\Re(I(z))=0$ between $z=z_c$ and $z=z_1$. In other words, only if $M(z_c)\le -\frac{P-1}{P}\pi$, $Q_{P,C}\cap S_1=\{z_1\}$. Therefore, when $P\le \frac{\pi}{\pi+ M(z_c)}$, for any $z\in Q_P$, we get $\Omega(z)\in \ii \mathbb{R}$. The $\cn$-type solutions are spectrally stable with respect to perturbations of period $2PT, \ P\in \mathbb{N}$.
	
	When $\frac{2E}{K}<1$ (denotes this case as type-II), we could analyze the upper half-plane since the lower half-plane can be obtained similarly. From Proposition \ref{prop: zr0 zi0}, we know that there exists a curve connecting $z_1$ to $z_3$, satisfying ${\Re}(I)=0$. Since $M(z_3)=\pi$, $M(z_1)=-\pi$ (see Proposition \ref{prop:Q-M-cn}) and $M(z)$ is continuous and monotonous, only when $P=1$, the set $Q_{P,C}=\{z_1,z_2,z_3,z_4\}$ holds. So if $\frac{2E(k)}{K(k)}<1$, the $\cn$-type solutions are spectrally stable with respect to co-periodic perturbations but no other subharmonic perturbation. 	
\end{proof-spec-cn}

The above theorem shows that two types of the $\cn$-type solutions have different stability properties. Now, we illustrate this fact by plotting the corresponding figures of the spectrum. For the type-I, choosing $k=\frac{1}{4},\alpha=1$, it is shown that $u(\xi)=\frac{1}{4}\cn(\xi)$ is spectrally stable with respect to $3$-subharmonic perturbations (Figure \ref{fig: sub cn}). For the type-II, choosing $k=\frac{19}{20},\alpha=\sqrt{2}$, we can plot the corresponding spectrum of the linearized spectral problem, in which there is no multi-subharmonic perturbation (See Figure \ref{fig:cn sub point}).

\begin{figure}[ht]
	\centering
	\includegraphics[width=0.95\linewidth]{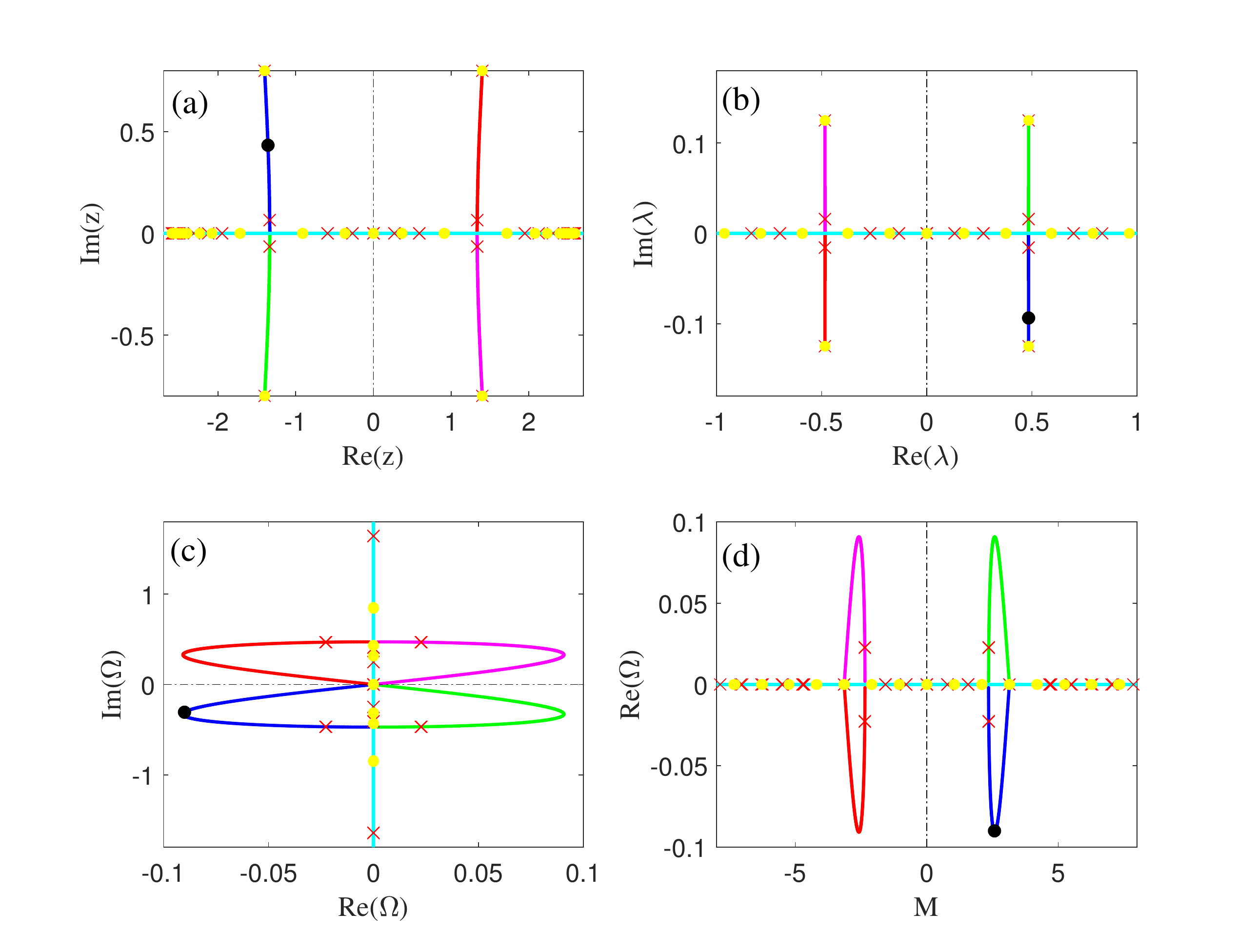}
	\caption{(a): $\{z| z\in Q \}$,  (b): $\{\lambda(z)| z\in Q \}$, (c): $\{\Omega(z)| z\in Q \}$,  (d): $\{(M(z), \Re(\Omega(z))) | z\in Q \}$. The red crosses in the figures denote the corresponding points of $W(\xi;\Omega)$ under the periodic perturbation of $4T$. The yellow points denote the corresponding points of $W(\xi;\Omega)$ under the periodic perturbation of $3T$. The black points denote the corresponding spectral parameters of the breather solutions constructed by the Darboux-B\"acklund transformation.}
	\label{fig: sub cn}
\end{figure}

\begin{figure}[ht]
	\centering
	\includegraphics[width=0.90\linewidth]{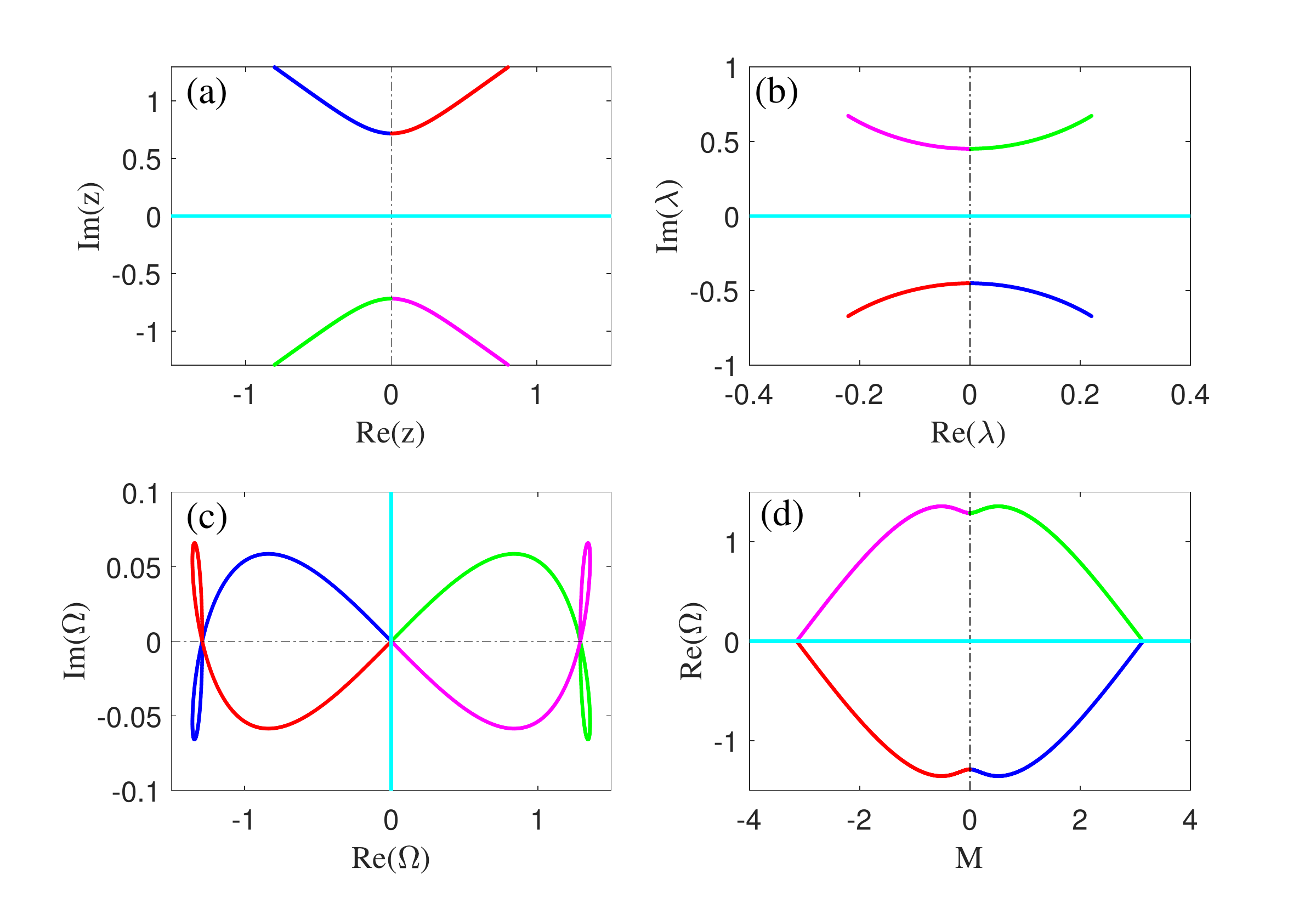}
	\caption{(a): $\{z| z\in Q \}$,  (b): $\{\lambda(z)| z\in Q \}$, (c): $\{\Omega(z)| z\in Q \}$,  (d): $\{(M(z), \Re(\Omega(z))) | z\in Q \}$.}
	\label{fig:cn sub point}
\end{figure}

Combining \eqref{eq:I} with \eqref{eq:M(z)-value}, we get that the function $M(z)=\pi-2\ii K\left(Z(\ii z)+Z(\ii z+K+\ii K')\right)$ is only related to the modulus $k$ and variable $z$. Since the value $z_{c}=-\frac{\ii}{2}F\left(\sin^{-1}\left( \sqrt{\frac{K(K-2E)}{(K-E)^2}}\right),k \right)$ is only dependent on the modulus $k$ from Remark \ref{remark:z_c}, $M(z_c)$ is only dependent on the modulus $k$. Thus, the region of value $P\le \frac{\pi}{\pi +M(z_c)}$ is only dependent on $k$. The value $\max(P)$ with respect to $k$ is plotted in Figure \ref{fig:subfig all condition}. The black point in Figure \ref{fig:subfig all condition} shows that the $\cn$-type solutions are $3$-subharmonic perturbations, not $4$-subharmonic perturbations, which is consistent with the results in Figure \ref{fig: sub cn}.

\begin{figure}[ht]
	\centering
	\includegraphics[width=0.850\linewidth]{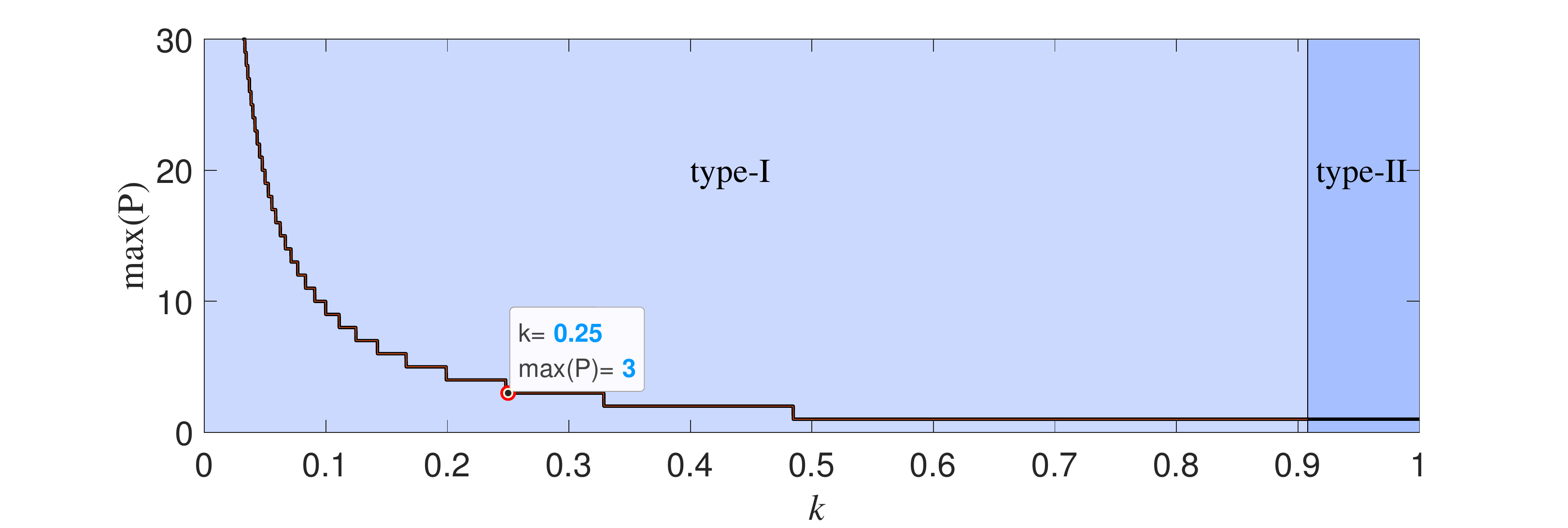}
	\caption{The maximum of value $P\in \mathbb{Z}$, where $P\le \frac{\pi}{\pi +M(z_c)}$. The red point denotes that when $k=0.25$, the maximum of value $P$ is $3$.}
	\label{fig:subfig all condition}
\end{figure}

\section{Orbital stability analysis}\label{section:orbial stability}

The previous section provides the conditions for the spectral stability of elliptic function solutions with respect to P-subharmonic perturbations. Based on them, we study the orbital stability of the cn-type and dn-type solutions in this section.

The orbital stability is characterized in terms of the spectrum of the second variation. Since the Krein signature can evaluate the second variation, we convert it to consider the Krein signature, which was used to establish the orbital stability of the periodic solutions in the defocusing mKdV equation \cite{Deconinck-10} and the cnoidal waves of the KdV equation \cite{DeconinckK-10}. To study the orbital stability, we elaborate some helpful information, including the higher-order conservation laws in Appendix \ref{appendix:Darboux transformation}, the infinite number of the Hamiltonian functional in \eqref{eq:H0}, the framework \cite{BonaSS-87,GrillakisSS-87,KapitulaP-13}, and so on.

The mKdV equation possesses an infinite number of conserved quantities (in Appendix \ref{appendix:Darboux transformation})
\begin{equation}\label{eq:H0}
	\mathcal{H}_0=\frac{1}{2}\int_{-PT}^{PT} u^2 \dd  x  ,\quad
	\mathcal{H}_1=\frac{1}{2}\int_{-PT}^{PT}\left(u_x^2-u^4\right) \dd  x , \quad
	\mathcal{H}_2= \frac{1}{2}\int_{-PT}^{PT} \left( u_{xx}^2-10u^2u_{x}^2+2u^6 \right)\dd x, \cdots
\end{equation} 
where the period of function $u$ is $2PT$. The conserved quantities $\mathcal{H}_0$ and $\mathcal{H}_1$ are known as moment and energy conservation, respectively. The Hamiltonian flows in the mKdV hierarchy are given by $u_{t_n}=\partial_x\mathcal{H}'_n(u)$, where the prime denotes the gradient of the Hamiltonian $\mathcal{H}_n$ with respect to $u$. The equation $u_{t_n}=\partial_x\mathcal{H}'_n(u),n=0,1,2$, is shown in \eqref{eq:ODE-t}. A linear combination of the above Hamiltonian to define the $n$-th mKdV equation with time variables $t_n$ under the moving coordinate form $(\xi,t_n)$ as
\begin{equation}\label{eq:H-mKdV}
	u_{t_n}=\mathcal{J}\hat{\mathcal{H}}'_n(u), \qquad
	 \hat{\mathcal{H}}_n:=\mathcal{H}_n+\sum_{i=0}^{n-1}c_{n,i}\mathcal{H}_i, \qquad \hat{\mathcal{H}}_0:=\mathcal{H}_0,
\end{equation}
where $c_{n,i}\in \mathbb{R},i=0,1,...,n-1$. The stationary solution of the $n$-th mKdV equation satisfies the ordinary differential equation $\mathcal{J}\hat{\mathcal{H}}'_n(u)=0$ in \eqref{eq:H-mKdV}. 

\begin{remark}	
		If $u$ is the stationary solution of equation $\mathcal{J}\hat{\mathcal{H}}'_1(u)=0$, then $u$ satisfies the equation $u_{\xi\xi\xi}+6u^2u_{\xi}-2s_2u_{\xi}=0.$ 	Differentiating both sides of the above equation, we get  
	$u_{\xi\xi\xi\xi\xi}+12u_{\xi}^3+36uu_{\xi}u_{\xi\xi}+6u^2u_{\xi\xi\xi}-2s_2u_{\xi\xi\xi}=0$. And integrating both sides of this equation, we obtain $u_{\xi\xi}+2u^3-2s_2u=\hat{c}_2$ and $ \frac{1}{2}u_{\xi}^2+\frac{1}{2}u^4-s_2u^2=\hat{c}_2 u+\hat{c}_1$, where $\hat{c}_1=(4s_4-s_2^2)/2$, $\hat{c}_2=0$. By the above equations, we find that the function $u$ with
	\begin{equation}\label{eq:c1c2}
		c_{2,0}=-4s_2^2+2c_{2,1}s_2+4\hat{c}_1, \qquad c_{2,1}\in \mathbb{R},
	\end{equation}
	satisfies the stationary equation $\mathcal{J}\hat{\mathcal{H}}'_2(u)=0$. Similarly, the function $u$ also satisfies the higher-order stationary equations $\mathcal{J}\hat{\mathcal{H}}'_n(u)=0, n=2,3,\cdots$.
\end{remark}

Based on the stationary solution $u$, we linearize the equations $u_{t_i}=\mathcal{J}\hat{\mathcal{H}}'_i(u),i=1,2,\cdots,n$ about $u$ with 
\begin{equation}\nonumber
	v(\xi,\mathbf{t})=u(\xi,\mathbf{t})+\epsilon w(\xi,\mathbf{t})+\mathcal{O}(\epsilon^2), \qquad \mathbf{t}=\left(t_1,t_2,\cdots,t_n\right),
\end{equation}
 and result in the linear system: $w_{t_i}=\mathcal{JL}_{i}w, i=1,2,\cdots,n$, where $\mathcal{L}_i$ is the variational derivative $\hat{\mathcal{H}}''_i$ evaluated at the stationary solution. Then, we obtain
\begin{equation} \label{eq:spectral Omega}
	\Omega_nW=\mathcal{JL}_nW,\qquad \Omega_n^*W^*=\mathcal{JL}_nW^*, 
\end{equation} 
where $W=W(\xi;\Omega_n)$.
\begin{definition}\label{defin:Krein}
	Krein signature is the sign of 
	\begin{equation}
		\mathcal{K}_n(z):=\left\langle W_n,\mathcal{L}_n W_n \right\rangle, \qquad \left\langle W_n,\mathcal{L}_n W_n \right\rangle=\int_{-PT}^{PT}W_n^*\mathcal{L}_n W_n \dd \xi,
	\end{equation}
	where $W_n=W(\xi;\Omega_n)$ is an eigenfunction of $n$-th mKdV equation \eqref{eq:spectral Omega}. The inner product is defined in the $L^2([-PT,PT])$ inner product space.
\end{definition}

When $W_1(\xi;\Omega_1)$ satisfies $\Omega_1W_1(\xi;\Omega_1)=\mathcal{JL}_1W_1(\xi;\Omega_1)$ with $\Omega_1\in \ii\mathbb{R}$, we consider the Krein signature $\mathcal{K}_1(z)$. We first study a special case that $\Omega_1=0$, i.e., $\lambda=0,\pm\lambda_1,\pm\lambda_1^*$. It is easy to know that when $\lambda=0$, the eigenfunction could be written as $W(\xi;0)=\partial_{\xi}u$, and the Krein signature is $\mathcal{K}_1=\left\langle\partial_{\xi}u,\mathcal{L}_1 \partial_{\xi}u \right\rangle=0$. When $\lambda=\pm\lambda_1,\pm\lambda_1^*$, by analyzing the exponent part of functions $\phi_i,\psi_i$, $i=1,2$, we know that the period of function $W(\xi;\Omega_1)$ is infinity. Now, we consider the value of $\mathcal{K}_1(z)$ when $\lambda\in \mathbb{R}$ and $\Omega_1\in \ii \mathbb{R}$. By Lemma \ref{lemma:W}, we know that $W(\xi;\Omega_1)=2\lambda(\phi_1^2-\psi_1^2)\exp(-\Omega_1 t),\lambda\in \mathbb{R}\backslash\{0\}$ is the eigenfunction of the linearized spectral problem \eqref{eq:spectral Omega} with the eigenvalue $\Omega_1$. By the matrix $\Phi(\xi,t;\lambda)=[\Phi_1\,\,\, \Phi_2]$ in \eqref{eq:Lax pair-1} and its Lax pair \eqref{eq:Lax-pair-1}, we get $\Phi_1^{\top}\Phi_{1,\xi}=\Phi_1^{\top}\mathbf{U}\Phi_1$, which implies $2\lambda(\phi_1^2-\psi_1^2)=\ii\partial_{\xi} (\phi_1^2+\psi_1^2)$. Thus, by $\Omega_1=8\ii \lambda y\in \ii \mathbb{R}$, we get
\begin{equation}\label{eq:WLW}
	\begin{split}
		W^*(\xi;\Omega_1)\mathcal{L}_1 W(\xi;\Omega_1)
		=\Omega_1 W^*(\xi;\Omega_1)\mathcal{J}^{-1} W(\xi;\Omega_1)
		=2\ii \lambda\Omega_1 (\phi_1^2+\psi_1^2)(\phi_1^{*2}-\psi_1^{*2}).
	\end{split}
\end{equation}	
By \eqref{eq:phi1,2-dn}, \eqref{eq:Re}, and \eqref{eq:I}, if and only if $\Re(I(z))=0$ and $\Re(\Omega(z))=0$, the exponential part of functions $\phi_1(\xi,t),\psi_1(\xi,t)$ is pure imaginary for all real variables $\xi,t\in \mathbb{R}$. Since $\lambda\in \mathbb{R}$ and $\Omega_1=8\ii \lambda y\in \ii \mathbb{R}$, we get $y\in \mathbb{R}$ and $\beta_{1,2}\in \mathbb{R}$. By \eqref{eq:beta1-beta2-theta-1-theta-2}, we obtain $\exp(\theta_1+\theta_1^*)=1$. Combining \eqref{eq:f-g-h-xi}, \eqref{eq:Phi} with $y^2=f^2(\xi,t ;\lambda)+g(\xi,t;\lambda)h(\xi,t;\lambda)$, we get
\begin{equation}\label{eq:phi-psi+}
	\begin{split} (\phi_1^2+\psi_1^2)(\phi_1^2-\psi_1^2)^* 
		=&4y\left(2u^2(\xi)-\beta_1-\beta_2\right)+\frac{4\lambda \mu \left(u^2(\xi)-\beta_1\right)\left(u^2(\xi)-\beta_2\right)}{(\lambda+\mu)(\lambda-\mu)}\\	
		=&4y\left(2u^2(\xi)-\beta_1-\beta_2\right)+8\ii \lambda uu_{\xi}.
	\end{split}
\end{equation} 
By \eqref{eq:u^2-beta1}, \eqref{eq:u^2-beta2}, \eqref{eq:WLW}, \eqref{eq:phi-psi+} and integrals formulas of Jacobi elliptic functions \cite[p.191]{ByrdF-54}, we obtain
\begin{equation} \label{eq:K_1(z)}
	\begin{split}
		\mathcal{K}_1(z)
		=&\int_{-PT}^{PT} W^*(\xi;\Omega_1)\mathcal{L}_1W(\xi;\Omega_1) \dd \xi\\
		=&2\ii \lambda\Omega_1\int_{-PT}^{PT} 4y\left(2u^2(\xi)-\beta_1-\beta_2\right) \dd \xi\\
		=&-\Omega_1^2\int_{-PT}^{PT}\left[\alpha^2\left( \dn^2(\ii (z-l))+\dn^2(\ii(z+l)+K+\ii K')-2\right)+2k^2\alpha^2\sn^2(\alpha \xi)\right]\dd  \xi \\
		=&-4\Omega_1^2\alpha P K(k)\left(\dn^2(\ii(z-l))+\dn^2(\ii (z+l)+K+\ii K')-\frac{2E(k)}{K(k)}\right)\\
		=&-4\Omega_1^2\alpha P K(k)\Im\left(I'(z)\right).
	\end{split}
\end{equation}
From Proposition \ref{prop: zr0 zi0} and Lemma \ref{lemma:z R O}, when $l=0$, the value of $\mathcal{K}_1(z)$ could be classified into the following two cases:
\begin{itemize}
	\item 	When $\frac{2E(k)}{K(k)}<1$, i.e., $k> \hat{k}\approx 0.9089$ and $\Im\left(I'(z)\right)> 0$ in \eqref{eq:I'zl}, we get $\mathcal{K}_1(z)>0$, when $\Omega_1\neq 0$, $\Omega_1\in \ii \mathbb{R}$.
	
	\item When $\frac{2E(k)}{K(k)}\ge 1$, i.e., $k\le \hat{k}\approx 0.9089$, from \eqref{eq:I'zl}, there exists a point $z_c$, such that $I'(z_c)=0$, which implies $\mathcal{K}_1(z_c)=0$. By \eqref{eq:I'zl}, the function $\Im\left(I'(z)\right)$ is monotonic increasing as $z\in [0,K']$ and monotonic decreasing as $z\in [-K',0]$. Because $\Im\left(I'(\pm K')\right)=+\infty$ and $\Im\left(I'(0)\right)=1-\frac{2E(k)}{K(k)}<0$, we get $\Im\left(I'(z)\right)>0$, as $z\in \left(-K',-z_{c}\right)\cup\left(z_{c},K'\right)$ and $\Im\left(I'(z)\right)<0$, as $z\in \left(-z_{c},z_{c}\right)$. Combining \eqref{eq:K_1(z)} with $\Omega_1\in \ii \mathbb{R}$, we obtain $\mathcal{K}_1(z)>0$ as $z\in \left(-K',-z_{c}\right)\cup\left(z_{c},K'\right)$ and $\mathcal{K}_1(z)<0$ as $z\in \left(-z_{c},0\right)\cup \left(0,z_{c}\right)$. Since $\Omega(0)=0$ and $I'(z_c)=0$, we get $\mathcal{K}_1(0)=\mathcal{K}_1(z_c)=0$.
\end{itemize} 
Thus, when $k> \hat{k}$, then $\mathcal{K}_1(z)>0 $ for $z\in (-K',0)\cup(0,K')$ and $\mathcal{K}_1(0)=0$ with $\Omega_1(0)=0$. While $k\le \hat{k}$, the function $\mathcal{K}_1(z)$ does not have a consistent sign for $z\in [-K',K']$.

We invoke to calculate the value of $\mathcal{K}_2(z)$. By \eqref{eq:spectral Omega}, we get
\begin{equation} \mathcal{K}_n(z)
	=\left\langle W,\hat{\mathcal{H}}''_n(u)W \right\rangle
	=\left\langle W,\Omega_n \mathcal{J}^{-1} W \right\rangle
	=\frac{\Omega_n}{\Omega_1}\left\langle W,\mathcal{L} W \right\rangle
	=\frac{\Omega_n}{\Omega_1} \mathcal{K}_1(z).\end{equation}
The relationship between $c_{2,1}$ and $c_{2,0}$ is obtained in \eqref{eq:c1c2}. By the $n$-th mKdV hierarchy in \eqref{eq:H-mKdV}, we  get the Lax
pair $\Phi_{t_i}=\mathbf{\hat{V}}_{i}\Phi,i=0,1,\cdots$ where $t_0=\xi, t_1=t$ and $\mathbf{\hat{V}}_0=\mathbf{U},\mathbf{\hat{V}}_1=\mathbf{V}$, $\hat{\mathbf{V}}_n=\mathbf{V}_n+\sum_{i=0}^{n-1} c_{n,i} \mathbf{V}_i$, are given in \eqref{eq:Lax pair} and \eqref{eq:Lax}. From Lemma \ref{lemma:W} and the function \eqref{eq:w-solution}, we know that the eigenvalue $\Omega$ is determined by the solution $\Phi(\xi,t;\lambda)$ of Lax pair \eqref{eq:Lax-pair-1}. Furthermore, we could obtain $\det(\hat{\mathbf{V}}_1-\frac{\Omega_1}{2}\mathbb{I})=0$. Thus, we consider the Lax pair $\Phi_{\xi}=\mathbf{U}\Phi,\Phi_{t_2}=\hat{\mathbf{V}}_2\Phi$ and obtain $\det(\hat{\mathbf{V}}_2-\frac{\Omega_2}{2}\mathbb{I})=0$.
By the linear algebra, the eigenvalue $\Omega_2$ of the second-order mKdV equation demonstrates
\begin{equation}\nonumber
	\Omega_2
	=\left(2s_2+4\lambda^2-c_{2,1}\right)\Omega_1
	=\alpha^2\left(\dn^2(\ii(z-l))+\dn^2(\ii (z+l)+K+\ii K')-2\dn^2(K+2\ii l)-\frac{c_{2,1}}{\alpha^2}\right)\Omega_1.
\end{equation}
Therefore, the Krein signature $\mathcal{K}_2(z)$ is linearly related to the function $\mathcal{K}_1(z)$ via the equation 
\begin{equation}\label{eq:K_2-K_1}
	\begin{split}
		\mathcal{K}_2(z)
		=&\alpha^2\left(\dn^2(\ii (z-l))+\dn^2(\ii (z+l)+K+\ii K')-2\dn^2(K+2\ii l)-\frac{c_{2,1}}{\alpha^2}\right)\mathcal{K}_1(z).
	\end{split}
\end{equation}
Just letting $c_{2,1}=-2\alpha^2\left(\dn^2(K+2\ii l)-\frac{E(k)}{K(k)} \right)$, the value of $\mathcal{K}_2(z)$ could be rewritten as 	 
\begin{equation}\label{eq:K_2(z)}
	\mathcal{K}_2(z)
	=-4\Omega_1^2\alpha^3 K(k)\left(\dn^2(\ii (z-l))+\dn^2(\ii (z+l)+K+\ii K') -\frac{2E(k)}{K(k)}\right)^2 .
\end{equation}
When $l=0$, we have $\mathcal{K}_2(z)\ge 0$. For all $z\in Q\cap \mathbb{R}$, the equality is valid only if $z=0$ or $z=\pm z_c$.

\begin{lemma}\label{lemma:alpha-0}
	If the $\cn$-type solutions of the mKdV equation are spectrally stable with respect to perturbations of the period $2PT, P\in \mathbb{N}$, we could get the following cases:
	\begin{enumerate}
		\item[\rm (a)] If $\frac{2E}{K}\ge 1$ and $M(z_{c})<-\frac{\pi(P-1)}{P}$, all $2PT$ periodic eigenfunctions except $\partial_{\xi}u(\xi)$ satisfy 
		\begin{equation}\label{eq:bounded away}
			\left< \mathcal{L}_2 W, W \right> \ge \alpha_0 \|W\|_{H^2([-PT,PT])}^2, \qquad \alpha_0>0;
		\end{equation} 
		\item[\rm(b)] If $\frac{2E}{K}\ge 1$ with $M(z_{c})=-\frac{\pi(P-1)}{P}$, all $2PT$ periodic eigenfunctions except $\partial_{\xi}u$ and $W(\xi; \Omega(\pm z_c))$ satisfy \eqref{eq:bounded away}.	
		\item[\rm(c)] If $\frac{2E}{K}<1$, all $2T$ periodic eigenfunctions except $\partial_{\xi}u(\xi)$ satisfy
		\begin{equation}\label{eq:bounded away-1}
			\left< \mathcal{L}_1 W, W \right> \ge \alpha_0 \|W\|_{H^1([-T,T])}^2.
		\end{equation} 
	\end{enumerate}
\end{lemma}

\begin{proof}
	(a): By \eqref{eq:K_2(z)}, we know $\mathcal{K}_2(z)>0,z\in Q\cap \mathbb{R}\setminus \{0,\pm z_c\}$. When $M(z_{c})<-\frac{\pi(P-1)}{P}$, the period of $W(\xi;\Omega(\pm z_c))$ is not $2PT$, which is not in the scope of our consideration. Now, we want to prove that there exists $M_1>0$ such that $|\mathcal{K}_2(z)|>M_1$. If not, we could find a sequence $\{z_k\}$ satisfying $\lim_{k\to \infty}\mathcal{K}_2(z_k)=0$ and $W(\xi;\Omega_2(z_k))\in H_{per}^2([-PT,PT])$. Since $\mathcal{K}_2(z)$ is continuous with respect to $z$ and $\mathcal{K}_2(0)=0$, we get that there exists a sub-sequence $\{z_{k_{i}}\}$, such that $\lim_{k_{i} \to \infty}z_{k_{i}}=0$. Without loss of generality, we assume that there exists a sequence such that the sequence $\lim_{n_{i} \to \infty}z_{n_{i}}=0$ satisfies $\lim_{n_{i}\to \infty}\mathcal{K}_2(z_{n_{i}})=0$. For $W(\xi;\Omega_2(z))\in H_{per}^2([-PT,PT])$ with $z\in Q\cap\mathbb{R}$, we can see $z\in Q_P$, i.e., $M(z)=\frac{\pi}{P}n,n\in \mathbb{Z}$. By the continuity of the function $M(z)$, there exists $\delta>0$ such that for any $z\in Q,|z-0|<\delta$, $M(z)\neq \frac{\pi}{P}n,n\in \mathbb{Z}$. Since $\lim_{n_{i}\to \infty} z_{n_i}=0$, there must exist $N$ such that $n_i>N$, $|z_{n_{i}}-0|<\delta$. Choosing $\hat{n}=N+2$, we get $z_{\hat{n}}\in \left\{z_{n_{i}}\right\}$ and $M(z_{\hat{n}})\neq \frac{\pi}{P}n,n\in \mathbb{Z}$, i.e., $W(\xi;\Omega(z_{\hat{n}}))\notin H_{per}^2([-PT,PT])$, which contradicts with $W(\xi;\Omega_2(z))\in H_{per}^2([-PT,PT])$, $z\in \left\{z_n\right\}$. Moreover, for $W(\xi;\Omega_2)\in H_{per}^2([-PT,PT])$, there exists a positive constant $M$ such that $\|W(\xi;\Omega_2)\|_{H^2([-PT,PT])}\leq M$. Thus, we obtain 
	\begin{equation} 
		\alpha_0=\inf_{W \in \mathcal{A}_0}\frac{\left\langle \mathcal{L}_2W,W \right\rangle }{\left\langle W,W \right\rangle }\ge \frac{M_1}{M}>0, 
	\end{equation}
	where $\mathcal{A}_0=\{W_2|\mathcal{L}_2W_2=\Omega_2 W_2\}\setminus \{W_2| \left\langle \mathcal{L}_2W_2,W_2 \right\rangle=0\}$ and $W_2=W(\xi;\Omega_2)$.
	
	(b): If $M(z_{c})=-\frac{\pi(P-1)}{P}$, then $W(\xi;\Omega(-z_c))$ and $W(\xi;\Omega(z_c))$ are $2PT$ periodic functions, where $z_c$ satisfies $\mathcal{K}_2(\pm z_c)=0$. Thus, when $M(z_{c})=-\frac{\pi(P-1)}{P}$, all $2PT$ periodic eigenfunctions except $\partial_{\xi}u$ and $W(\xi; \Omega(\pm z_c))$ satisfy \eqref{eq:bounded away}.
	
	(c): When $\frac{2E(k)}{K(k)}<1$, i.e. $k> \hat{k}\approx 0.9089$, since for all $z\in \mathbb{R}, z\in Q$, the function $\Im(I'(z))> 0$ in \eqref{eq:I'zl}, we get $\mathcal{K}_1(z)>0$, when $\Omega_1\neq 0$. Thus, by a similar procedure as above, we obtain all $2T$ periodic eigenfunctions except $\partial_{\xi}u$ satisfying \eqref{eq:bounded away-1}. 
\end{proof}

\begin{lemma}\label{lemma:H<u}
	$\hat{\mathcal{H}}_2$ is continuous in $H^2_{per}([-PT,PT])$ on the bounded sets; in other words, for any $\epsilon>0$, there exist constants $M_1,\delta>0$, if $\| u-v\|_{H^2}\le\delta$, $\| u\|_{H^2},\| v\|_{H^2}\le M_1$, we have
	\begin{equation} | \hat{\mathcal{H}}_2(u)-\hat{\mathcal{H}}_2(v)|<\epsilon.\end{equation}	
\end{lemma}

\begin{proof}
	By the definition of norm, we know $\|u\|_{H^1}\le\|u\|_{H^2} \le M_1$. Considering the embedding theorem, we obtain $\|u\|_{\infty}\le C\|u\|_{H^1}\le CM_1$ and $\|u_{\xi}\|_{\infty}\le C\|u_{\xi}\|_{H^1}\le C \|u\|_{H^2}\le CM_1$, where $C\in \mathbb{R}$ is a constant. Set $M=\max\{ CM_1,M_1 \}$. From the definition of $\hat{\mathcal{H}}_2$, we know that $\hat{\mathcal{H}}_2(u)=\mathcal{H}_2(u)+c_{2,1}\mathcal{H}_1(u)+c_{2,0}\mathcal{H}_0(u)$. By the H\"{o}lder's inequality, we could get 
	\begin{equation}
			\left|\mathcal{H}_0(u)-\mathcal{H}_0(v)\right|
			\le 2 M\| u-v \|_{H^2} . 	
		\end{equation}
		Similarly, we obtain 
		\begin{equation}
				\left|\mathcal{H}_1(u)-\mathcal{H}_1(v)\right|
				\le (\| u\|_{\infty}^2+\| v\|_{\infty}^2) \left|\int_{-PT}^{PT}(u^2-v^2)\dd \xi\right|+2M\| u-v \|_{H^2} 
				\le ( 4  M^3 + 2 M )\| u-v \|_{H^2}.
	\end{equation}
		Furthermore, we get
		\begin{equation}\begin{split}
				\left|\mathcal{H}_2(u) -\mathcal{H}_2(v)\right|
				\le & M \| u-v\|_{H^2}+6M^5\| u-v\|_{H^2}\\
				& +5 M^2 \left|\int_{-PT}^{PT}(v_{\xi}-u_{\xi})(v+u)+(v_{\xi}+u_{\xi})(v-u) \dd \xi\right|\\
				\le & M \| u-v\|_{H^2} +20M^3\| u-v\|_{H^2} +6M^5\| u-v\|_{H^2}
		\end{split} \end{equation} 
		By the above equations, we have	
		\begin{equation}\begin{split}
				\left|\hat{\mathcal{H}}_2(u) - \hat{\mathcal{H}}_2(v)\right|
				=&\left|\mathcal{H}_2(u)-\mathcal{H}_2(v)+c_{2,1}\left( \mathcal{H}_1(u)-\mathcal{H}_1(v) \right) +c_{2,0}\left( \mathcal{H}_0(u)-\mathcal{H}_0(v) \right)\right| \\
				\le& \left(M+20M^3+6M^5\right) \| u-v\|_{H^2} \\
				&+|c_{2,1} | \left( 4 M^3  +2M\right)\| u-v \|_{H^2}+2|c_{2,0}| M\| u-v \|_{H^2}\\
				\le & C \| u-v \|_{H^2},
	\end{split} \end{equation} 
	where $C=M+20M^3+6M^5+|c_{2,1} | \left( 4 M^3  +2M\right)+2|c_{2,0}|$, which follows that for any $\epsilon>0$, let $\delta=\frac{\epsilon}{C+1}>0$, we have 
	\begin{equation} \left|\hat{\mathcal{H}}_2(u) - \hat{\mathcal{H}}_2(v)\right| \le  C \| u-v \|_{H^2} =C\delta <\epsilon.\end{equation} 
\end{proof}

\newenvironment{proof-orbital-cn}{\emph{Proof of Theorem \ref{theorem:orbital-cn-}.}}{\hfill$\Box$\medskip}
\begin{proof-orbital-cn}
	Colliander et al. \cite{CollianderKSTT-03} studied that the Cauchy problem for the mKdV equation with the periodic boundary condition is globally well-posedness for the initial data $u(\xi,0)\in H^{s}(T),s>\frac{1}{2}$, so it is also global well-posedness for the initial data $u(\xi,0)\in H^2([-PT,PT])$.
	
	At this point, we consider the disturbance
	\begin{equation}\label{eq:h-dis}
		h(\xi,t):=v(\xi,t)-T(\gamma(t))u, \qquad h(\xi,t)\in H^2([-PT,PT]),
	\end{equation}
	in Definition \ref{define:orbital stable}. Set $f(\gamma):=\left\langle v(\xi,t)-T(\gamma(t))u,v(\xi,t)-T(\gamma(t))u\right \rangle$. Consider $\inf_{\gamma\in \mathbb{R}}\|v(\xi,t)-T(\gamma(t))u(\xi)\|$, at the minimum point $$f'(\gamma)=-2\left\langle v(\xi,t)-T(\gamma)u,T'(\gamma)u\right \rangle=-2\left\langle h(\xi,t),T(\gamma)\partial_{\xi}u\right \rangle=0.$$ 
	Without loss of generality, we suppose $\gamma(t)=0$, then we get $T(\gamma(t))u=u$ by \eqref{eq:defin-operator}. And, the perturbation $h(\xi,t)$ belongs to the nonlinear set $\mathcal{A}:=\{h\in H^2([-PT,PT])|\mathcal{H}_0(h(\xi,t)+u)=\mathcal{H}_0(u),\left\langle h(\xi,t),\partial_{\xi}u\right \rangle=0\}$.
	
	The functional of $\mathcal{\hat{H}}_2(u+h)-\mathcal{\hat{H}}_2(u)$ in powers of $h$ yields the expansion
	\begin{equation}\label{eq:H_2}
		\begin{split}
			\mathcal{\hat{H}}_2(u+h)-\mathcal{\hat{H}}_2(u)
			=&\left\langle\frac{\delta \mathcal{\hat{H}}_2}{\delta u}(u),h \right\rangle +\frac{1}{2} \left\langle\frac{\delta^2 \mathcal{\hat{H}}_2}{\delta u^2}(u)h,h \right\rangle +\mathcal{O}(\|h\|_{H^2}^3)\\
			=&\frac{1}{2} \left\langle \mathcal{L}_2(u)h,h \right\rangle +\mathcal{O}(\|h\|_{H^2}^3),
		\end{split}		
	\end{equation}
	where $h(\xi,t)$ is in the nonlinear set $\mathcal{A}$. Then, we consider a tangent plane at $h(\xi,t)=0$ to get a linear space.
	
	Now, we want to prove that there exists a constant $\hat{\alpha}_0$, such that 
	\begin{equation}\label{eq:Lhh,A}
		\left\langle \mathcal{L}_2(u)h,h \right\rangle\ge \hat{\alpha}_0\|h\|_{H^2}^2, \quad h\in \mathcal{A}.
	\end{equation}
	For the small $h(\xi,t)$, it is sufficient to convert \eqref{eq:Lhh,A} for $h(\xi,t)$ in the tangent plane to the admissible space at $h(\xi,t)=0$. Taylor expanding $\mathcal{H}_0$ yields
	\begin{equation}\mathcal{H}_0(u+h)-\mathcal{H}_0(u)=\left\langle u,h \right\rangle +\frac{1}{2} \|h\|_{2}^2.\end{equation}
	So, the linearized version of the nonlinear constraint in $\mathcal{A}$ is the condition $\left\langle u,h \right\rangle=0$. And then, we define the linear admissible space $$\mathcal{A}_1:=\{h_1\in H^2([-PT,PT])| \left\langle h_1(\xi,t),\partial_{\xi}u\right \rangle=\left\langle u,h_1 (\xi,t)\right\rangle=0\}.$$  
	We claim that for any $h(\xi,t)\in \mathcal{A}$ with $\|h\|_{H^2}$ sufficiently small, $h(\xi,t)$ could be decomposed into $h(\xi,t)=h_1(\xi,t)+\hat{c}u(\xi),$ where $\hat{c}=\hat{c}(h)$ and $h_1\in \mathcal{A}_1$. Setting $g(h(\xi,t),\hat{c}):=\left\langle h(\xi,t)-\hat{c}u(\xi),u(\xi)\right \rangle$, we get $g(0,0)=0$ and $g_{\hat{c}}(0,0)=-\left\langle u,u\right \rangle\neq 0$. Consequently, by the implicit function theorem, there exists a neighborhood of $(0,0)$ and a unique functional $\hat{c}(h)$ such that $g(h(\xi,t),\hat{c}(h))=\left\langle h(\xi,t)-\hat{c}(h)u(\xi),u(\xi)\right \rangle\equiv 0$. Letting $h_1(\xi,t)=h(\xi,t)-\hat{c}(h)u(\xi)$, we obtain $\left\langle h_1(\xi,t),u(\xi)\right \rangle=\left\langle h(\xi,t)-\hat{c}(h)u(\xi),u(\xi)\right \rangle=0$, which means $h_1(\xi,t)\in \mathcal{A}_1$. Therefore, we gain the above decomposition. Since $\mathcal{H}_0(u)=\left\langle u,u\right \rangle$ and $\mathcal{H}_0(u+h)=\mathcal{H}_0(u),h\in \mathcal{A}$, we get
	\begin{equation}
		\left\langle u+h,u+h\right \rangle-\left\langle u,u\right \rangle
		=2\left\langle u,h_1+\hat{c}u\right \rangle+\|h\|_{2}^2=2\hat{c}\|u\|_2^2+\|h\|_2^2=0,
	\end{equation}
	which implies $\hat{c}=-\frac{\|h\|_2^2}{2\|u\|_2^2}$. Combined with Lemma \ref{lemma:alpha-0}, $\left\langle \mathcal{L}_2h_1,h_1 \right\rangle\ge \alpha_0 \|h_1\|_{H^2}^2$, if $k\le \hat{k}\approx 0.9089$ and $P<\frac{\pi}{\pi+M(z_c)}$. Thus,
	\begin{equation}\begin{split}
			\left\langle \mathcal{L}_2(u)h,h \right\rangle
			=&\left\langle \mathcal{L}_2(u)(h_1+\hat{c}u),(h_1+\hat{c}u) \right\rangle\\
			\ge &\left\langle \mathcal{L}_2(u)h_1,h_1\right\rangle+2\hat{c}\left\langle \mathcal{L}_2(u)u,h_1\right\rangle+\hat{c}^2\left\langle \mathcal{L}_2(u)u,u\right\rangle \\
			\ge &\alpha_0 \|h_1\|_{H^2}^2,
	\end{split} \end{equation}
	where $\mathcal{L}_2(u)u=0$. Using Minkowski inequality, we could get $\|h_1\|_{H^2}^2\ge \|h\|_{H^2}^2-\hat{c}^2\|u\|_{H^2}^2\ge \|h\|_{H^2}^2-c\|h\|_{H^2}^4$, $c=\|u\|_{H^2}^2/(4\|u\|_{2}^4)$. For $\|h\|_{H^2}^2<\frac{1}{2c}$ sufficient small, we could get $\|h_1\|_{H^2}^2\ge \frac{1}{2}\|h\|_{H^2}^2$.
	Thus, 
	\begin{equation}\label{eq:Lhh-A}
		\left\langle \mathcal{L}_2(u)h,h \right\rangle\ge \frac{\alpha_0}{2}\|h\|_{H^2}^2.
	\end{equation}
	By \eqref{eq:H_2}, we get
	\begin{equation}\label{eq:H_2-3}
		|\mathcal{\hat{H}}_2(u+h)-\mathcal{\hat{H}}_2(u)|\ge \frac{\alpha_0}{4}\|h\|_{H^2}^2-\beta\|h\|_{H^2}^3,
	\end{equation}
	with $\beta>0$. 
	
	Then we want to prove that for any $\frac{1}{2c}>\epsilon>0$, there exists $\hat{\delta}(\epsilon)>0$, when $\Delta:=|\mathcal{\hat{H}}_2(u+h)-\mathcal{\hat{H}}_2(u)|<\hat{\delta}(\epsilon)$, such that $\|h\|_{H^2}<\epsilon$. To analyze the property of \eqref{eq:H_2-3} conveniently, we introduce the cubic function $q(\nu):=\beta \nu^3-\frac{\alpha_0}{4}\nu^2+\Delta$, where $\nu=\|h\|_{H^2}$. It is easy to see that the equation $q(\nu)=0$ has three real roots $\nu_1(\Delta)<0<\nu_2(\Delta)<\nu_3(\Delta)$ for $\Delta>0$.  The set of $\{\nu | q(\nu)\ge 0\}$ is equivalent to $\nu\in [\nu_1(\Delta),\nu_2(\Delta)]\cup [\nu_3(\Delta),+\infty)$. Then, we want to show that if $\|h(\xi,0)\|_{H^2}\le \nu_2(\Delta)$, then $\|h(\xi,t)\|_{H^2}\le \nu_2(\Delta)$ also holds. Actually, if the claim is not valid, there must exist a point $t_0\in \mathbb{R}_+$ such that $\|h(\xi,t_0)\|_{H^2}> \nu_2(\Delta)$, we find $\|h(\xi,t_0)\|_{H^2}\ge \nu_3(\Delta)$ since $\|h(\xi,t_0)\|_{H^2}$ satisfies inequality \eqref{eq:H_2-3}. By the continuity of functions $\|h(\xi,t)\|_{H^2}$ and $\|h(\xi,0)\|_{H^2}\le \nu_2(\Delta), \|h(\xi,t_0)\|_{H^2}\ge \nu_3(\Delta)$, there must exist $t_1\in(0,t_0)$ such that $\nu_2(\Delta)< \|h(\xi,t_1)\|_{H^2}<\nu_3(\Delta)$, which does not satisfy \eqref{eq:H_2-3}. Therefore, we get the contradiction. Then, we obtain that if $\|h(\xi,0)\|_{H^2}\le \nu_2(\Delta)$, $\|h(\xi,t)\|_{H_2}\le \nu_2(\Delta)$. Thus, for any $\epsilon>0$, by choosing $\hat{\delta}(\epsilon)=-\beta\epsilon^3+\frac{\alpha_0^2}{4}\epsilon^2$ we get $\|h(\xi,0)\|_{H_2}<\epsilon$, which further implies $\|h(\xi,t)\|\le\nu_2(\Delta)<\epsilon$.
	Moreover, from Lemma \ref{lemma:H<u}, we know that for the above fixed  $\hat{\delta}(\epsilon)>0$, there exists $\delta(\hat{\delta})$ ($\min\{\epsilon,\frac{1}{2c}\}>\delta(\hat{\delta})>0$), $\|(u(\xi,0)+h(\xi,0))-u(\xi,0)\|_{H^2}\le \delta(\hat{\delta})$, such that
	\begin{equation} |\mathcal{\hat{H}}_2(u(\xi,0)+h(\xi,0))-\mathcal{\hat{H}}_2(u(\xi,0))| < \hat{\delta}(\epsilon),
	\end{equation}
	which further implies
	\begin{equation}
		|\mathcal{\hat{H}}_2(u+h)-\mathcal{\hat{H}}_2(u)| 
		=|\mathcal{\hat{H}}_2(u(\xi,0)+h(\xi,0))-\mathcal{\hat{H}}_2(u(\xi,0))| \le  \hat{\delta}(\epsilon).
	\end{equation}
	
	Therefore, by \eqref{eq:h-dis}, we could obtain that for any $\epsilon>0$, there exists $\delta(\epsilon)>0$, if $\|v(\xi,0)-T(\gamma)u(\xi,0)\|_{H^2}\le \delta(\epsilon)$ and $t\in\mathbb{R}$, the inequality $\inf_{\gamma\in \mathbb{R}} \|v(\xi,t)-T(\gamma)u(\xi,t)\|_{H^2}< \epsilon$ holds, which implies
	\begin{equation}\sup_{t\in\mathbb{R}} \inf_{\gamma\in \mathbb{R}} \|v(\xi,t)-T(\gamma)u(\xi,t)\|_{H^2} < \epsilon. \end{equation}
	From Definition \ref{define:orbital stable}, we get that solution $u(\xi)=\alpha k\cn(\alpha\xi,k), k\le \hat{k}\approx 0.9089$ is orbitally stable in the space $H^2([-PT,PT])$, where $P<\frac{\pi}{\pi+M(z_c)}$.
\end{proof-orbital-cn}

The above proof is in the case of $k\le \hat{k}\approx 0.9089$. When $k> \hat{k}$, by \eqref{eq:I'zl}, we know that for any $z\in \mathbb{R}\cap Q$, the inequality $I'(z)\neq0$ holds. By \eqref{eq:K_1(z)}, we know $\mathcal{K}_1(z)\ge 0$, and only when $\Omega_1=0$, the value $\mathcal{K}_1(z)=0$. Based on Lemma \ref{lemma:alpha-0}, we use a similar proof as the condition $k\le \hat{k}\approx 0.9089$ and obtain that the $\cn$-type solutions are orbitally stable in the space $H^1([-T,T])$ when $k> \hat{k}$.

In the same procedure as above, the $\dn$-type solutions on the periodic space $H^2([-PT,PT])$ are also orbitally stable.

\begin{remark}\label{remark:K_1-dn}
	When $l=\frac{K'}{2}$, for all $z\in Q$ satisfying $\Omega_1(z)\in \ii \mathbb{R}$, the inequality $\mathcal{K}_1(z)\ge 0$ does not hold uniformly. Combined with the half arguments formulas \cite[p.24]{ByrdF-54} of Jacobi elliptic functions, the function $\mathcal{K}_1(z)$ in \eqref{eq:K_1(z)} could be written as 
		\begin{equation}
			\begin{split}
				\mathcal{K}_1(z)
				=&-4\Omega_1^2(z)\alpha PK(k)\left(\dn^2(\ii(z-l))+\dn^2(\ii (z-l)+K)-\frac{2E(k)}{K(k)}\right)\\
				=&-4\Omega_1^2(z)\alpha PK(k)\left(2-k^2\sn^2(\ii(z-l))-k^2\sn^2(\ii (z-l)+K)-\frac{2E(k)}{K(k)}\right)\\
				=&-4\Omega_1^2(z)\alpha PK(k)\left(\frac{-2k^2}{1+\dn(2\ii(z-l))}+2-\frac{2E(k)}{K(k)}\right).
			\end{split}
		\end{equation}
		If $z\in Q\cap \mathbb{R}=\left[-\frac{K'}{2},\frac{3K'}{2}\right]$, then $\ii(z-l)\in \left[-\ii K',\ii K'\right]$. As $2\ii(z-l)\in [0,\ii K')$, the function satisfies $\dn(2\ii(z-l)) \in [1,+\infty)$, which implies $\frac{-2k^2}{1+\dn(2\ii(z-l))}\in [-k^2,0)$. As $2\ii(z-l)\in (\ii K',2\ii K']$, the function satisfies $\dn(2\ii(z-l)) \in (-\infty,-1)$, which means $\frac{-2k^2}{1+\dn(2\ii(z-l))}\in (0,+\infty)$. By the even function $\dn(z)$ and the inequality \eqref{eq:value-E-K-lin-d}, we get that for all $2\ii(z-l)\in [-2\ii K',2\ii K']$,
		\begin{equation}\label{eq:K_1(z)-dn1}
				\mathcal{K}_1(z)
				\ge-4\Omega_1^2(z)\alpha PK(k)\left(-k^2+2-\frac{2E(k)}{K(k)}\right)
				=-4\Omega_1^2(z)\alpha PK(k)\left(k^2+2k'^2-\frac{2E(k)}{K(k)}\right)\ge 0,
		\end{equation} 
		 only when $\Omega_1(z)=0$, $\mathcal{K}_1(z)=0$.  
		
		When $z\in Q\setminus \mathbb{R}$, considering $z=z_R+\ii \frac{K}{2}$ and using shift formulas of the Jacobi elliptic functions \cite[p.20]{ByrdF-54}, we know $\dn(2\ii(z-l))=\dn(2\ii(z_R-l)-K)=k'\nd(2\ii(z_R-l))$. As $\ii(z_R-l)\in [-\ii K',\ii K']$, $\dn(2\ii(z_R-l))\in (-\infty,-1]\cup [1,\infty)$, which means $\dn(2\ii(z-l))=k'\nd(2\ii(z_R-l))\in [-k',k']$. Thus, by inequality \eqref{eq:value-E-K-lin-c}, we obtain
		\begin{equation}\label{eq:K_1(z)-dn2}
				\mathcal{K}_1(z)
				\le -4\Omega_1^2(z)\alpha PK(k)\left(\frac{-2k^2}{1+k'}+2-\frac{2E(k)}{K(k)}\right)
				=-4\Omega_1^2(z)\alpha PK(k)\left(2k'-\frac{2E(k)}{K(k)}\right)\le 0,
		\end{equation}
	and only if $\Omega_1(z)=0$, equation $\mathcal{K}_1(z)=0$ holds.
\end{remark}

\newenvironment{proof-orbital-dn}{\emph{Proof of Theorem \ref{theorem:orbital-dn-}.}}{\hfill$\Box$\medskip}
\begin{proof-orbital-dn}
	As shown in Remark \ref{remark:K_1-dn}, we find that for all $z\in Q$ satisfying $\Omega_1(z)\in \ii \mathbb{R}$, the statement $\mathcal{K}_1(z)\ge 0$ does not always hold. Similarly, we consider the value $\mathcal{K}_2(z)$ in \eqref{eq:K_2(z)}, 
	\begin{equation}
		\mathcal{K}_2(z)	
		=-4\Omega_1^2(z)\alpha PK(k)\left(\dn^2(\ii(z-l))+\dn^2(\ii (z-l)+K)-\frac{2E(k)}{K(k)}\right)^2.
	\end{equation}
	Combining \eqref{eq:K_1(z)-dn1} with \eqref{eq:K_1(z)-dn2}, we obtain $\mathcal{K}_2(z)>0, z\in Q$ if $\Omega_1(z)\neq 0$. In the same way as Lemma \ref{lemma:alpha-0}, excepting function $\partial_{\xi}u $, there exists $\alpha_1$ such that $\left< \mathcal{L}_2 W, W \right> \ge \alpha_1 \|W\|_{H^2([-PT,PT])}^2, P\in \mathbb{Z}_+$. Similar to the proof of Theorem \ref{theorem:orbital-cn-}, we obtain that the solution $u=\alpha \dn(\alpha \xi,k)$ is orbitally stable in the space $H^2([-PT,PT])$, $P\in \mathbb{Z}_+$.
\end{proof-orbital-dn}

\section{Breather solutions on the elliptic function background}\label{section:breather solutions}
In the above sections, we study the linear and orbital stability of elliptic function solutions for the mKdV equation. In this section, we would like to utilize the Darboux-B\"{a}cklund transformation to construct breather solutions $u^{[1]}(\xi,t)$ and $u^{[2]}(\xi,t)$, which can be used to describe the stable or unstable dynamics of elliptic function solutions. Very recently, rogue waves on the elliptic function background are constructed by the Darboux transformation and the nonlinearization in \cite{ChenP-18-NLS,ChenP-18-mKdV,ChenP-19,ChenP-21,ChenPW-19,ChenPW-20}.

\begin{theorem}\label{theorem:construct}
	Suppose $u(\xi)$ is an elliptic function solution of the mKdV equation \eqref{eq:mKdV1}  and $\Phi(\xi,t;\lambda_1)$ is the corresponding fundamental solution of Lax pair \eqref{eq:Lax-pair-1} with the parameter $\lambda_1$, we could construct a new solution $u^{[1]}(\xi,t)$ of the mKdV equation \eqref{eq:mKdV1} with the parameter $\lambda_1$ as follows:
	\begin{equation}\label{eq:bt-1}
		u^{[1]}(\xi,t)=u(\xi)+\frac{2\ii(\lambda_1^*-\lambda_1)\phi_1\psi^*_1}{|\phi_1|^2+|\psi_1|^2},\qquad \lambda_1\in \ii \mathbb{R},\qquad 
		\phi_1\psi_1^*-\phi_1^*\psi_1=0,\\
	\end{equation}
	or
	\begin{equation}\label{eq:bt-2}
		u^{[2]}(\xi,t)
		=u(\xi)-2\ii \begin{bmatrix}
			\phi_1 & \phi_1^*
		\end{bmatrix} \mathbf{M}^{-1} \begin{bmatrix}
			\psi_1^* \\  \psi_1
		\end{bmatrix},\qquad
		\mathbf{M}=\begin{bmatrix}
			\frac{\Phi_1^{\dagger} \Phi_1}{\lambda_1-\lambda_1^*} &
			\frac{\Phi_1^{\dagger} \Phi_1^*}{-\lambda_1^*-\lambda_1^*} \\[5pt]
			\frac{\Phi_1^{\top}\Phi_1}{\lambda_1+\lambda_1} &
			\frac{\Phi_1^{\top}\Phi_1^*}{-\lambda_1^*+\lambda_1}
		\end{bmatrix}, \qquad \lambda_1\in \mathbb{C}\backslash(\ii \mathbb{R}\cup\mathbb{R}),
	\end{equation}
	where $\Phi_1=\Phi(\xi,t;\lambda_1)\mathbf{c}\equiv[\phi_1,\psi_1]^{\top}$ and $\mathbf{c}=\begin{bmatrix}
		1 & c_1
	\end{bmatrix}^{\top}$, $c_1\in \mathbb{C}$.
\end{theorem}

The proof of Theorem \ref{theorem:construct} is given in Appendix \ref{appendix:Darboux transformation}.
\begin{remark}\label{remark:u_2}
	Based on the linear algebra, we rewrite formula \eqref{eq:bt-2} in a compact form: 
	\begin{equation}\label{eq:tilde-u2}
		u^{[2]}(\xi,t)
		=\frac{\det (u(\xi)\mathbf{M}-2\ii \mathbf{N})}{u(\xi)\det(\mathbf{M})},\qquad
		\mathbf{N}=\begin{bmatrix}
			\phi_1\psi^*_1  &	 \phi_1^*\psi^*_1 	\\
			\phi_1\psi_1 &   	 \phi_1^*\psi_1
		\end{bmatrix}.
	\end{equation}
\end{remark}
Beforehand, we introduce the following notations: (1) $\lambda_1=\lambda(z_1)$; (2) equations $E_1(z),E_2(z)$ represent $E_1(\xi,t;z), E_2(\xi,t;z)$ in \eqref{eq:E1-E2} respectively.
\subsection{Explicit stable solutions on the $\dn$-type solution background}
Before constructing an exact solution to describe the stability property of the $\dn$-type solutions, we analyze the constraint in formula \eqref{eq:bt-1} in detail. In this subsection, we choose the case $l=\frac{K'}{2}$, corresponding to the $\dn$-type solution background.

\begin{lemma}\label{lemma:T_1}
	For the formula \eqref{eq:bt-1}, a sufficient condition to the constraint $\phi_1\psi_1^*-\phi_1^*\psi_1=0$ is $u(\xi)=\alpha\dn(\alpha\xi)$ and $|c_1|=1$.
\end{lemma}
\begin{proof}
	Since $\lambda_1\in \ii \mathbb{R}$ in \eqref{eq:lambda}, we could obtain $z_1=z_R\pm\ii \frac{K}{2}, z_1\in S$, based on Proposition \ref{prop:lambda-comforming} and Lemma \ref{lemma:map}. We first consider $z_1=z_R+\ii \frac{K}{2}$. By the shift formula of Jacobi theta functions \cite[p.86]{ArmitageE-06}, it is easy to get 
	\begin{equation}\label{eq:theta-1-3}
		-\frac{\vartheta_1\left(\frac{\ii(z_1-l)\pm \alpha\xi }{2K}\right)}{\vartheta_4\left(\frac{\ii(z_1-l)}{2K}\right)}=\frac{\vartheta_3\left(\frac{\ii(z_1^*+l)\pm \alpha\xi }{2K}\right)}{\vartheta_2\left(\frac{\ii(z_1^*+l)}{2K}\right)}\exp\left(\pm\ii\frac{\alpha\xi}{2K}\pi\right)=\frac{\vartheta_2\left(\frac{\ii(z_1^*-l)\pm \alpha\xi }{2K}\right)}{\vartheta_3\left(\frac{\ii(z_1^*-l)}{2K}\right)}.
	\end{equation}
	Combining the solution $\Phi(\xi,t;\lambda)$ in \eqref{eq:Phi-Theta} with \eqref{eq:theta-1-3}, we obtain  
	\begin{equation}\label{eq:phipsi*}
		\phi_1\psi_1^*
		=D^2 
		\mathbf{E}(z_1)
		\begin{bmatrix}
			\frac{\vartheta_3\left(\frac{\ii(z_1^*+l)-\alpha\xi }{2K}\right)\vartheta_3\left(\frac{\ii(z_1^*+l)+\alpha\xi }{2K}\right)}{\vartheta_2\left(\frac{\ii(z_1^*+l)}{2K}\right)\vartheta_2\left(\frac{\ii(z_1^*+l)}{2K}\right)}& 
			-\frac{\vartheta_2\left(\frac{\ii(z_1^*-l)-\alpha\xi }{2K}\right)\vartheta_3\left(\frac{\ii(z_1+l)-\alpha\xi }{2K}\right)}{\vartheta_3\left(\frac{\ii(z_1^*-l)}{2K}\right)\vartheta_2\left(\frac{\ii(z_1+l)}{2K}\right)}\\
			-\frac{\vartheta_2\left(\frac{\ii(z_1-l)+\alpha\xi }{2K}\right)\vartheta_3\left(\frac{\ii(z_1^*+l)+\alpha\xi }{2K}\right)}{\vartheta_3\left(\frac{\ii(z_1-l)}{2K}\right)\vartheta_2\left(\frac{\ii(z_1^*+l)}{2K}\right)} & 
			\frac{\vartheta_3\left(\frac{\ii(z_1+l)+\alpha\xi }{2K}\right)\vartheta_3\left(\frac{\ii(z_1+l)-\alpha\xi }{2K}\right)}{\vartheta_2\left(\frac{\ii(z_1+l)}{2K}\right)\vartheta_2\left(\frac{\ii(z_1+l)}{2K}\right)}
		\end{bmatrix} 
		\mathbf{E}^{\dagger}(z_1),
	\end{equation}
	where $Z(2\ii l+K)=Z(\ii K'+K)=-\frac{\ii \pi }{2K}$, $D=\frac{\alpha\vartheta_2\vartheta_4}{\vartheta_3\vartheta_4(\frac{\alpha\xi  }{2K})}$ and $	\mathbf{E}(z)=\begin{bmatrix}
		E_1(z) & c_1E_2(z)
	\end{bmatrix}$. Taking the conjugate transpose to the right side of \eqref{eq:phipsi*}, we get
	\begin{equation}
		\phi_1\psi_1^*-\phi_1^*\psi_1
		= D^2  \frac{\vartheta_1\left(\frac{\ii(z_1^*+z_1+2l) }{2K}\right)\vartheta_1\left(\frac{\ii(z_1^*-z_1)}{2K}\right)\vartheta_4^2\left(\frac{\alpha\xi }{2K}\right)}{\vartheta_2^2\left(\frac{\ii(z_1^*+l)}{2K}\right)\vartheta_2^2\left(\frac{\ii(z_1+l)}{2K}\right)}\left(|E_1(z_1)|^2-|c_1|^2|E_2(z_1)|^2\right).
	\end{equation}
	Since $|E_1(z_1)|=|E_2(z_1)|=1$ and $\frac{\vartheta_1\left(\frac{\ii(z_1^*+z_1+2l) }{2K}\right)\vartheta_1\left(\frac{\ii(z_1^*-z_1)}{2K}\right)\vartheta_4^2\left(\frac{\alpha\xi }{2K}\right)}{\vartheta_2^2\left(\frac{\ii(z_1^*+l)}{2K}\right)\vartheta_2^2\left(\frac{\ii(z_1+l)}{2K}\right)}\ne 0,\xi \in \mathbb{R}$, we know that when $\phi_1\psi_1^*-\psi_1\phi_1^*=0$ the value of $c_1$ must satisfy $|c_1|=1$. Similarly, we could get a similar result for $z_1=z_R-\ii\frac{K}{2},z_1\in S$ satisfying $\lambda(z_1)\in \ii \mathbb{R}$.
\end{proof}

When $|c_1|=1$ and $\lambda_1\in \ii \mathbb{R}$, we consider how to reduce formula \eqref{eq:bt-1} with $u(\xi)=\alpha\dn(\alpha\xi)$ and \eqref{eq:Phi-Theta}. Using addition formulas of theta functions in \cite[p.25]{KharchevZ-15}, we get
\begin{equation}\label{eq:M1}
	\frac{|\phi_1|^2+|\psi_1|^2}{\lambda_1-\lambda_1^*}
	=-2\ii D
	\mathbf{E}(z_1)
	\begin{bmatrix}
		\frac{\vartheta_4\left(\frac{\ii(z_1-z_1^*)-\alpha \xi }{2K}\right)}{\vartheta_2\left(\frac{\ii(z_1-z_1^*) }{2K}\right)} & \frac{\vartheta_2\left(\frac{\ii(z_1+z_1^*)-\alpha \xi }{2K}\right)}{\vartheta_3\left(\frac{\ii(z_1+z_1^*) }{2K}\right)} \\
		\frac{\vartheta_2\left(\frac{\ii(z_1+z_1^*)+\alpha \xi }{2K}\right)}{\vartheta_3\left(\frac{\ii(z_1+z_1^*) }{2K}\right)} &\frac{\vartheta_4\left(\frac{\ii(z_1-z_1^*)+\alpha \xi }{2K}\right)}{\vartheta_2\left(\frac{\ii(z_1-z_1^*) }{2K}\right)}
	\end{bmatrix}
	\mathbf{E}^{\dagger}(z_1),
\end{equation}
where $|c_1|=1$ in $\mathbf{E}(z_1)$, $\lambda_1\in \ii \mathbb{R},z_1 \in S$ and	\begin{equation}\label{eq:lambda-lambda*-1}
	\lambda_1-\lambda_1^*
	=\frac{\ii\alpha\vartheta_2^2\vartheta_4\vartheta_1\left( \frac{\ii(z_1-z_1^*)}{2K} \right) \vartheta_3\left( \frac{\ii(z_1+z_1^*)}{2K} \right) } {2\vartheta_3\vartheta_2\left( \frac{\ii(z_1+l)}{2K} \right) \vartheta_4\left( \frac{\ii(z_1-l)}{2K} \right)\vartheta_2\left( \frac{\ii(z_1^*+l)}{2K} \right) \vartheta_4\left( \frac{\ii(z_1^*-l)}{2K} \right)}.
\end{equation}
Combining \eqref{eq:Phi-Theta} with functions \eqref{eq:M1} and \eqref{eq:lambda-lambda*-1}, we reduce the function $u^{[1]}(\xi,t)$ in \eqref{eq:bt-1} to a common denominator and obtain a new periodic solution:
\begin{equation}\label{eq:breather-dn-function}
	\begin{split}
		u^{[1]}(\xi,t)&=\frac{K_{1}(\xi,t)}{H_{1}(\xi,t)},
	\end{split} 
\end{equation} 
in which functions $K_1(\xi,t)$ and $H_1(\xi,t)$ are given by
\begin{equation}
	K_1(\xi,t)
	=
	\mathbf{F}_1(z_1)
	\begin{bmatrix}
		\frac{\vartheta_3\left(\frac{\ii(z_1^*-z_1)+\alpha \xi }{2K}\right)}{\vartheta_1\left(\frac{\ii(z_1-z_1^*) }{2K}\right)} & \frac{\vartheta_1\left(\frac{\ii(z_1+z_1^*)-\alpha \xi }{2K}\right)}{\vartheta_3\left(\frac{\ii(z_1+z_1^*) }{2K}\right)} \\
		-\frac{\vartheta_1\left(\frac{\ii(z_1+z_1^*)+\alpha \xi }{2K}\right)}{\vartheta_3\left(\frac{\ii(z_1+z_1^*) }{2K}\right)} 
		& \frac{\vartheta_3\left(\frac{\ii(z_1-z_1^*)+\alpha \xi }{2K}\right)}{\vartheta_1\left(\frac{\ii(z_1-z_1^*) }{2K}\right)}
	\end{bmatrix}
	\mathbf{F}_2^{\dagger}(z_1),
\end{equation}	
and 
\begin{equation}
	H_1(\xi,t)
	=\frac{\vartheta_3}{\alpha\vartheta_4}
	\mathbf{E}(z_1)
	\begin{bmatrix}
		\frac{\vartheta_4\left(\frac{\ii(z_1-z_1^*)-\alpha \xi }{2K}\right)}{\vartheta_2\left(\frac{\ii(z_1-z_1^*) }{2K}\right)} & \frac{\vartheta_2\left(\frac{\ii(z_1+z_1^*)-\alpha \xi }{2K}\right)}{\vartheta_3\left(\frac{\ii(z_1+z_1^*) }{2K}\right)} \\
		\frac{\vartheta_2\left(\frac{\ii(z_1+z_1^*)+\alpha \xi }{2K}\right)}{\vartheta_3\left(\frac{\ii(z_1+z_1^*) }{2K}\right)} &\frac{\vartheta_4\left(\frac{\ii(z_1-z_1^*)+\alpha \xi }{2K}\right)}{\vartheta_2\left(\frac{\ii(z_1-z_1^*) }{2K}\right)}
	\end{bmatrix}
	\mathbf{E}^{\dagger}(z_1), \end{equation}
respectively, and $|c_1|=1$,
\begin{equation}
		\mathbf{F_1}(z_1)=\begin{bmatrix}
			\frac{\vartheta_2\left(\frac{\ii(z_1+l)}{2K}\right)}{\vartheta_1\left(\frac{\ii(z_1+l)}{2K}\right)}E_1(z_1) & \frac{\vartheta_4\left(\frac{\ii(z_1-l)}{2K}\right)}{\vartheta_3\left(\frac{\ii(z_1-l)}{2K}\right)}c_1E_2(z_1)
		\end{bmatrix},
		\mathbf{F_2}(z_1)=\begin{bmatrix}
			\frac{\vartheta_1\left(\frac{\ii(z_1+l)}{2K}\right)}{\vartheta_2\left(\frac{\ii(z_1+l)}{2K}\right)}E_1(z_1) & \frac{\vartheta_3\left(\frac{\ii(z_1-l)}{2K}\right)}{\vartheta_4\left(\frac{\ii(z_1-l)}{2K}\right)}c_1E_2(z_1)
		\end{bmatrix}.
\end{equation}

To construct the breather solution to describe the stable dynamics of the $\dn$-type solutions of the mKdV equation, we must choose a small enough parameter $\lambda_1$. Based on the elliptic function solution $u=\alpha \dn(\alpha \xi,k)$ with $k=0.9975,\alpha=\frac{1}{16}$, the solution $u^{[1]}(\xi,t)$ of \eqref{eq:mKdV1}, constructed by equation \eqref{eq:breather-dn-function}, is shown in Figure \ref{fig:breaher-dn} by choosing parameters $l=\frac{K'}{2}$, $z_1\approx 0.5726+2.0199\ii$, $c_1=1$, $\lambda_1\approx -0.0292\ii$, $\Omega_1\approx -1.2513\times 10^{-5}\ii $. And the period of the function $u(\xi)$ is $T= \frac{2K}{\alpha}=32K$. By $E_1(z_1)\approx\exp\left(0.0242\ii \xi+6.2567\times10^{-6}\ii t \right)$ and $E_2(z_1)\approx\exp\left(0.0219\ii\xi-6.2567\times10^{-6}\ii t \right)$, the periods on the $\xi$-axis and $t$-axis are $T_{\xi}=10T=320K$ and $T_{t}\approx 5.0211\times10^5$, respectively. 

Furthermore, the solution $u^{[1]}(\xi,t)$ with $z_1\approx 0.5726+2.0199\ii$ shows the stable dynamics for the $\dn$-type solution under perturbations. To compare the dynamics between the $\dn$-type solutions and the corresponding solution $u^{[1]}(\xi,t)$, we shift the traveling wave solution $u(\xi)$ to be $T(\gamma)u:=u(\xi-\alpha^{-1}(z_1-z_1^*)\ii)=\alpha\dn(\alpha\xi-\ii (z_1-z_1^*),k)$, and plot the corresponding curves $T(\gamma)u$ in red in Figure \ref{fig:dn-u-u1}. Choosing the time points $t_0=4\times10^{5}, 0, -2\times10^{5}$, we obtain the figure of $u^{[1]}(\xi,t_0)$ by blue curves. Compared functions $T(\gamma)u$ and $u^{[1]}(\xi,t_0)$ in graphs (i), (ii), and (iii) of Figure \ref{fig:dn-u-u1}, $u^{[1]}(\xi,t_0)$ could be considered as the $\dn$-type solutions adding a small perturbation on $T(\gamma)u$. Figure \ref{fig:dn-3d} shows the $3$-d figure of the function $u^{[1]}(\xi,t)$. By numerical calculations, we obtain the norm  $\|u^{[1]}(\xi,t_0)-T(\gamma)u\|_{H^2([-T_{\xi}/2,T_{\xi}/2])}
$ with $t_0=4\times10^{5}, 0, -2\times10^{5}$ are $0.7065, 0.7043$, and $0.7112$, respectively. When $\|u^{[1]}(\xi,0)-T(\gamma)u\|_{H^2([-T_{\xi}/2,T_{\xi}/2])}=0.7065<\delta=0.8$, we get $\inf_{\gamma\in \mathbb{R}}\|u^{[1]}(\xi,t)-T(\gamma)u\|_{H^2([-T_{\xi}/2,T_{\xi}/2])}\le \|u^{[1]}(\xi,t)-T(\gamma)u\|_{H^2([-T_{\xi}/2,T_{\xi}/2])}<\epsilon=1$, which verifies the stable property. It should be pointed out that the above explicit solution only shows a stable dynamic behavior of the elliptic function solution $u(\xi)$ under a perturbation, which can be regarded as a piece of evidence that the $\dn$-type solutions are stable.

\begin{figure}[ht]
	\centering
	\subfigure[The density plot of the function $u^{[1]}(\xi,t)$. The 3-d figure of (I) is shown in Figure \ref{fig:dn-3d} and the sectional view (i), (ii) and (iii) are shown in Figure \ref{fig:dn-u-u1}.]{
		\label{fig:dn-u1}
		\includegraphics[width=0.87\linewidth]{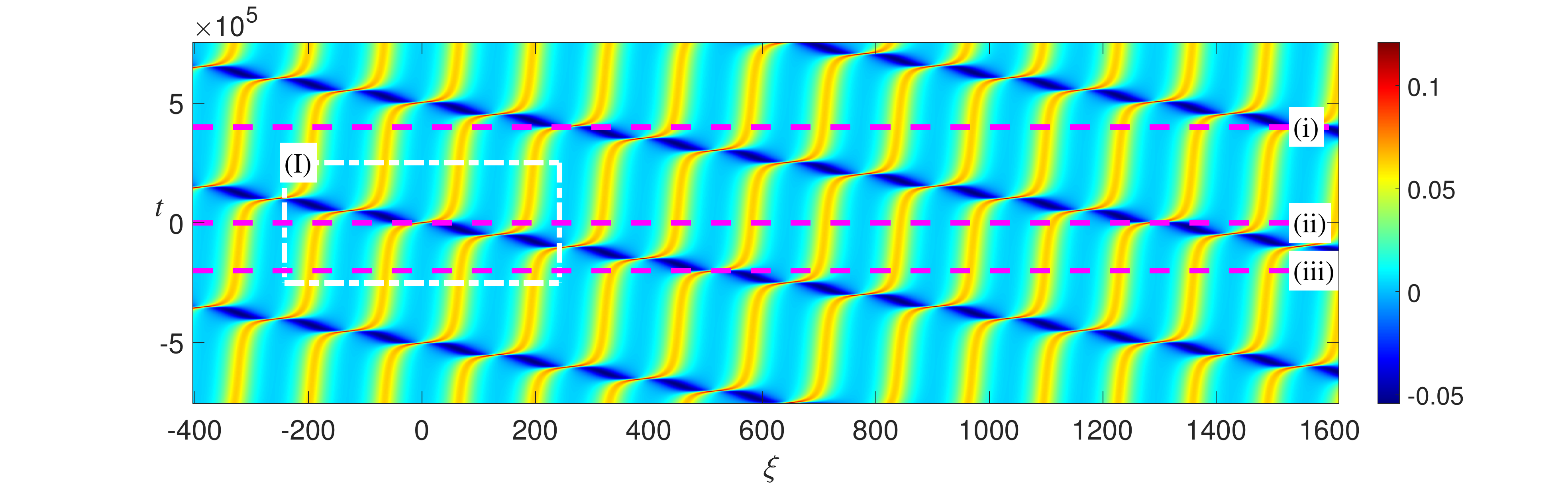}}
	\subfigure[The blue curves in figures (i), (ii) and (iii) are the solution $u^{[1]}(\xi,t_0)$ with $t_0=4\times 10^{5}, 0,-2\times 10^{5}$, respectively. The red curves show the function $u(\xi+\ii(z_1-z_1^*)\alpha^{-1})=\alpha \dn(\alpha \xi+\ii(z_1-z_1^*),k)$.]{
		\label{fig:dn-u-u1}
		\includegraphics[width=0.87\linewidth]{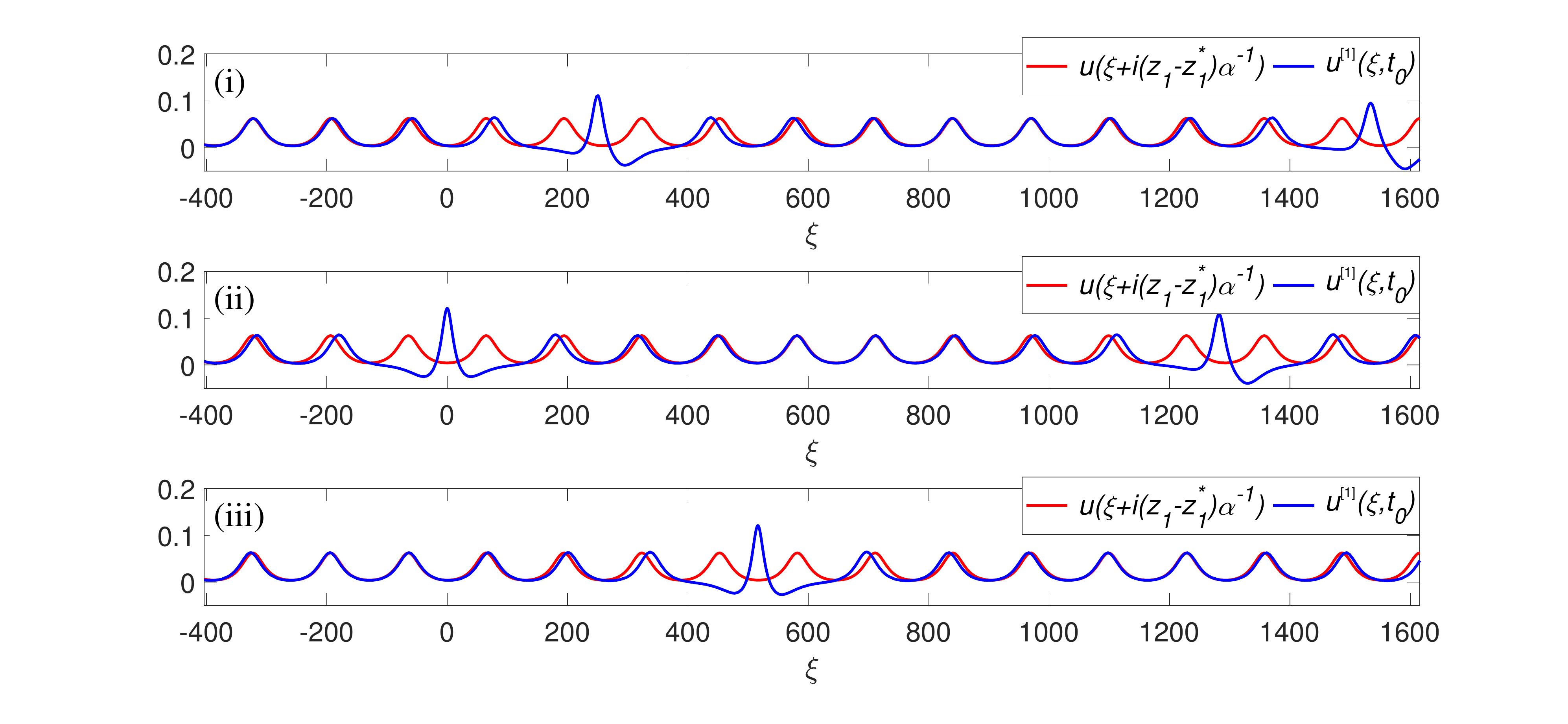}}
	\subfigure[The $3$-d figure of the function $u^{[1]}(\xi,t)$, which reflects the small disturbance of above functions.]{
		\label{fig:dn-3d}
		\includegraphics[width=0.45\linewidth]{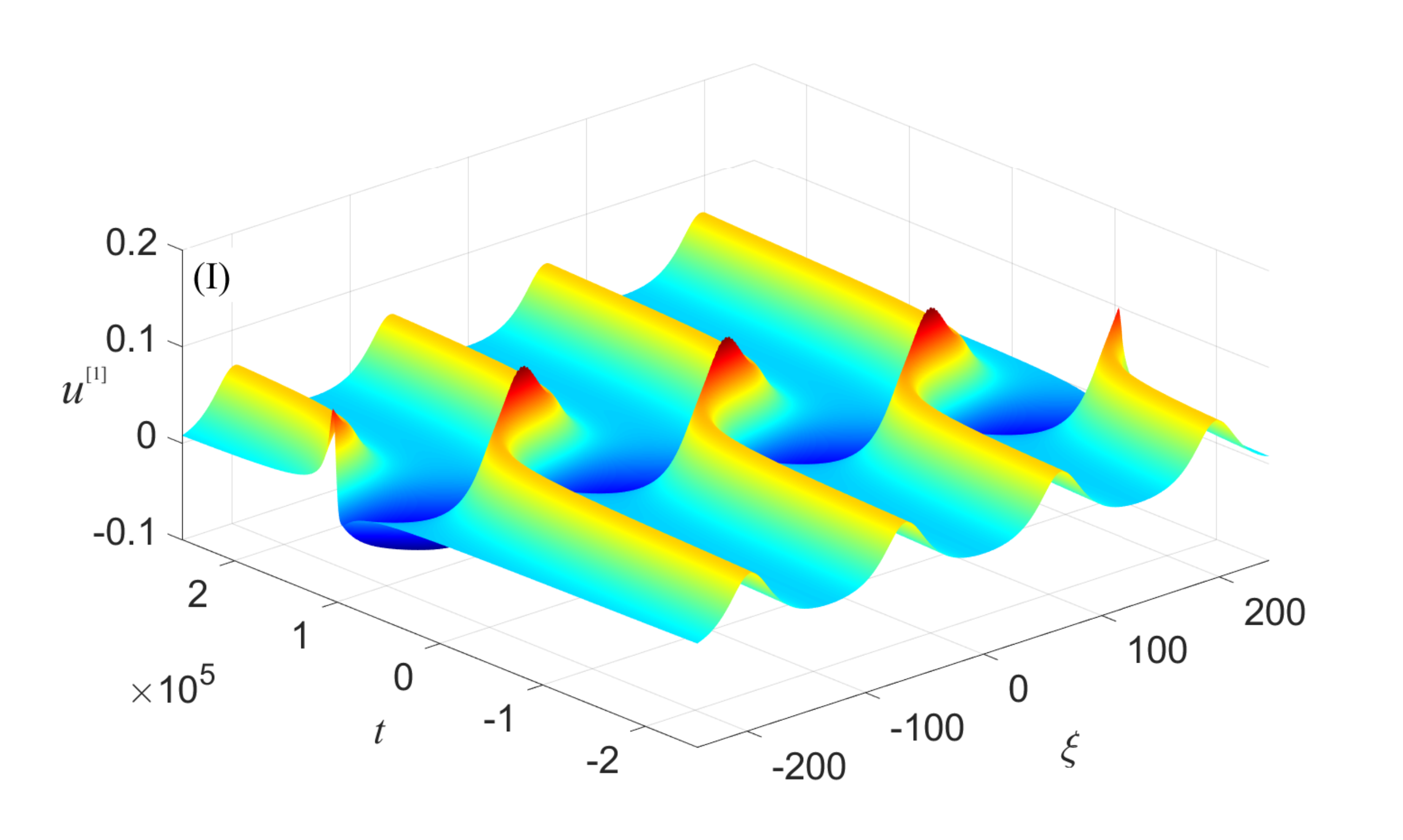}}
	\caption{The solution $u^{[1]}$ is given by \eqref{eq:breather-dn-function} with the parameters setting $k=0.9975,\alpha=\frac{1}{16},l=\frac{K'}{2},z_1\approx 0.5726+2.0199\ii ,c_1=1$. 
	}
	\label{fig:breaher-dn}
\end{figure}

\subsection{Explicit unstable solutions on the $\cn$-type solution background}\label{sub:cn-unstable-explicit}
In what follows, we consider the breather solutions constructed by the $\cn$-type solutions, in which the procedure is similar to the $\dn$-type solutions. 
Based on the expressions of functions in \eqref{eq:Phi-Theta}, \eqref{eq:lambda} and addition formulas of theta functions in \cite[p.25]{KharchevZ-15}, the matrix $\mathbf{M}$ defined in \eqref{eq:bt-2} could be written as 
\begin{equation}\label{eq:M-matrix}
	\mathbf{M}=\begin{bmatrix}
		\frac{\Phi_1^{\dagger} \Phi_1}{\lambda_1-\lambda_1^*} &
		\frac{\Phi_1^{\dagger} \Phi_1^*}{-\lambda_1^*-\lambda_1^*} \\[5pt]
		\frac{\Phi_1^{\top}\Phi_1}{\lambda_1+\lambda_1} &
		\frac{\Phi_1^{\top}\Phi_1^*}{-\lambda_1^*+\lambda_1}
	\end{bmatrix}=\begin{bmatrix}
		M(-z_1^*,z_1) &	M(-z_1^*,-z_1^*) \\
		M(z_1,z_1) & M(z_1,-z_1^*)
	\end{bmatrix}
\end{equation}
where
\begin{equation}
	M(a,b):
	=-2\ii D
	\mathbf{E}(a)
	\begin{bmatrix}
		\frac{\vartheta_4\left( \frac{\ii(a+b)-\alpha\xi}{2K} \right)}{\vartheta_1\left( \frac{\ii(a+b)}{2K} \right)} & \frac{\vartheta_2\left( \frac{\ii(a-b)+\alpha\xi}{2K} \right)}{\vartheta_3\left( \frac{\ii(a-b)}{2K} \right)} \\
		\frac{\vartheta_2\left( \frac{\ii(a-b)-\alpha\xi}{2K} \right)}{\vartheta_3\left( \frac{\ii(a-b)}{2K} \right)} & \frac{\vartheta_4\left( \frac{\ii(a+b)+\alpha\xi}{2K} \right)}{\vartheta_1\left( \frac{\ii(a+b)}{2K} \right)}
	\end{bmatrix}
	\mathbf{E}^{\top}(b).
\end{equation}
Utilizing addition formulas of theta functions in \cite[p.25]{KharchevZ-15} and conversion formulas between Jacobi elliptic functions and theta functions \cite[p.83]{ArmitageE-06}, we get the expression of the matrix $\mathbf{K}$:\\
\begin{equation}\label{eq:K-matrix}
	\mathbf{K}=\begin{bmatrix}
		\frac{\Phi_1^{\dagger}\Phi_1}{\lambda_1-\lambda_1^*}-\frac{2{N}_{1,1}}{u} & \frac{\Phi_1^{\dagger}\Phi_1^*}{-\lambda_1^*-\lambda_1^*}-\frac{2N_{1,2}}{u}\\[5pt]
		\frac{\Phi_1^{\top}\Phi_1}{\lambda_1+\lambda_1}-\frac{2N_{2,1}}{u} & \frac{\Phi_1^{\top}\Phi_1^*}{\lambda_1-\lambda_1^*}-\frac{2N_{2,2}}{u}
	\end{bmatrix}=\begin{bmatrix}
		K(-z_1^*,z_1) & K(-z_1^*,-z_1^*)\\
		K(z_1,z_1) & K(z_1,-z_1^*) 
	\end{bmatrix},
\end{equation}
where $N_{i,j}$ represents the $(i,j)$-elements of the matrix $\mathbf{N}$ in \eqref{eq:tilde-u2},
\begin{equation}
	K(a,b)= -2\ii D\frac{ \alpha\vartheta_4}{ \vartheta_3} 
	\mathbf{F}_3(a)
	\begin{bmatrix}
		\frac{ \vartheta_2\left( \frac{\ii (a+b)-\alpha\xi}{2K}\right)}{ \vartheta_1\left( \frac{\ii (a+b)}{2K}\right)} & \frac{ \vartheta_4\left( \frac{\ii (a-b)+\alpha\xi}{2K}\right)}{ \vartheta_3\left( \frac{\ii(a-b)}{2K}\right)}\\
		-\frac{ \vartheta_4\left( \frac{\ii (a-b)-\alpha\xi}{2K}\right)}{ \vartheta_3\left( \frac{\ii (a-b)}{2K}\right)} & \frac{ \vartheta_2\left( \frac{\ii (a+b)+\alpha\xi}{2K}\right)}{ \vartheta_1\left( \frac{\ii(a+b)}{2K}\right)}
	\end{bmatrix}
	\mathbf{F}_4^{\top}(b), 
\end{equation}
and
\begin{equation}
		\mathbf{F}_3(a)=\begin{bmatrix}
			\frac{ \vartheta_2\left( \frac{\ii a}{2K}\right)}{ \vartheta_4\left( \frac{\ii a}{2K}\right)}E_1(a) &  \frac{ \vartheta_4\left( \frac{\ii a}{2K}\right)}{ \vartheta_2\left( \frac{\ii a}{2K}\right)}c_1 E_2(a)
		\end{bmatrix},\qquad
		\mathbf{F}_4(b)=\begin{bmatrix}
			\frac{ \vartheta_4\left( \frac{\ii b}{2K}\right)}{ \vartheta_2\left( \frac{\ii b}{2K}\right)}E_1(b) &
			\frac{ \vartheta_2\left( \frac{\ii b}{2K}\right)}{ \vartheta_4\left( \frac{\ii b}{2K}\right)}c_1 E_2(b)
		\end{bmatrix}.
\end{equation}
Based on Remark \ref{remark:u_2} and the formula $u^{[2]}(\xi,t)$ in \eqref{eq:tilde-u2}, we get the breather solution on the $\cn$-type solution background:	\begin{equation}\label{eq:breather-cn-function}
	u^{[2]}(\xi,t)
	=\frac{\det (\mathbf{K})}{\det(\mathbf{M})},
\end{equation}
where matrices $\mathbf{K}$ and $\mathbf{M}$ are defined in \eqref{eq:M-matrix} and \eqref{eq:K-matrix}, respectively. The parameter $z_1\in S$ in the above solutions \eqref{eq:breather-cn-function} satisfies $\lambda(z_1)\notin \ii \mathbb{R}$.

We study the asymptotic analysis of formula \eqref{eq:breather-cn-function} for all $z\in Q$.  For convenience, we introduce notations $E_{I}(z_1):={\Im}(\alpha Z(\ii z_1)+\ii \lambda_1)$ and $E_{R}(z_1):={\Re}(\alpha Z(\ii z_1)+\ii \lambda_1)=0$, $z_1\in Q$. By \eqref{eq:w-solution}, we rewrite the function $E_1(\xi,t;z_1)$ defined in \eqref{eq:E1-E2} as $E_{1}(\xi,t;z_1)=\exp\left( \ii E_{I}(z_1)\xi+\frac{\Omega(z_1)}{2}t \right)$. By the relationship between function $E_1(\xi,t;z_1)$ and $E_2(\xi,t;z_1)$, we get $E_{2}(\xi,t;z_1)=\exp\left(- \ii E_{I}(z_1)\xi-\frac{\Omega(z_1)}{2}t \right)$. Without loss of generality, we set $\Re(\Omega(z_1))>0$. As $t\rightarrow \pm \infty$, the breather solution $u^{[2]}(\xi,t)$ \eqref{eq:breather-cn-function} will tend to the stationary solution $u(\xi)$ with a shift:
\begin{equation}\label{eq:cn-infty}
	\begin{split}
		u^{[2]}_{\pm \infty}(\xi)=\lim_{t\to \pm \infty}u^{[2]}(\xi,t)=\alpha \frac{\vartheta_2\vartheta_4 \vartheta_2\left(\frac{\alpha\xi \pm 2\ii (z^*_1-z_1) }{2K}\right)}{\vartheta_3^2\vartheta_4\left(\frac{\alpha\xi \pm 2\ii (z^*_1-z_1) }{2K}\right)}
		=\alpha k\cn(\alpha\xi \pm 2\ii (z_1^*-z_1)).
	\end{split}
\end{equation}
Based on the addition formulas of theta functions in \cite[p.25]{KharchevZ-15} and exact expressions of the solution $u^{[2]}(\xi,t)$ in \eqref{eq:breather-cn-function}, as $t\rightarrow \pm \infty$, the asymptotic expansion of solution \eqref{eq:breather-cn-function} is given by
\begin{equation}\label{eq:breather-cn-asym}
	u^{[2]}(\xi,t)
	=u^{[2]}_{\pm\infty}(\xi)+\left(\ee^{\mp \ii \Im(\Omega)t}A_{\pm}(\xi;z_1)+\ee^{\pm \ii \Im(\Omega)t}A^*_{\pm}(\xi;z_1)\right)
	\ee^{\mp \Re(\Omega)t}+\mathcal{O}\left(\ee^{\mp 2\Re(\Omega)t} \right),
\end{equation} 
where \begin{equation}\label{eq:A-pm}
	\begin{split}
		A_{\pm}(\xi;z_1)=&
		B_{\pm}\frac{\alpha^2\vartheta_2^2\vartheta_4^2 \ee^{-2\ii E_{I}(\alpha\xi+2\ii(z_1^*-z_1))}}
		{\vartheta_3^2\vartheta_4^2\left(\frac{\alpha \xi\pm2\ii (z_1^*-z_1)}{2K} \right)} 
		\left( \frac{\vartheta_1^2\left(\frac{\alpha\xi \pm\ii(2z_1^*-z_1)}{2K}\right)}{\vartheta_4^2\left(\frac{\ii z_1}{2K}\right)} 
		- \frac{\vartheta_3^2\left(\frac{\alpha\xi \pm\ii(2z_1^*-z_1)}{2K}\right)}{\vartheta_2^2\left(\frac{\ii z_1}{2K}\right)} \right),\\
		B_{\pm}=&c_1^{\pm 1}\frac{\vartheta_1\left(\frac{2\ii  z_1}{2K}\right)\vartheta_3\left(\frac{2\ii  z_1}{2K}\right)\vartheta_1\left(\frac{\ii  (z_1-z_1^*)}{2K}\right)\vartheta_3\left(\frac{\ii  (z_1-z_1^*)}{2K}\right)}
		{\alpha\vartheta_2\vartheta_4\vartheta_1\left(\frac{\ii  (z_1+z_1^*)}{2K}\right)\vartheta_3\left(\frac{\ii  (z_1+z_1^*)}{2K}\right)}. 
	\end{split}
\end{equation}
Then, we consider the coefficients of $ \ee^{\mp \Re(\Omega)t}$ as $t\rightarrow \pm \infty$. Comparing the expressions between function $A_{\pm}(\xi;z_1)$ in \eqref{eq:A-pm} and $W(\xi;\Omega)$ in \eqref{eq:Phi-Theta} and \eqref{eq:w-solution}, we get 
\begin{equation}\label{eq:A+-}
		A_{+}(\xi;z_1)=B_{+}W_2(\xi+ 2\alpha^{-1}\ii(z_1^*-z_1)), \qquad
		A_{-}(\xi;z_1)=B_{-}W_1(\xi- 2\alpha^{-1}\ii(z_1^*-z_1)),
\end{equation} $W_1(\xi)=(\phi_1^2(\xi,t)-\psi_1^2(\xi,t))\exp(-\Omega t)$ and $W_2(\xi)=(\phi_2^2(\xi,t)-\psi_2^2(\xi,t))\exp(\Omega t)$. By \eqref{eq:breather-cn-asym}, we could define
\begin{equation}\label{eq:defin-w}
			w_{\pm}(\xi,t)=w(\xi\pm 2\alpha^{-1}\ii (z_1^*-z_1),t)
			=A_{\pm}(\xi;z_1)\ee^{\mp \Omega t}+A^*_{\pm}(\xi;z_1)\ee^{\mp \Omega^* t}.
\end{equation}
Therefore, it is easy to verify that \eqref{eq:u2-cn} holds.

By the proper translation, we can see that the perturbation $w(\xi,t)$ \eqref{eq:w} of the solution \eqref{eq:perturbation} in the linear stability analysis corresponds precisely to the asymptotic form $w_{\pm}(\xi,t)$ of the solution \eqref{eq:breather-cn-asym}. In other words, as $t\rightarrow \pm \infty$, the asymptotic analysis \eqref{eq:breather-cn-asym} is consistent with solutions $w(\xi\pm 2\alpha^{-1}\ii (z_1^*-z_1),t)$. Furthermore, the perturbation condition \eqref{eq:perturbation} is completely consistent with the asymptotic behavior \eqref{eq:breather-cn-asym}. When $t\rightarrow \pm \infty$, $\exp(\mp2E_Rt)\rightarrow 0$. And then, the function $w(\xi\pm 2\alpha^{-1}\ii(z_1^*-z_1),t)$ could be seen as a small perturbation on function $u^{[2]}(\xi,t)$. As time $t$ changes, functions $w(\xi\pm 2\alpha^{-1}\ii(z_1^*-z_1),t)$ are not always small enough. Therefore, the above phenomena explain that the solution is linearly unstable if ${\Re}(\Omega)=2E_R\ne 0$.

We exhibit a breather solution that can be utilized to describe the unstable dynamics for the $\cn$-type solutions of the mKdV equation. We consider the function $u^{[2]}(\xi,t)$ with parameters $\alpha=1,k=\frac{1}{4},l=0, z_1\approx -1.358+0.433\ii $, or $ \lambda_1 \approx 0.484-0.094\ii, \Omega(\lambda_1) \approx -0.090-0.307\ii$. In Figure \ref{fig:breaher-cn},  the plotting of function $u^{[2]}(\xi,t)$ is shown by the density plot and $3$-d figure. Since $\Omega(\lambda_1)\notin\ii \mathbb{R}$, $u^{[2]}(\xi,t)$ is a localized function in the $t$-axis and as $t\rightarrow \pm \infty$ $u^{[2]}(\xi,t)$ tends to a $4K$-periodic function that could be seen as a translation of the function $u=\alpha k\cn(\alpha \xi)$, in \eqref{eq:t-infty}. On the $\xi$-axis, considering the exponent part of $u^{[2]}(\xi,t)$, we get $E_1\approx \exp(0.410\ii\xi-(0.045+0.154\ii )t)$ and $ E_2\approx\exp(-0.410\ii\xi+(0.045+0.154\ii )t)$. It is easy to obtain that the period of $u^{[2]}(\xi,t)$ is $T=12K(k)\approx 19.155$ in the $\xi$-direction. It is seen that the dynamics of $u^{[2]}(\xi,0)$ are entirely different from the one of the $\cn$-type solutions, which verifies that the small perturbation for the $\cn$-type solutions will yield enormous variation with the evolution of time.

\begin{figure}[ht]
	\centering
	\includegraphics[width=0.98\linewidth]{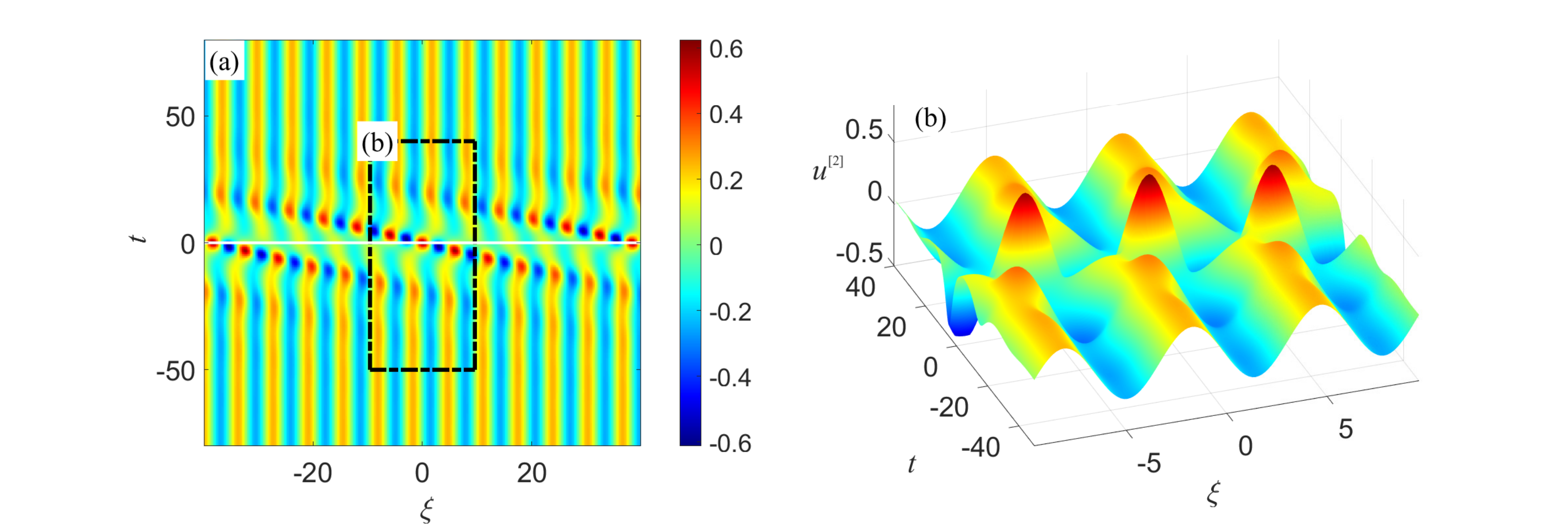}
	\caption{(a): The density plot of $u^{[2]}(\xi,t)$; (b): The 3-d figure of the black rectangle in figure (a).
			The parameters: $l=0, \alpha=1,k=\frac{1}{4},z_1\approx -1.358+0.433\ii ,c_1=1$.	
	}
	\label{fig:breaher-cn}
\end{figure}

\section*{Acknowledgement}

This work is supported by the National Natural Science Foundation of China (Grant No. 12122105) and the Guangzhou Science and Technology Program of China (Grant No. 201904010362). We would also like to thank the referee for their valuable comments, questions, syntax checking and suggestions.

\appendix

\titleformat{\section}[display]
{\centering\LARGE\bfseries}{ }{11pt}{\LARGE}

\renewcommand{\appendixname}{Appendix \ \Alph{section} }

\section{\appendixname. The definitions and properties of elliptic functions}\label{appendix:Elliptic functions}

\setcounter{equation}{0}
\setcounter{definition}{0}
\setcounter{prop}{0}
\setcounter{lemma}{0}
\renewcommand\theequation{\Alph{section}.\arabic{equation}}

\renewcommand\thedefinition{\Alph{section}.\arabic{definition}}
\renewcommand\theprop{\Alph{section}.\arabic{prop}}

\renewcommand\thelemma{\Alph{section}.\arabic{lemma}}

In this Appendix, we enumerate the definition of special functions obtained in \cite{ArmitageE-06,ByrdF-54,KharchevZ-15} and provide relevant results, which will be utilized in this paper.

\subsection*{Complete elliptic integrals}
Functions $K$ and $E$ are called the first and second complete elliptic integrals defined as
\begin{equation}\label{eq:J-K-E-int}
	K\equiv K(k)=\int_{0}^{\frac{\pi}{2}}\frac{\dd \theta}{\sqrt{1-k^2\sin^2\theta}},\quad \text{and} \quad
	E\equiv E(k) =\int_{0}^{\frac{\pi}{2}}\sqrt{1-k^2\sin^2\theta}\dd \theta.
\end{equation}
In addition to the above two integrals, we usually use an associated complete elliptic integral $K'=K(k'), k'=\sqrt{1-k^2}$. Meanwhile, we provide some inequalities showing the relationship between the complete elliptic integrals and the modulus $k$.
\begin{prop}
	For any $k\in(0,1)$, the following four inequalities hold:
	\begin{subequations}\label{eq:value-E-K-lin}
		\begin{align}
			E(k)-k'^2K(k)&>\lim_{k\rightarrow 0}E(k)-(k'^2)K(k)=0, \label{eq:value-E-K-lin-a} \\
			K(k)-E(k)&>\lim_{k\rightarrow 0}K(k)-E(k)=0, \label{eq:value-E-K-lin-b}\\
			E(k)-k'K(k)&> \lim_{k\rightarrow 0}E(k)-k'K(k)=0. \label{eq:value-E-K-lin-c} \\
			(1+k'^2)K(k)-2E(k)&> \lim_{k\rightarrow 0}(1+k'^2)K(k)-2E(k)=0. \label{eq:value-E-K-lin-d}
		\end{align}
	\end{subequations}
\end{prop}

\begin{proof}
	According to the derivatives of the elliptic integrals with respect to the modulus \cite[p.282]{ByrdF-54}, we obtain
	\begin{equation}\label{eq:der-E-K-lin}
		\begin{split}
			&\frac{\dd(E(k)-k'^2K(k))}{\dd k}=kK(k)>0, \quad  \text{and} \quad
			\frac{\dd (K(k)-E(k))}{\dd k}=\frac{kE(k)}{k'^2}>0,\\
		\end{split}
	\end{equation}
	where $K(k)$, $E(k)>0$ and $k\in (0,1)$. By the definition of $K(k)$ and $E(k)$, it is easy to get that $\lim_{k\rightarrow 0}E(k)-K(k)=0$. Then, the inequalities \eqref{eq:value-E-K-lin-a} and \eqref{eq:value-E-K-lin-b} holds. Furthermore, combining the derivatives of the elliptic integrals \cite[p.282]{ByrdF-54} with inequalities \eqref{eq:value-E-K-lin-a} and \eqref{eq:der-E-K-lin}, we get 
	\begin{equation}
		\begin{split}
			\frac{\dd}{\dd k}(E(k)-k'K(k))=&\frac{(K(k)-E(k))(1-k')}{kk'}>0,\\
			\frac{\dd}{\dd k}((1+k'^2)K(k)-2E(k))=&\frac{(E(k)-k'^2K(k))k}{k'}>0.
		\end{split}
	\end{equation} 
	Therefore, we obtain the inequalities \eqref{eq:value-E-K-lin-c} and \eqref{eq:value-E-K-lin-d}.
\end{proof}

\subsection*{Jacobi Theta function}
\begin{definition}
	The theta functions are defined as the summation:
	\begin{equation}\label{eq:theta1234}
		\begin{split}
			\vartheta_1(z)&=\ii\sum_{n=-\infty}^{+\infty}(-1)^{n}q^{\left(n-\frac{1}{2}\right)^2}\ee^{(2n-1)\ii z}, \qquad 
			\vartheta_3(z)=\sum_{n=-\infty}^{+\infty}q^{n^2}\ee^{2n\ii z},\\
			\vartheta_2(z)&=\sum_{n=-\infty}^{+\infty}q^{\left(n-\frac{1}{2}\right)^2}\ee^{(2n-1)\ii z}, \qquad
			\vartheta_4(z)=\sum_{n=-\infty}^{+\infty}(-1)^{n}q^{n^2}\ee^{2n\ii z},
		\end{split}
	\end{equation}
	where $q=\ee^{\ii \pi \frac{K'}{K}}$.
\end{definition}

\subsection*{Weierstrassian Zeta function}
\begin{definition}\label{define:zeta}
	The Weierstrass Zeta function $\zeta(z)$ is defined by
	\begin{equation} \zeta(z)=\frac{1}{z}+\sum_{\omega\ne 0}\left( \frac{1}{z-\omega}+\frac{1}{\omega}+\frac{z}{\omega^2}\right), \end{equation}
	where $\omega=2m\omega_1+2n\omega_3$, $n,m\in \mathbb{Z}$ and $2\omega_1,2\omega_3$ are two periods of the derivative function of $\zeta(z)$.
\end{definition}

The shift formulas are given by
\begin{equation}\label{eq:zeta+omega}
	\zeta(z+2\omega_1)=\zeta(z)+2\eta_1, \quad \zeta(z+2\omega_3)=\zeta(z)+2\eta_3,
\end{equation}
where $\eta_1=\zeta(\omega_1)$ and  $\eta_3=\zeta(\omega_3)$.
Furthermore, the $\zeta(z)$ function could be written as 
\begin{equation}\label{eq:zeta sigma vartheta}
	\zeta(z)=\frac{\sigma'(z)}{\sigma(z)},\qquad \sigma(z)=\frac{2\omega}{\vartheta_1'}\exp\left(\frac{\eta z^2}{2 \omega}\right)\vartheta_1\left(\frac{z}{2\omega} \right),\qquad  \eta=\eta_1, \qquad \vartheta_1'=\vartheta'_1(0).
\end{equation}

\subsection*{Jacobi Zeta function}
\begin{definition}
	The Jacobi Zeta function is defined by 
	\begin{equation}
		Z(z)\equiv \int_0^z \left( \dn^2(u)-\frac{E}{K} \right) \dd u,
	\end{equation}
	where $E\equiv E(k), K\equiv K(k)$ are the complete elliptic integrals defined in \eqref{eq:J-K-E-int}.
\end{definition}

With the help of \cite{ArmitageE-06}, we get some formulas on elliptic functions.
\begin{prop}\label{prop:int-theta}
	If $f(z)$ is an elliptic function with simple poles $\beta_r,r=1,2\cdots m$, in a periodic region $(2\omega_1,2\omega_3)$, we get the integration 
	\begin{equation}\label{eq:int-f} \int_{0}^{z}f(s)\dd s=Cz+\sum_{r=1}^{m}B_r\ln \frac{\vartheta_1\left( \frac{\beta_r-z}{2\omega_1}\right)}{\vartheta_1\left( \frac{\beta_r}{2\omega_1}\right)}, \end{equation}
	where $B_r$ is the residue of the pole $\beta_r$ and $C$ is a constant which will be determined during the calculation.
\end{prop}
\begin{proof}
	Set $\varphi(z)=\sum_{r=1}^{m}B_r \zeta (z-\beta_r)$.
	Based on the residue theorem, equation $\sum_{r=1}^{m}B_r=0$ holds, since $B_r,r=1,2...$ are the residues of all poles $\beta_r$. By \eqref{eq:zeta+omega}, we could verify that $\omega_1$ and $\omega_3$ are the periods of the function $\varphi(z)$ by equations
	\begin{equation}
		\begin{split}
			\varphi(z+2\omega_1)=&\sum_{r=1}^{m}B_r\left(\zeta (z-\beta_r)+2\eta_1\right)=\varphi(z)+ 2\eta_1\sum_{r=1}^{m}B_r=\varphi(z),\\
			\varphi(z+2\omega_3)=&\sum_{r=1}^{m}B_r\left(\zeta (z-\beta_r)+2\eta_3\right)=\varphi(z)+ 2\eta_3\sum_{r=1}^{m}B_r=\varphi(z).
		\end{split}
	\end{equation}
	Thus, functions $f(z)$ and $\varphi(z)$ share the same poles and periods. By the Liouville theorem, we get $f(z)=\varphi(z)+C$, where $C=f(0)-\varphi(0)$ is a constant. By \eqref{eq:zeta sigma vartheta}, we get 
	\begin{equation}\begin{split}
			\int_{0}^{z}f(s)\dd s
			&=\int_{0}^{z} (C+\varphi(s))\dd s \\
			&=Cz+\int_{0}^{z} \sum_{r=1}^{m}B_r\left( \ln \sigma (s-\beta_r) \right)_s\dd s\\
			&=Cz+\sum_{r=1}^{m}\int_{0}^{z}\left[B_r\frac{\eta s}{\omega_1}+ B_r\left( \ln \vartheta_1\left( \frac{s-\beta_r}{2\omega_1}\right) \right)_s\right]\dd s\\
			&=Cz+\sum_{r=1}^{m}B_r\frac{\eta z}{\omega_1}+ \sum_{r=1}^{m}B_r\int_{0}^{z}\left( \ln \vartheta_1\left( \frac{\beta_r-s}{2\omega_1}\right) \right)_s\dd s\\
			&=Cz+ \sum_{r=1}^{m}B_r  \ln \frac{\vartheta_1\left( \frac{\beta_r-z}{2\omega_1}\right)}{\vartheta_1\left( \frac{\beta_r}{2\omega_1}\right)}.
	\end{split} \end{equation}
	Thus, \eqref{eq:int-f} holds.
\end{proof}

Based on Proposition \ref{prop:int-theta}, we gain the following results on the elliptic integration:	
\begin{lemma}\label{lemma:int}
	\begin{equation}\label{eq:int beta}
		\begin{split}
			\int_{0}^{\xi  }\frac{2\ii \lambda \beta_1}{u^2(s)-\beta_1}\dd  s
			&=\frac{1}{2}\ln\frac{\vartheta_1(\frac{\ii (z-l)-\alpha\xi  }{2K})}{\vartheta_1(\frac{\ii (z-l)+\alpha \xi  }{2K})}+\alpha Z(\ii (z-l))\xi,\\ 
			\int_{0}^{\xi  }\frac{2\ii \lambda \beta_2}{u^2(s)-\beta_2}\dd  s
			&=-\frac{1}{2}\ln\frac{\vartheta_1(\frac{\ii (z+l)+K+\ii K'-\alpha \xi  }{2K})}{\vartheta_1(\frac{\ii (z-l)+K+\ii K'+\alpha\xi }{2K})}-\alpha Z(\ii (z+l)+K+\ii K')\xi ,
		\end{split}
	\end{equation}
	where the expressions of functions $2\ii \lambda\beta_1$, $2\ii \lambda \beta_2$, $u^2(\xi)-\beta_1$ and $u^2(\xi)-\beta_2$ are shown in Lemma \ref{lemma:lambda-2 y beta1-2 }.
\end{lemma}

\section{\appendixname. The conformal mapping between $\lambda(z)$ and $z(\lambda)$}\label{appendix:commonality map}

We first prove the conformal mapping ${\tau}_1,{\tau}_2$ in Lemma \ref{lemma:map}, which have the same formulas as \eqref{eq:lambda} and \eqref{eq:z-l-lambda}.

\begin{lemma}\label{lemma:map}
	The functions 
	\begin{equation}\label{eq:tau12}
		{\tau}_1(z)=\sn(z)\cd(z),
		\qquad \text{and}\qquad {\tau}_2(z)=\dn(z)\tn(z),
	\end{equation}
	map the  rectangle $[-\frac{K}{2},\frac{K}{2}] \times [-\ii K',\ii K']$ onto the complex plane with two different cuts.
\end{lemma}

\begin{proof}
	By functions ${\tau}_1(z)$ and ${\tau}_2(z)$, we get 
	\begin{equation}\label{eq:z-tau-int}
		z({\tau}_1)=\int_{0}^{{\tau}_1}\frac{\dd s}{\sqrt{(s^2-\hat{\tau}_1^2)(s^2-\hat{\tau}^{2}_2)}}, \qquad 
		z({\tau}_2)=\int_{0}^{{\tau}_2}\frac{\dd s}{\sqrt{(s^2-\hat{\tau}_3^2)(s^2-\hat{\tau}^{*2}_3)}},
	\end{equation}
	where $\hat{\tau}_1=\frac{1}{1+k'},\hat{\tau}_2=\frac{1}{1-k'}, \hat{\tau}_3=k+\ii k'$. By the Christoffel-Schwarz integral formula, we know that $z(\tau_1)$ is a conformal mapping, which maps the upper half plane onto a rectangle $[-\frac{K}{2},\frac{K}{2}]\times [0,\ii K']$ (see Figure \ref{fig:map1-z-tau1}). Furthermore, we can extend the map $z({\tau}_1)$ from the whole complex plane with cuts on the real line onto the rectangle $[-\frac{K}{2},\frac{K}{2}]\times[-\ii K',\ii K']$.
	\begin{figure}[ht]
		\centering
		\includegraphics[width=0.9\linewidth]{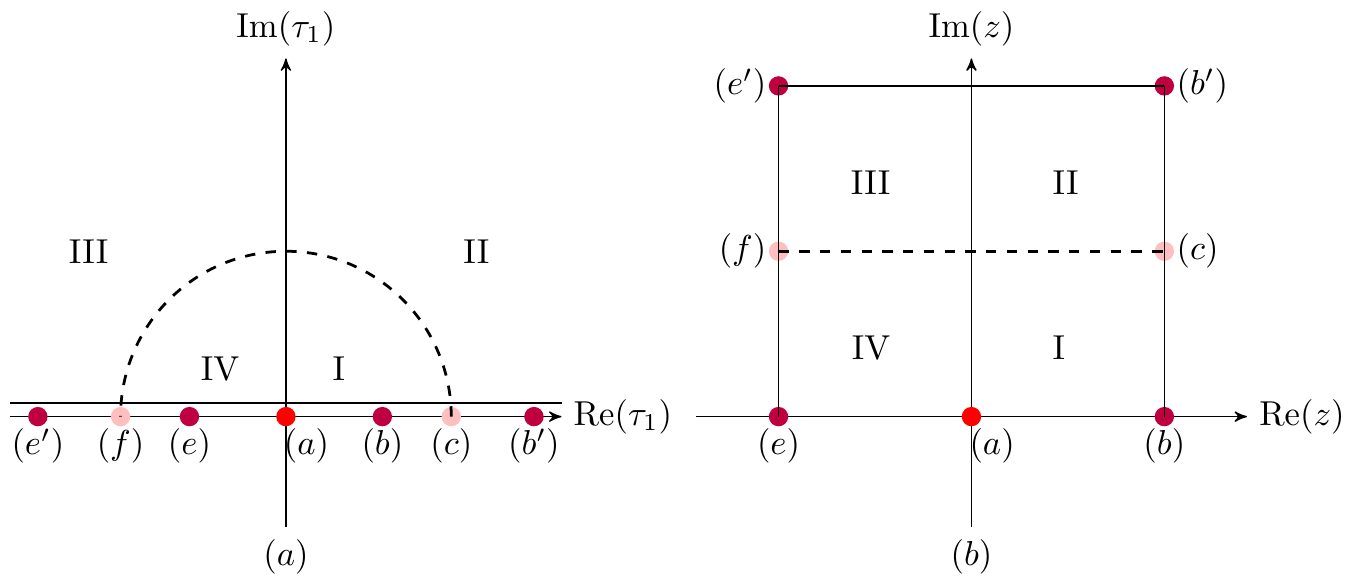}
		\caption{(a):$\{ {\tau}_1(z)|z\in [-\frac{K}{2},\frac{K}{2}]\times [0,\ii K']\}$; (b):$\{z|z\in [-\frac{K}{2},\frac{K}{2}]\times [0,\ii K']\}$. The same symbols represent the corresponding points in different planes, such as the point $(a)$ in ${\tau}_1$-plane is mapping into the point $(a)$ in $z$-plane by function${\tau}_1(0)=0$.}
		\label{fig:map1-z-tau1}
	\end{figure}
	
	Then we analyze the conformal map $z({\tau}_2)$. By equations ${\tau}_1(z)$ and ${\tau}_2(z)$ in \eqref{eq:tau12}, we get  
	\begin{equation}\begin{split}
			k\left(k{\tau}_1+\frac{1}{k{\tau}_1}\right)
			&=\frac{\cn^2(z)+\dn^2(z)-\cn^2(z)\dn^2(z)}{\dn(z)\cn(z)\sn(z)},\\
			{\tau}_2+\frac{1}{{\tau}_2}
			&=\frac{\cn^2(z)+\dn^2(z)-\cn^2(z)\dn^2(z)}{\dn(z)\cn(z)\sn(z)}.
	\end{split}\end{equation}
	Comparing the right side of the above equation, we set 
	\begin{equation}\label{eq:Zhu}
		s=\frac{k}{2}\left(k{\tau}_1+ \frac{1}{k{\tau}_1}\right),\qquad 
		s=\frac{1}{2}\left({\tau}_2+ \frac{1}{{\tau}_2}\right).
	\end{equation}
	Based on the Zhukovskii function \cite[p.77]{Remmert-1991-theory}, we consider the first equation of \eqref{eq:Zhu}. For the convenience of analyzing, we could set the upper half plane of ${\tau}_1$-plane as ${\tau}_1=r_1(\cos(\theta_1)+\ii \sin(\theta_1))$ with $r_1\in[0,+\infty)$ and $\theta_1\in[0,\pi]$ in Figure \ref{fig:map2} (a). Thus, we get $\Re(s)=\frac{k}{2}\left( kr_1+\frac{1}{kr_1} \right)\cos(\theta_1)$ and $\Im(s)=\frac{k}{2}\left(k r_1-\frac{1}{kr_1} \right)\sin(\theta_1)$. When $r_1\in (0,\frac{1}{k})$, the first equation of \eqref{eq:Zhu} maps the semicircle in the upper half ${\tau}_1$-plane with radius $r_1$ into a half ellipse in the lower half $s$-plane with the major axis $\frac{k}{2}\left( kr_1+\frac{1}{kr_1} \right)$ and minor axis $\frac{k}{2}\left( \frac{1}{kr_1}-kr_1\right)$ (See the orange curve in Figure \ref{fig:map2} (a) and Figure \ref{fig:map2} (b)). As the green curve is shown in Figure \ref{fig:map2} (a) and Figure \ref{fig:map2} (b), when $r_1\in (\frac{1}{k},+\infty)$, it maps the semicircle in the upper half ${\tau}_1$-plane with radius $r_1$ into a half ellipse in the upper half $s$-plane with the major axis $\frac{k}{2}\left( kr_1+\frac{1}{kr_1} \right)$ and minor axis $\frac{k}{2}\left( kr_1-\frac{1}{kr_1}\right)$. In particular, the semicircle with a radius $\frac{1}{k}$ is mapped into the line $[-k,k]$. Furthermore, the first equation of \eqref{eq:Zhu} maps the interval $[0,\frac{1}{k}]$ in ${\tau}_1$-plane into the ray $[k,+\infty)$ in $s$-plane and maps the ray $[\frac{1}{k},+\infty)$ into the ray $[k,+\infty)$. So, we get a conformal map between the upper half plane of the ${\tau}_1$-plane and the $s$-plane with cuts $(-\infty,-k)\cup(k,+\infty)$ (See Figure \ref{fig:map2} (a) and Figure \ref{fig:map2} (b)).
	
	Similarly, we consider the second equation in \eqref{eq:Zhu}. We obtain that the upper half plane of the $s$-plane is mapped onto the exterior of the unit circle in ${\tau}_2$-plane, and the lower half plane is mapped onto the interior of the unit circle (See Figure \ref{fig:map2} (b) and Figure \ref{fig:map2} (c)). And the cuts in the real axis of the $s$-plane can map onto the whole real axis, $(e)-(f)$ and $(c)-(b)$ in the ${\tau}_2$-plane. Thus, we establish the conformal map between the $s$-plane and the upper half the ${\tau}_2$-plane. Then, the ${\tau}_2$-plane can be related to the ${\tau}_1$-plane. By the above two maps, we know that there exists a conformal map from the ${\tau}_2$-plane onto the ${\tau}_1$-plane, with the cut from the real axis and two curves $(f)-(e)$ and $(c)-(b)$ onto the whole real axis, successfully. 
	
	In summary, we find the functions  ${\tau}_1(z)$ and ${\tau}_2(z)$ map $[-\frac{K}{2},\frac{K}{2}]\times[-\ii K',\ii K']$ onto the whole complex plane.
\end{proof}

\begin{figure}[ht]
	\centering
	\includegraphics[width=1\linewidth]{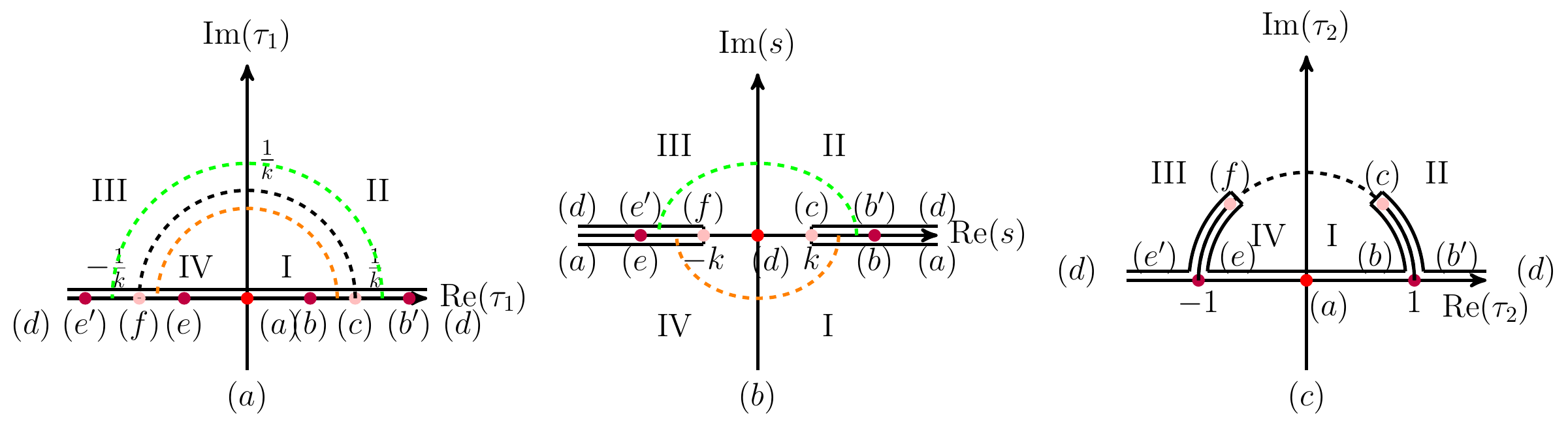}
	\caption{(a): $\{{\tau}_1|\Im({\tau}_1)>0\}$;  (b): $\{s|s=\frac{k}{2}\left(k{\tau}_1+\frac{1}{k{\tau}_1}\right),\Im({\tau}_1)>0\}$;  (c): $\{{\tau}_2|\Im({\tau}_2)>0\}$. The Figure (b) could also be seen as $\{s|s=\frac{1}{2}\left({\tau}_2+\frac{1}{{\tau}_2}\right),\Im({\tau}_2)>0\}$.
	}
	\label{fig:map2}
\end{figure}

	\begin{remark}\label{remark:comformal-z-lambda}
		By Lemma \ref{lemma:map}, we obtain that the function ${\tau}_1(z)$ maps the region $\left[-\frac{K}{2},\frac{K}{2}\right]\times\left[-\ii K',\ii K'\right]$ in the $z$-plane onto the whole ${\tau}_1$-plane. Combining ${\tau}_1(z)$ in \eqref{eq:tau12} with $\lambda(z)$ in \eqref{eq:lambda-b}, we get that the function $\lambda(z)$ maps the region $(z-l)\in \left[-K'+l,K'+l\right]\times\left[-\frac{\ii K}{2},\frac{\ii K}{2}\right]$ onto the whole $\lambda$-plane with the cuts $(b)-(b')$ and $(e)-(e')$ in Figure \ref{fig:map1} (a) and Figure \ref{fig:map1} (c). Similarly, by the conformal map ${\tau}_2(z)$ studied in Lemma \ref{lemma:map} and $\lambda(z)$ in \eqref{eq:lambda-a}, we obtain that $\lambda(z)$ is also a conformal map, which maps $\left[-K'+l,K'+l\right]\times\left[-\frac{\ii K}{2},\frac{\ii K}{2}\right]$ onto the whole $\lambda-$plane with cuts $(f)-(c)$ and $(h)-(g)$, shown in Figure \ref{fig:map1} (a) and Figure \ref{fig:map1} (b).
	\end{remark}

\begin{figure}[ht]
	\centering
	\includegraphics[width=1\linewidth]{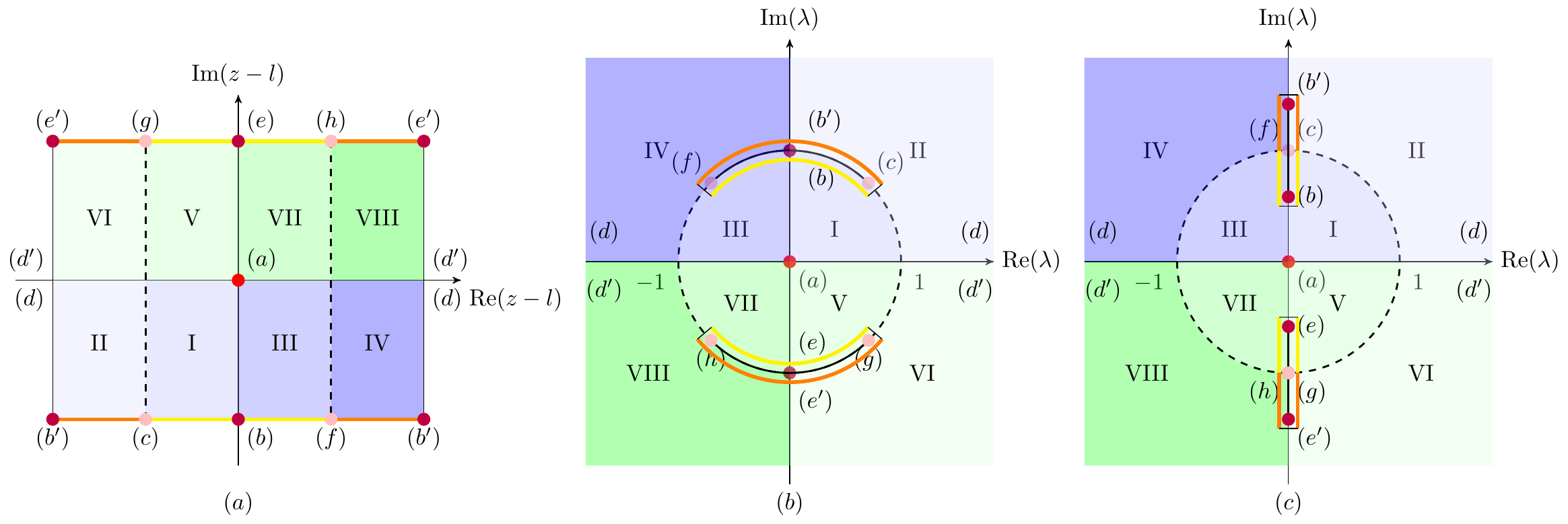}
	\caption{ (a): $\{z-l|z\in S,l=0 \text{ or }\frac{K'}{2}\}$; (b): $\{\lambda(z)|z\in S,l=0\}$, where the function $\lambda(z)$ is defined in \eqref{eq:lambda-a}; (c): $\{\lambda(z)|z\in S,l=\frac{K'}{2}\}$, where the function $\lambda(z)$ is defined in \eqref{eq:lambda-b}. The symbols (such as points (a), (b), (c), and so on) in different planes represent the corresponding points by the conformal map.
	}
	\label{fig:map1}
\end{figure}

\subsection*{A lemma of squared eigenfunctions}
\begin{lemma}\label{lemma:W2=W3}
	There are two linearly independent squared eigenfunctions with parameters $\lambda=\pm \lambda_1,\pm \lambda_1^*$ with the period: $\frac{4K}{\alpha}$.
\end{lemma}

\begin{proof}
	Combining \eqref{eq:Phi} with \eqref{eq:w-solution}, we set four functions $W_{2}(\xi)$, $W_{3}(\xi)$, $W_{4}(\xi)$, $W_{5}(\xi)$ with four different values $\lambda_1,-\lambda_1, \lambda_1^*,-\lambda_1^*$, respectively. When $\lambda$ is equal to the above four values, we get $\beta_1=\beta_2$, $\theta_1=\theta_2$ in \eqref{eq:beta1-beta2-theta-1-theta-2} and $\Omega_1=0$ in \eqref{eq:w-solution}, since $y=0$ in \eqref{eq:y lambda }. Therefore, 
	\begin{equation}
		\begin{split}
			W_2(\xi)=&\phi_1^2(\xi,t;\lambda_1)-\psi_1^2(\xi,t;\lambda_1)
			=(u^2(\xi)-\beta_1)(\exp(2\theta_1)-\exp(-2\theta_1)),\\
			W_3(\xi)=&\phi_1^2(\xi,t;-\lambda_1)-\psi_1^2(\xi,t;-\lambda_1)
			=(u^2(\xi)-\beta_1)(\exp(-2\theta_1)-\exp(2\theta_1)),\\
			W_4(\xi)=&\phi_1^2(\xi,t;\lambda_1^*)-\psi_1^2(\xi,t;\lambda_1^*)
			=(u^2(\xi)-\beta_1^*)(\exp(2\theta_1^*)-\exp(-2\theta_1^*)),\\
			W_5(\xi)=&\phi_1^2(\xi,t;-\lambda_1^*)-\psi_1^2(\xi,t;-\lambda_1^*)=(u^2(\xi)-\beta_1^*)(\exp(-2\theta_1^*)-\exp(2\theta_1^*)).
		\end{split}
	\end{equation}
	By the above four functions, we get  $W_2(\xi)=-W_3(\xi)=W_4^*(\xi)=-W_5^*(\xi)\in \mathbb{C}$. Thus, functions $W_2(\xi)$ and $W_3(\xi)$ are linearly dependent, and functions $W_4(\xi)$ and $W_5(\xi)$ are also linearly dependent. Since functions $u^2(\xi)-\beta_1$ and $u^2(\xi)-\beta_1^*$ have different poles in the $\xi$-complex plane, we get that functions $W_2(\xi)$ and $W_4(\xi)$ are linearly independent with different poles. Furthermore, by the exact expression of the function 
	\begin{equation}
		W_2(\xi) = \frac{\alpha^2 \vartheta_2^2\vartheta_4^2}{\vartheta_3^2\vartheta_4^2\left(\frac{\alpha \xi}{2K}\right)}\left( \frac{\vartheta_1^2\left(\frac{K+\ii K'+2\alpha \xi}{4K}\right)}{\vartheta_4^2\left(\frac{K+\ii K'}{4K}\right)}-\frac{\vartheta_3^2\left(\frac{K+\ii K'+2\alpha \xi}{4K}\right)}{\vartheta_2^2\left(\frac{K+\ii K'}{4K}\right)} \right)\exp\left( \ii \frac{\alpha \xi}{2K} \pi \right),
	\end{equation}
	it is easy to verify that  $W_2(\xi+\frac{2K}{\alpha})=-W_2(\xi)$ and $W_2(\xi+\frac{4K}{\alpha})=W_2(\xi)$, i.e., $\frac{4K}{\alpha}$ is the period of function $W_2(\xi)$.
\end{proof}

\section{\appendixname. The integrability structure of the mKdV equation}\label{appendix:Darboux transformation}

\setcounter{equation}{0}
\renewcommand\theequation{\Alph{section}.\arabic{equation}}
\setcounter{theorem}{0}
\renewcommand\thetheorem{\Alph{section}.\arabic{theorem}}
\renewcommand\thelemma{\Alph{section}.\arabic{lemma}}

In this section, we mainly introduce the integrability structure of the mKdV equation: the mKdV hierarchy, the Hamiltonian conserved quantity, the Darboux matrix and the Lax pair of the higher-order mKdV hierarchy.

The mKdV hierarchy can be derived by the AKNS scheme\cite{AblowitzKNS-74}. For the $x$-part of Lax pair \eqref{eq:Lax pair}, we could set the $t$-part as
	\begin{equation}\label{eq:anks-scheme}
		\Phi_t(x,t;\lambda)=\begin{bmatrix}
			A & B \\ C & -A
		\end{bmatrix}\Phi(x,t;\lambda),
	\end{equation}
where $A\equiv A(x,t;\lambda),B\equiv B(x,t;\lambda),C\equiv C(x,t;\lambda)$. By the zero curvature equation or the compatibility condition $\Phi_{xt}=\Phi_{tx}$, we obtain equations $-A_x+u(C+B)=0$, $u_t-B_x-2\ii \lambda B-2u A=0$, and $-u_t-C_x+2\ii \lambda C-2u A=0$, which implies 
\begin{equation}
	A=\partial_x^{-1}\begin{bmatrix}
		-u & u
	\end{bmatrix}\begin{bmatrix}
		-B \\ C
	\end{bmatrix}+A_0, \qquad A_0\equiv A_0(\lambda).
\end{equation}
To keep the compatibility of the mKdV hierarchy, we suppose $A_0=-\frac{\ii}{2}(2  \lambda)^{2n+1}$, $B=\sum_{i=1}^{2n+1}b_j(x,t)\lambda^{2n+1-j}$ and $C=\sum_{i=1}^{2n+1}c_j(x,t)\lambda^{2n+1-j}$. Comparing the coefficients of the parameter $\lambda$, we obtain the following equations:
	\begin{equation}\nonumber
		\begin{bmatrix}
			u\\-u
		\end{bmatrix}_t=L\begin{bmatrix}
			-b_{2n+1}\\c_{2n+1}
		\end{bmatrix}, \quad
		2\ii\begin{bmatrix}
			-b_{j+1}\\c_{j+1}
		\end{bmatrix}=L\begin{bmatrix}
			-b_{j}\\c_{j}
		\end{bmatrix}, \quad
		\begin{bmatrix}
			b_{1}\\c_{1}
		\end{bmatrix}=2^{2n}\begin{bmatrix}
			u\\-u
		\end{bmatrix}, \quad
		L:=-\sigma_3\partial_x+2\begin{bmatrix}
			u\\ u
		\end{bmatrix}\partial_x^{-1}\begin{bmatrix}
			-u & u
		\end{bmatrix}.
	\end{equation}
Thus the mKdV hierarchy can be defined as
	\begin{equation}\label{eq:L-functional}
		\begin{bmatrix}
			u\\-u
		\end{bmatrix}_{t_n}=(-1)^{n+1}L^{2n+1}\begin{bmatrix}
			u\\u
		\end{bmatrix},
	\end{equation} 
which could be expressed as follows:
	\begin{equation}
		u_{t_n}=\partial_x \mathcal{F}^{n}u=\partial_x \mathcal{H}_n'(u),\qquad n=0,1,2,\cdots,
	\end{equation} 
with the recursion formula \cite{Olver-1993-applications} $\mathcal{H}_{n}'=\mathcal{F}\mathcal{H}_{n-1}'$.The prime $'$ of $\mathcal{H}_{n}'$ is defined as the gradient of functional $\mathcal{H}_{n}$ for the scalar product. Based on the functional matrix $L$ in \eqref{eq:L-functional}, the recursion operator $\mathcal{F}$
is defined as
	\begin{equation}
		\mathcal{F}:=-(\partial_x^2+4u^2-4u\partial_x^{-1}u_x), \qquad \partial_x^{-1}u=\frac{1}{2}\left(\int_{-PT}^{x}u(y)\dd y-\int_{x}^{PT}u(y)\dd y\right),
	\end{equation}
and the Hamiltonian functional could be expressed as 
	\begin{equation}
	\mathcal{H}_n=\int_{0}^{1}\left(\mathcal{F}^n(\rho u),u\right)\dd \rho= \int_{-PT}^{PT}\int_{0}^{1}\mathcal{F}^n(\rho u)\dd \rho \dd x,\qquad n=0,1,2,\cdots.
	\end{equation}
	Letting $n=0,1,2$, we obtain the first three Hamiltonian functionals in \eqref{eq:H0}. The corresponding equations are expressed as follows:
	\begin{equation}\label{eq:ODE-t}
		\begin{split}
			u_{t_0} &=\partial_x \mathcal{H}_0'
			=\partial_{x}u,
			\qquad \qquad
			u_{t_1} =\partial_x \mathcal{H}_1'=-\partial_{x}^3u-6u^2\partial_{x}u,\\
			u_{t_2} &=\partial_x \mathcal{H}_2'=\partial_{x}^5u+10u_{x}^3+40uu_{x}\partial_{x}^2u+10 u^2\partial_{x}^3u+30 u^4 \partial_{x} u.
		\end{split} 
	\end{equation}
When $n=1$, the \eqref{eq:mKdV} equation is a Hamiltonian system of the form $u_t=\partial_x\mathcal{H}_1'(u)$, which could be expressed in the recursion formula readily as $u_t=\partial_x\mathcal{F}\mathcal{H}_0'(u)$.  

If the derivative of the Hamiltonian functional $\mathcal{H}_i$, $i=0,1,2,\cdots$ with respect to time $t$ is zero, i.e., $\frac{\dd \mathcal{H}_i}{\dd t}=0$, the Hamiltonian functional $\mathcal{H}_i$ is the Hamiltonian conserved quantity. The Definition of the Poisson bracket \cite{Lax-1975-periodic} for the class of $\mathbb{C}^{\infty}([-PT,PT])$ functionals $\mathcal{H}_i,\mathcal{H}_j$ of the smooth periodic functions $u$ with $2PT$ period is $\left[\mathcal{H}_i,\mathcal{H}_j\right]=\left(\mathcal{H}'_i,\partial_x\mathcal{H}'_j\right)$, where the $(\cdot,\cdot)$ denotes the $L^2\left([-PT,PT]\right)$ scalar product. Combining definitions of the gradient and the Poisson bracket with equation \eqref{eq:ODE-t}, we get
	\begin{equation}\label{eq:gradient}
		\frac{\dd \mathcal{H}_i}{\dd t}=\left( \mathcal{H}'_i, u_t\right)=\left( \mathcal{H}'_i, \partial_x \mathcal{H}'_1 \right)=\left[\mathcal{H}_i,\mathcal{H}_1\right].
	\end{equation}
	It follows that the Hamiltonian functionals $\mathcal{H}_i$ are conserved if and only if $\left[\mathcal{H}_i,\mathcal{H}_1\right]=0$, $i=0,1,2,\cdots$.

Then, we introduce the Darboux transformation of the mKdV equation \cite{Cieslinski-2009-algebraic}. Under the $(\xi,t)$ moving coordinate frame \eqref{eq:traveling-wave-transform}, the Darboux matrix $\mathbf{T}_i(\lambda;\xi,t)$, $i=1,2$ could convert the old Lax pair into a new Lax pair
\begin{equation}\nonumber
	\Phi_{ \xi}^{[i]}(\xi,t;\lambda)=\mathbf{U}^{[i]}(\lambda;u^{[i]})\Phi^{[i]}(\xi,t;\lambda), \qquad \Phi_{ t}^{[i]}(\xi,t;\lambda)=\hat{\mathbf{V}}^{[i]}(\lambda;u^{[i]})\Phi^{[i]}(\xi,t;\lambda) ,
\end{equation} 
where $\Phi^{[i]}(\xi,t;\lambda):=\mathbf{T}_i(\lambda;\xi,t)\Phi(\xi,t;\lambda), \, \mathbf{U}^{[i]}(\lambda;u^{[i]})\equiv\mathbf{U}(\lambda;u^{[i]}),\, \hat{\mathbf{V}}^{[i]}(\lambda;u^{[i]})\equiv\hat{\mathbf{V}}(\lambda;u^{[i]})$, $i=1,2$.
Based on the symmetric properties of matrices $\mathbf{U}(\lambda;u)$ and $\hat{\mathbf{V}}(\lambda;u)$ in equation \eqref{eq:sym-1}, we could obtain that the Darboux matrix $\mathbf{T}_i(\lambda;x,t)$ satisfies
	\begin{equation}\label{eq:T-sym}
		\mathbf{T}_i^{-1}(\lambda;\xi,t)=\mathbf{T}_i^{\dagger}(\lambda^*;\xi,t), \qquad \mathbf{T}_i^{-1}(\lambda;\xi,t)=\mathbf{T}_i^{\top}(-\lambda;\xi,t).
	\end{equation}
\begin{lemma}\label{lemma:T-T-1}
The Darboux matrix
	\begin{equation}\label{eq:dt-1}
		\mathbf{T}_1(\lambda;\xi,t)=\mathbb{I}-\frac{\lambda_1-\lambda_1^*}{\lambda -\lambda_1^*}\mathbf{P}_1(\xi,t),\qquad \mathbf{P}_1(\xi,t)=\frac{{\Phi}_1 {\Phi}^{\dagger}_1}{{\Phi}^{\dagger}_1{\Phi}_1 },\qquad \Phi_1=\Phi(\xi,t;\lambda_1)\mathbf{c}\equiv[\phi_1,\psi_1]^{\top},
	\end{equation}
	keeps the first symmetric relation of \eqref{eq:sym-1}, and the corresponding B\"acklund transformation between old and new potential functions is given in \eqref{eq:bt-1}.
\end{lemma}
\begin{proof}
	Suppose the Lax pair has the following analytic matrix solutions
	\begin{equation}\label{eq:Phi-m}
		\Phi(\xi,t;\lambda)=\mathbf{m}(\lambda;\xi,t)\ee^{-\ii\lambda[\xi+2(2\lambda^2+s_2)t]\sigma_3}\mathbf{m}^{-1}(\lambda;0,0)
	\end{equation}
	where the meromorphic function matrix $\mathbf{m}(\lambda;\xi,t)$ can be expanded at the neighborhood of $\infty$:
	\begin{equation}\label{eq:m}
		\mathbf{m}(\lambda;\xi,t)=\mathbb{I}+\mathbf{m}_1(\xi,t)\lambda^{-1}+\mathcal{O}(\lambda^{-2}).
	\end{equation}
	Define 
	\begin{equation}\label{eq:A}
		\mathbf{A}(\lambda;\xi,t)\equiv\ii\,\mathbf{m}(\lambda;\xi,t)\sigma_3\mathbf{m}^{-1}(\lambda;\xi,t)=\ii\,\sigma_3+\sum_{i=1}^{\infty}\mathbf{A}_i(\xi,t)\lambda^{-i}.
	\end{equation}
	We can verify
	\begin{equation}\label{eq:recursive}
		\frac{\partial}{\partial \xi}\mathbf{A}(\lambda;\xi,t)=[\mathbf{U}(\lambda;u),\mathbf{A}(\lambda;\xi,t)],\qquad \mathbf{A}^2(\lambda;\xi,t)=-\mathbb{I}.
	\end{equation}
	Then $\mathbf{A}_i(\xi,t)$ can be determined recursively by \eqref{eq:recursive}. The first three of them are 
	\begin{equation}\label{eq:A-1-2-3}
		\begin{split}
			\mathbf{A}_1(\xi,t)=&-\mathbf{Q}=-\ii[\sigma_3,\mathbf{m}_1(\xi,t)], \\
			\mathbf{A}_2(\xi,t)=&-\frac{\ii}{2}\sigma_3(\mathbf{Q}_\xi-\mathbf{Q}^2), \\
			\mathbf{A}_3(\xi,t)=&\frac{1}{4}(\mathbf{Q}_{\xi\xi}-2\mathbf{Q}^3-\mathbf{Q}_\xi\mathbf{Q}+\mathbf{Q}\mathbf{Q}_\xi).
		\end{split}
	\end{equation}
	It follows that
	\begin{equation}\nonumber
		\begin{split}
			\Phi_{\xi}(\xi,t;\lambda)\Phi^{-1}(\xi,t;\lambda)
			=&\mathbf{m}_{\xi}(\lambda;\xi,t)\mathbf{m}^{-1}(\lambda;\xi,t)-\ii\lambda \mathbf{m}(\lambda;\xi,t)\sigma_3\mathbf{m}^{-1}(\lambda;\xi,t)
			=\mathbf{U}(\lambda;u),\\
			\Phi_{t}(\xi,t;\lambda)\Phi^{-1}(\xi,t;\lambda)
			=&\mathbf{m}_{t}(\lambda;\xi,t)\mathbf{m}^{-1}(\lambda;\xi,t)-\ii(4\lambda^3+2s_2\lambda) \mathbf{m}(\lambda;\xi,t)\sigma_3\mathbf{m}^{-1}(\lambda;\xi,t)
			=\hat{\mathbf{V}}(\lambda;u).
		\end{split}
	\end{equation}
	Applying the Darboux transformation to the wave function $\Phi(\xi,t;\lambda)$,  we obtain a new wave function $\Phi^{[1]}(\xi,t;\lambda)=\mathbf{T}_1(\lambda;\xi,t)\Phi(\xi,t;\lambda)\mathbf{T}_1^{-1}(\lambda;0,0)$ that is analytic in the whole complex plane $\mathbb{C}$. For the new wave function $\Phi^{[1]}(\xi,t;\lambda)$, the function $\mathbf{m}(\lambda;\xi,t)$ will be replaced by $\mathbf{m}^{[1]}(\lambda;\xi,t)=\mathbf{T}_1(\lambda;\xi,t)\mathbf{m}(\lambda;\xi,t)$ that also can be expanded in the neighborhood of $\infty$:
	\begin{equation}
		\mathbf{m}^{[1]}(\lambda;\xi,t)=\mathbb{I}+\mathbf{m}_1^{[1]}(\xi,t)\lambda^{-1}+\mathcal{O}(\lambda^{-2}),
		\qquad
		\mathbf{m}_1^{[1]}(\xi,t)=\mathbf{m}_1(\xi,t)-(\lambda_1-\lambda_1^*)\mathbf{P}_1(\xi,t).
	\end{equation}
	Furthermore, we have
	\begin{equation}
			\mathbf{Q}^{[1]}=\mathbf{Q}-\ii(\lambda_1-\lambda_1^*)[\sigma_3,\mathbf{P}_1(\xi,t)], \quad
			\mathbf{U}^{[1]}(\lambda;u^{[1]})=\mathbf{U}(\mathbf{Q}\to\mathbf{Q}^{[1]}),\quad \hat{\mathbf{V}}^{[1]}(\lambda;u^{[1]})=\hat{\mathbf{V}}(\mathbf{Q}\to\mathbf{Q}^{[1]}).
	\end{equation}
	As for the symmetric property, through $\Phi(\xi,t;\lambda)\Phi^{\dag}(\xi,t;\lambda^*)=\mathbb{I}$ and  $\mathbf{T}_1(\lambda;\xi,t)\mathbf{T}_1^{\dag}(\lambda^*;\xi,t)=\mathbb{I}$, we obtain $\Phi^{[1]}(\xi,t;\lambda)\Phi^{[1]\dag}(\xi,t;\lambda^*)=\mathbb{I}$, which implies $\mathbf{U}^{[1]\dag}(\lambda^*;u^{[1]})=-\mathbf{U}^{[1]}(\lambda;u^{[1]})$ and $\hat{\mathbf{V}}^{[1]\dag}(\lambda^*;u^{[1]})=-\hat{\mathbf{V}}^{[1]}(\lambda;u^{[1]})$.
\end{proof}

\newenvironment{aproof}{\emph{Proof of Theorem \ref{theorem:construct}.}}{\hfill$\Box$\medskip}
\begin{aproof}
	If the Darboux transformation in Lemma \ref{lemma:T-T-1} also satisfies the second equation of \eqref{eq:T-sym},
	i.e., $\lambda_1^*+\lambda_1=0$ and $\mathbf{P}_1^{\top}(\xi,t)=\mathbf{P}_1(\xi,t)$, then the Darboux transformation will keep the second symmetric property \eqref{eq:sym-1} of matrix $\mathbf{U}(\lambda;u)$. The corresponding B\"acklund transformation could be expressed as \eqref{eq:bt-1}.
	
	As for the case $\lambda_1+\lambda_1^*\neq 0$,  we need to consider the two-fold Darboux transformation
	\begin{equation}
		\mathbf{T}_2(\lambda;\xi,t)
		=\mathbb{I}-
		\begin{bmatrix}
			\Phi_1 & \Phi_1^*
		\end{bmatrix}\mathbf{M}^{-1} \mathbf{D}^{-1}
		\begin{bmatrix}
			\Phi_1^{\dagger} \\[5pt] \Phi_1^{\top}
		\end{bmatrix}, 	\qquad \mathbf{D}={\rm diag}\left(\lambda-\lambda_1^*,\lambda+\lambda_1\right),
	\end{equation}
	which also satisfies symmetric properties: $\mathbf{T}_2^{-1}(\lambda;\xi,t)=\mathbf{T}_2^{\top}(-\lambda;\xi,t)$,
	and the corresponding B\"acklund transformation is given by \eqref{eq:bt-2}.
\end{aproof}

Using a similar method as Lemma \ref{lemma:T-T-1}, we obtain the polynomial form of the third members of the mKdV hierarchy
\begin{equation}\label{eq:Lax}
	\begin{split}
		\mathbf{V}_2(\lambda;u)=&4\lambda^2 \mathbf{V}(\lambda,u)-2\ii \lambda \sigma_3\left(\mathbf{Q}_{xxx}-6\mathbf{Q}^2\mathbf{Q}_{x}+\mathbf{Q}^4-2\mathbf{Q}\mathbf{Q}_{xx}+\mathbf{Q}_{x}^2\right)\\
		&-10\mathbf{Q}^2\mathbf{Q}_{xx}-10 \mathbf{Q}_{x}^2\mathbf{Q}+6\mathbf{Q}^5+\mathbf{Q}_{xxxx}, 
	\end{split}
\end{equation}
which admits the evolution part of Lax pair: $\Phi_{t_2}=\mathbf{V}_2(\lambda;u)\Phi$, which also can be derived directly by the above-mentioned AKNS scheme \eqref{eq:anks-scheme}.

\bibliographystyle{siam}
\bibliography{references}

\begin{thebibliography}{10}

\bibitem{AblowitzKNS-74}
{\sc M.~J. Ablowitz, D.~J. Kaup, A.~C. Newell, and H.~Segur}, {\em The inverse
  scattering transform-{F}ourier analysis for nonlinear problems}, Stud. Appl.
  Math., 53 (1974), pp.~249--315.

\bibitem{AblowitzS-81}
{\sc M.~J. Ablowitz and H.~Segur}, {\em Solitons and the Inverse Scattering
  Transform}, Society for Industrial and Applied Mathematics (SIAM),
  Philadelphia, Pa., 1981.

\bibitem{Alejo-2013-nonlinear}
{\sc M.~A. Alejo and C.~Mu{\~n}oz}, {\em Nonlinear stability of mkdv
  breathers}, Commun. Math. Phys., 324 (2013), pp.~233--262.

\bibitem{Pava-05}
{\sc J.~Angulo~Pava}, {\em Stability of dnoidal waves to {H}irota-{S}atsuma
  system}, Differ. Integral. Equ, 18 (2005), pp.~611--645.

\bibitem{Pava-07}
\leavevmode\vrule height 2pt depth -1.6pt width 23pt, {\em Nonlinear stability
  of periodic traveling wave solutions to the {S}chr\"{o}dinger and the
  modified {K}orteweg-de {V}ries equations}, J. Differential Equations, 235
  (2007), pp.~1--30.

\bibitem{PavaBS-06}
{\sc J.~Angulo~Pava, J.~L. Bona, and M.~Scialom}, {\em Stability of cnoidal
  waves}, Adv. Differential Equations, 11 (2006), pp.~1321--1374.

\bibitem{PavaN-09}
{\sc J.~Angulo~Pava and F.~M.~A. Natali}, {\em Stability and instability of
  periodic travelling wave solutions for the critical {K}orteweg-de {V}ries and
  nonlinear {S}chr\"{o}dinger equations}, Phys. D, 238 (2009), pp.~603--621.

\bibitem{ArmitageE-06}
{\sc J.~V. Armitage and W.~F. Eberlein}, {\em Elliptic Functions}, Cambridge
  University Press, Cambridge, 2006.

\bibitem{BelokolosBEI-94}
{\sc E.~D. Belokolos, A.~I. Bobenko, V.~Z. Enol’skii, A.~R. Its, and V.~B.
  Matveev}, {\em Algebro-Geometric Approach to Nonlinear Integrable Equations},
  Springer Series Studies in Nonlinear Dynamics, 1994.

\bibitem{Benjamin-72}
{\sc T.~B. Benjamin}, {\em The stability of solitary waves}, Proc. Roy. Soc.
  London Ser. A, 328 (1972), pp.~153--183.

\bibitem{BilmanM-19}
{\sc D.~Bilman and P.~D. Miller}, {\em A robust inverse scattering transform
  for the focusing nonlinear {S}chr\"{o}dinger equation}, Comm. Pure Appl.
  Math., 172 (2019), pp.~1722--1805.

\bibitem{Bona-75}
{\sc J.~L. Bona}, {\em On the stability theory of solitary waves}, Proc. Roy.
  Soc. London Ser. A, 344 (1975), pp.~363--374.

\bibitem{BonaSS-87}
{\sc J.~L. Bona, P.~E. Souganidis, and W.~A. Strauss}, {\em Stability and
  instability of solitary waves of {K}orteweg-de {V}ries type}, Proc. Roy. Soc.
  London Ser. A, 411 (1987), pp.~395--412.

\bibitem{BottmanD-09}
{\sc N.~Bottman and B.~Deconinck}, {\em Kd{V} conidal waves are spectrally
  stable}, Discrete Contin. Dyn. Syst., 25 (2009), pp.~1163--1180.

\bibitem{BottmanDN-11}
{\sc N.~Bottman1, B.~Deconinck, and M.~Nivala}, {\em Elliptic solutions of the
  defocusing {N}{L}{S} equation are stable}, J. Phys. A, 44 (2011),
  pp.~285201--285225.

\bibitem{BronskiJK-11}
{\sc J.~C. Bronski, M.~A. Johnson, and T.~Kapitula}, {\em An index theorem for
  the stability of periodic traveling waves of {K}orteweg-de {V}ries type},
  Proc. Roy. Soc. Edinburgh Sect. A, 141 (2011), pp.~1141--1173.

\bibitem{ByrdF-54}
{\sc P.~F. Byrd and M.~D. Friedman}, {\em Handbook of elliptic integrals for
  engineers and physicists}, Springer-verlag Berlin Heidelberg Gmbh, 1954.

\bibitem{Cheemaa-2020-study}
{\sc N.~Cheemaa, A.~R. Seadawy, T.~G. Sugati, and D.~Baleanu}, {\em Study of
  the dynamical nonlinear modified {K}orteweg--de {V}ries equation arising in
  plasma physics and its analytical wave solutions}, Results Phys., 19 (2020),
  p.~103480.

\bibitem{ChenP-18-NLS}
{\sc J.~Chen and D.~E. Pelinovsky}, {\em Rogue periodic waves of the focusing
  nonlinear {S}chr\"{o}dinger equation}, Proc. A., 474 (2018), pp.~20170814,
  18.

\bibitem{ChenP-18-mKdV}
\leavevmode\vrule height 2pt depth -1.6pt width 23pt, {\em Rogue periodic waves
  of the modified {K}d{V} equation}, Nonlinearity, 31 (2018), pp.~1955--1980.

\bibitem{ChenP-19}
\leavevmode\vrule height 2pt depth -1.6pt width 23pt, {\em Periodic travelling
  waves of the modified {K}d{V} equation and rogue waves on the periodic
  background}, J. Nonlinear Sci., 29 (2019), pp.~2797--2843.

\bibitem{ChenP-21}
\leavevmode\vrule height 2pt depth -1.6pt width 23pt, {\em Rogue waves on the
  background of periodic standing waves in the derivative nonlinear
  {S}chr\"{o}dinger equation}, Phys. Rev. E, 103 (2021), pp.~062206, 25.

\bibitem{ChenPW-19}
{\sc J.~Chen, D.~E. Pelinovsky, and R.~E. White}, {\em Rogue waves on the
  double-periodic background in the focusing nonlinear {S}chr\"{o}dinger
  equation}, Phys. Rev. E, 100 (2019), pp.~052219, 18.

\bibitem{ChenPW-20}
\leavevmode\vrule height 2pt depth -1.6pt width 23pt, {\em Periodic standing
  waves in the focusing nonlinear {S}chr\"{o}dinger equation: rogue waves and
  modulation instability}, Phys. D, 405 (2020), pp.~132378, 13.

\bibitem{Cieslinski-2009-algebraic}
{\sc J.~L. Cie{\'s}li{\'n}ski}, {\em Algebraic construction of the darboux
  matrix revisited}, J. Phys. A-Math. Theor., 42 (2009), p.~404003.

\bibitem{CollianderKSTT-03}
{\sc J.~Colliander, M.~KEEL, G.~STAFFILANI, H.~TAKAOKA, and T.~TAO}, {\em Sharp
  global well-posedness for {K}d{V} and modified {K}d{V} on $\mathbb{R}$ and
  $\mathbb{T}$}, J. Amer. Math. Soc., 16 (2003), pp.~705--749.

\bibitem{CourantH-53}
{\sc R.~Courant and D.~Hilbert}, {\em Methods of mathematical physics. {V}ol.
  {II}: {P}artial differential equations}, Interscience Publishers (a division
  of John Wiley \& Sons), New York-London, 1962.

\bibitem{CuccagnaPV-05}
{\sc S.~Cuccagna, D.~Pelinovsky, and V.~Vougalter}, {\em Spectra of positive
  and negative energies in the linearized {NLS} problem}, Comm. Pure Appl.
  Math., 58 (2005), pp.~1--29.

\bibitem{DeconinckK-10}
{\sc B.~Deconinck and T.~Kapitula}, {\em The orbital stability of the cnoidal
  waves of the {K}orteweg-de {V}ries equation}, Phys. Lett. A, 372 (2010),
  pp.~4018--4022.

\bibitem{DeconinckK-06}
{\sc B.~Deconinck and J.~N. Kutz}, {\em Computing spectra of linear operators
  using the {F}loquet-{F}ourier-{H}ill method}, J. Comput. Phys., 219 (2006),
  pp.~296--321.

\bibitem{Deconinck-10}
{\sc B.~Deconinck and M.~Nivala}, {\em The stability analysis of the periodic
  traveling wave solutions of the m{K}d{V} equation}, Stud. Appl. Math., 126
  (2011), pp.~17--48.

\bibitem{DeconinckS-17}
{\sc B.~Deconinck and B.~L. Segal}, {\em The stability spectrum for elliptic
  solutions to the focusing {NLS} equation}, Phys. D, 346 (2017), pp.~1--19.

\bibitem{DeconinckyU-20}
{\sc B.~Deconinck and J.~Upsal}, {\em The orbital stability of elliptic
  solutions of the focusing nonlinear {S}chr\"{o}dinger equation}, SIAM J.
  Math. Anal., 52 (2020), pp.~1--41.

\bibitem{Dickey-2003-soliton}
{\sc L.~A. Dickey}, {\em Soliton equations and Hamiltonian systems}, vol.~26,
  World scientific, 2003.

\bibitem{FengLT-19}
{\sc B.~Feng, L.~Ling, and D.~A. Takahashi}, {\em Multi-breather and high-order
  rogue waves for the nonlinear {S}chr\"{o}dinger equation on the elliptic
  function background}, Stud. Appl. Math., 144 (2020), pp.~46--101.

\bibitem{Floquet-83}
{\sc G.~Floquet}, {\em Sur les \'{e}quations diff\'{e}rentielles lin\'{e}aires
  \`a coefficients p\'{e}riodiques}, Ann. Sci. \'{E}cole Norm. Sup. (2), 12
  (1883), pp.~47--88.

\bibitem{KapitulaH-07-small}
{\sc T.~Gallay and M.~H\u{a}r\u{a}gu\c{s}}, {\em Stability of small periodic
  waves for the nonlinear {S}chr\"{o}dinger equation}, J. Differential
  Equations, 234 (2007), pp.~544--581.

\bibitem{KapitulaH-07-per}
{\sc T.~Gallay and M.~H\v{a}r\v{a}gus}, {\em Orbital stability of periodic
  waves for the nonlinear {S}chr\"{o}dinger equation}, J. Dynam. Differential
  Equations, 19 (2007), pp.~825--865.

\bibitem{GardnerGKM-67}
{\sc C.~S. Gardner, J.~M. Greene, M.~D. Kruskal, and R.~M. Miura}, {\em
  Korteweg-de {V}ries equation and generalization. {VI}. {M}ethods for exact
  solution}, Comm. Pure Appl. Math., 27 (1974), pp.~97--133.

\bibitem{GrillakisSS-87}
{\sc M.~Grillakis, J.~Shatah, and W.~Strauss}, {\em Stability theory of
  solitary waves in the presence of symmetry. {I}}, J. Funct. Anal., 74 (1987),
  pp.~160--197.

\bibitem{GrillakisSS-90}
\leavevmode\vrule height 2pt depth -1.6pt width 23pt, {\em Stability theory of
  solitary waves in the presence of symmetry. {II}}, J. Funct. Anal., 94
  (1990), pp.~308--348.

\bibitem{HaragusK-08}
{\sc M.~H\v{a}r\v{a}gu\c{s} and T.~Kapitula}, {\em On the spectra of periodic
  waves for infinite-dimensional {H}amiltonian systems}, Phys. D, 237 (2008),
  pp.~2649--2671.

\bibitem{Kamchatnov-97}
{\sc A.~M. Kamchatnov}, {\em New approach to periodic solutions of integrable
  equations and nonlinear theory of modulational instability}, Phys. Rep., 286
  (1997), pp.~199--270.

\bibitem{Kamchatnov-00}
\leavevmode\vrule height 2pt depth -1.6pt width 23pt, {\em Nonlinear periodic
  waves and their modulations}, World Scientific Publishing Co., Inc., River
  Edge, NJ, 2000.

\bibitem{KapitulaD-15}
{\sc T.~Kapitula and B.~Deconinck}, {\em On the spectral and orbital stability
  of spatially periodic stationary solutions of generalized {K}orteweg--de
  {V}ries equations}, Fields Inst. Commun., 75 (2015), pp.~285--322.

\bibitem{KapitulaKS-04}
{\sc T.~Kapitula, P.~G. Kevrekidis, and B.~Sandstede}, {\em Counting
  eigenvalues via the {K}rein signature in infinite-dimensional {H}amiltonian
  systems}, Phys. D, 195 (2004), pp.~263--282.

\bibitem{KapitulaP-13}
{\sc T.~Kapitula and K.~Promislow}, {\em Spectral and dynamical stability of
  nonlinear waves}, Springer, New York, 2013.

\bibitem{KenigPV-93}
{\sc C.~E. Kenig, G.~Ponce, and L.~Vega}, {\em Well-posedness and scattering
  results for the generalized {K}orteweg-de {V}ries equation via the
  contraction principle}, Comm. Pure Appl. Math., 46 (1993), pp.~527--620.

\bibitem{KharchevZ-15}
{\sc S.~Kharchev and A.~Zabrodin}, {\em Theta vocabulary {I}}, J. Geom. Phys.,
  94 (2015), pp.~19--31.

\bibitem{Lax-68}
{\sc P.~D. Lax}, {\em Integrals of nonlinear equations of evolution and
  solitary waves}, Comm. Pure Appl. Math., 21 (1968), pp.~467--490.

\bibitem{Lax-1975-periodic}
{\sc P.~D. Lax}, {\em Periodic solutions of the {K}d{V} equation}, Commun. Pur.
  Appl. Math., 28 (1975), pp.~141--188.

\bibitem{LevitanS-75}
{\sc B.~M. Levitan and I.~S. Sargsjan}, {\em Introduction to spectral theory:
  selfadjoint ordinary differential operators}, American Mathematical Society,
  Providence, R.I., 1975.

\bibitem{LingS-2021}
{\sc L.~Ling and X.~Sun}, {\em Stability of elliptic function solutions for the
  focusing modified kdv equation}, arXiv:2109.05454,  (2021).

\bibitem{MaF-96}
{\sc W.~X. Ma and B.~Fuchssteiner}, {\em The bi-{H}amiltonian structure of the
  perturbation equations of the {K}d{V} hierarchy}, Phys. Lett. A, 213 (1996),
  pp.~49--55.

\bibitem{MiuraGK-68}
{\sc R.~M. Miura, C.~S. Gardner, and M.~D. Kruskal}, {\em Korteweg-de {V}ries
  equation and generalizations. {II}. {E}xistence of conservation laws and
  constants of motion}, J. Math. Phys., 9 (1968), pp.~1204--1209.

\bibitem{Olver-1993-applications}
{\sc P.~J. Olver}, {\em Applications of Lie groups to differential equations},
  vol.~107, Springer Science \& Business Media, 1993.

\bibitem{Pelinovsky-05}
{\sc D.~E. Pelinovsky}, {\em Inertia law for spectral stability of solitary
  waves in coupled nonlinear {S}chr\"{o}dinger equations}, Proc. R. Soc. Lond.
  Ser. A Math. Phys. Eng. Sci., 461 (2005), pp.~783--812.

\bibitem{PSKP-21}
{\sc D.~E. Pelinovsky, A.~V. Slunyaev, A.~V. Kokorina, and E.~N. Pelinovsky},
  {\em Stability and interaction of compactons in the sublinear {K}d{V}
  equation}, Commun. Nonlinear Sci. Numer. Simul., 101 (2021), pp.~105855, 16.

\bibitem{PelinovskyY-05}
{\sc D.~E. Pelinovsky and J.~Yang}, {\em Instabilities of multihump vector
  solitons in coupled nonlinear {S}chr\"{o}dinger equations}, Stud. Appl.
  Math., 115 (2005), pp.~109--137.

\bibitem{Remmert-1991-theory}
{\sc R.~Remmert}, {\em Theory of complex functions}, vol.~122, Springer Science
  \& Business Media, 1991.

\bibitem{Rowlands-74}
{\sc G.~Rowlands}, {\em On the stability of solutions of the non-linear
  {S}chr\"{o}dinger equation}, IMA J. Appl. Math., 13 (1974), p.~367–377.

\bibitem{Schamel-1973-modified}
{\sc H.~Schamel}, {\em A modified korteweg-de vries equation for ion acoustic
  wavess due to resonant electrons}, J. Plasma Phys., 9 (1973), pp.~377--387.

\bibitem{Semenov-2022-orbital}
{\sc A.~Semenov}, {\em Orbital stability of a sum of solitons and breathers of
  the modified {K}orteweg--de {V}ries equation}, Nonlinearity, 35 (2022),
  p.~4211.

\bibitem{Shin-12}
{\sc H.~J. Shin}, {\em Soliton dynamics in phase-modulated lattices}, J. Phys.
  A, 45 (2012), pp.~255206, 18.

\bibitem{TracyCL-84}
{\sc E.~R. Tracy, H.~H. Chen, and Y.~C. Lee}, {\em Study of quasiperiodic
  solutions of the nonlinear {S}chr\"{o}dinger equation and the nonlinear
  modulational instability}, Phys. Rev. Lett., 53 (1984), pp.~218--221.

\bibitem{Weinstein-85}
{\sc M.~I. Weinstein}, {\em Modulational stability of ground states of
  nonlinear {S}chr\"{o}dinger equations}, SIAM J. Math. Anal., 16 (1985),
  pp.~472--491.

\bibitem{Weinstein-86}
\leavevmode\vrule height 2pt depth -1.6pt width 23pt, {\em Lyapunov stability
  of ground states of nonlinear dispersive evolution equations}, Comm. Pure
  Appl. Math., 39 (1986), pp.~51--67.

\bibitem{Yang-10}
{\sc J.~Yang}, {\em Nonlinear waves in integrable and nonintegrable systems},
  Society for Industrial and Applied Mathematics (SIAM), Philadelphia, PA,
  2010.

\bibitem{ZakharovS-72}
{\sc V.~E. Zakharov and A.~B. Shabat}, {\em Exact theory of two-dimensional
  self-focusing and one-dimensional self-modulation of waves in nonlinear
  media}, \v{Z}. \`Eksper. Teoret. Fiz., 61 (1971), pp.~118--134.

\end{thebibliography}

\end{document}